\documentclass[aps,prx,reprint,longbibliography,superscriptaddress,floatfix]{revtex4-1}
\pdfoutput=1

\usepackage[utf8]{inputenc}
\usepackage{url}

\usepackage{multirow}
\usepackage{amsmath}
\usepackage{mathtools}
\usepackage{amsthm}
\usepackage{amssymb}
\usepackage{graphicx}
\usepackage{epsfig}
\usepackage{epstopdf}
\usepackage{float}
\usepackage{tabularx}
\usepackage{booktabs}

\usepackage{color}

\usepackage{hyperref}
\usepackage[usenames,dvipsnames]{xcolor}\hypersetup{
	colorlinks=true,       
	linkcolor=Maroon,          
	citecolor=OliveGreen,        
	filecolor=magenta,      
	urlcolor=Blue           
}
\usepackage[margin=1in]{geometry}
\usepackage{multirow}

\makeatother

\usepackage[english]{babel}
\usepackage[capitalise,nameinlink]{cleveref}

	\newcommand{\normEW}[1]{\|#1\|_{\text{EW}}}
	\newcommand{\normSH}[1]{\|#1\|_{\text{SC}}}

	\global\long\def\onebody{\mathrm{One}}%

	\global\long\def\ii{\mathcal{I}}%

	\global\long\def\t{\mathrm{T}}%

	\global\long\def\ket#1{|#1\rangle }%

	\global\long\def\braket#1#2{\left\langle #1|#2\right\rangle }%

\begin{document}


\title{Quantum computing enhanced computational catalysis}



\newcommand{\Microsoft}[0]{Microsoft Quantum, Redmond, Washington 98052, USA}

\author{Vera von Burg}
\affiliation{Laboratorium f\"ur Physikalische Chemie, ETH Z\"urich, Vladimir-Prelog-Weg 2, 8093
Zürich, Switzerland}
\author{Guang Hao Low}
\affiliation{\Microsoft}
\author{Thomas H\"aner}
\author{Damian S. Steiger}
\affiliation{Microsoft Quantum, 8038 Z\"urich, Switzerland}
\author{\\Markus Reiher}
\email[To whom correspondence should be addressed. Electronic mail:  ]{ markus.reiher@phys.chem.ethz.ch }
\affiliation{Laboratorium f\"ur Physikalische Chemie, ETH Z\"urich, Vladimir-Prelog-Weg 2, 8093 Zürich, Switzerland}
\author{Martin Roetteler}
\author{Matthias Troyer}
\email[To whom correspondence should be addressed. Electronic mail: ]{ mtroyer@microsoft.com}
\affiliation{\Microsoft}



\date{\today}

\begin{abstract}
The quantum computation of electronic energies can break the curse of dimensionality that plagues many-particle
quantum mechanics. It is for this reason that a universal quantum computer has the potential to fundamentally
change computational chemistry and materials science, areas in which strong electron correlations present severe
hurdles for traditional electronic structure methods. Here, we present a state-of-the-art analysis of accurate
energy measurements on a quantum computer for computational catalysis, using improved quantum algorithms with more than an order of magnitude improvement over the best previous algorithms. As a prototypical example of local catalytic
chemical reactivity we consider the case of a ruthenium catalyst that can bind, activate, and transform
carbon dioxide to the high-value chemical methanol. We aim at accurate resource estimates for the quantum computing
steps required for assessing the electronic energy of key intermediates and transition states of its catalytic
cycle. In particular, we present  new quantum algorithms for double-factorized representations of the four-index integrals that can significantly reduce the computational cost over previous algorithms, and we discuss the challenges of increasing active space sizes to accurately deal with dynamical correlations. We address the requirements for future quantum hardware in order to make a universal quantum computer a
successful and reliable tool for quantum computing enhanced computational materials science and chemistry, and identify open questions for further research.
\end{abstract}


\maketitle

\section{Introduction}
Quantum computing~\cite{feynman82,SvoreTroyer2016,cao19,McArdle2020} has the potential to efficiently solve some
computational problems that are exponentially hard to solve on classical computers. Among these problems, one of the most prominent cases is the calculation of quantum electronic energies in molecular systems~\cite{abrams97,aspuru-guzik05,veis10,veis12,veis14}. Due to its many applications in chemistry and materials science, this problem is widely regarded as the ``killer application'' of future quantum computers~\cite{CENarticle17},
a view that was supported by our first rigorous resource estimate study
for the accurate calculation of electronic energies of a challenging chemical problem~\cite{reiher17}.

At the heart of chemistry is predicting the outcome of chemical processes in order to produce chemicals, drugs,
or functional molecular assemblies and materials. A prerequisite for this ability to predict chemical processes is an understanding of the underlying
reaction mechanisms. Quantum mechanics allows one to assign energies to molecular structures so that a comparison
of these energies in a sequence of molecular transformations can be taken as a
measure to rate the viability of such a chemical reaction. Relative energies are a direct means to
predict reaction heats and activation barriers for chemical processes. However, the reliability of such predictions depends crucially
on the accuracy of the underlying energies, of which the electronic energy is often the most important ingredient.
This energetic contribution of the dynamics of the electrons in a molecule can be calculated by solving
the electronic Schr\"odinger equation, typically done in a so-called one-particle basis, the set of molecular orbitals.

The computational complexity of an exact solution of the electronic Schr\"odinger equation on classical computers is prohibitive as the many-electron basis expansion of the quantum state
of interest grows exponentially with the number of molecular orbitals (often called the ``curse of dimensionality''). An exact diagonalization of the electronic Hamiltonian in the full many-electron
representation is therefore hard and limited to small molecules that can be described by comparatively few (on the order of twenty) orbitals.
Once accomplished, it is said that a full configuration interaction (full-CI) solution in this finite
orbital basis has been found. Since typical molecular systems will require on the order of 1000 molecular orbitals
for their reliable description, exact-diagonalization methods need to restrict the orbital space to about twenty orbitals chosen
from the valence orbital space (called complete active space (CAS) CI).
Accordingly, approximate methods have been developed in quantum chemistry for large orbital spaces.
By contrast, quantum computing allows for an encoding of a quantum state in a number of qubits that scales only linearly with the number of molecular orbitals and has therefore the potential to deliver full-CI solutions for large orbital spaces that are inaccessible to traditional computing because of the exponential scaling.

In our previous case study on the feasibility of quantum computing for chemical reactivity~\cite{reiher17}, we chose to investigate electronic structures of a biocatalyst
with still unknown mode of action, i.e., the active site of nitrogenase which is a polynuclear iron-sulfur cluster.
This choice was motivated by the fact that such polynuclear $3d$ transition metal clusters are known to exhibit strongly correlated
electronic structures that are hard to describe reliably with approximate electronic structure methods
on classical computers (electronically excited states of molecules in photophysical and photochemical applications
represent another example of this class of problems). In Ref.~\cite{reiher17}, which relied on
a structure model of the active site of nitrogenase in its resting state, we avoided biasing toward
the electronic ground state of that particular structure by optimizing arbitrary electronic states for
that structure for spin and charge states. These did not coincide with the electronic ground state of the resting state
as emphasized in the supporting information of Ref.~\cite{reiher17}.
Although the resource estimates provided in Ref.~\cite{reiher17} supported the view that quantum computing is a true
competitor of state-of-the-art CAS-CI-type electronic structure methods such as the density matrix renormalization group (DMRG)~\cite{white92,baiardi20}
or full configuration interaction quantum Monte Carlo (FCIQMC)~\cite{booth09} for studies of chemical reaction mechanisms,
we did not consider an actual reaction mechanism.

In this work, we therefore revisit the problem of quantum computing enhanced reaction mechanism elucidation by considering the latest algorithmic advances, but now with a focus on a specific
chemical reaction that is prototypical for homogeneous catalysis and equally well relevant for heterogeneous catalysis.
The example that we chose is the catalytic functionalization of carbon dioxide, i.e., the capture of the small green-house
gas carbon dioxide by a catalyst that activates and eventually transforms it to a useful chemical such as methanol.
Hence, we continue to focus on small-molecule activation catalysis, but emphasize that, despite the obvious
interest into this specific system in the context of carbon capture technologies, our analysis is of general value to
a huge body of chemical reactivity studies.
Moreover, we re-examine some of our initial assumptions that are a moving target in the fast developing field of
quantum computing. These are the gate counts for the quantum algorithm
that performs the energy measurement and assumptions on the error corrected gate times on a future quantum computer, which are heavily
dependent on the underlying technology of its hardware.
Furthermore, we extend our previous work with respect to the steps that need to be carried out by a quantum computer,
specifically regarding the resources of the state preparation step.

Since our previous work~\cite{reiher17}, there has been significant progress in quantum algorithms, but also a better understanding of what it takes to build a scalable quantum computer has been reached.
On the algorithmic side, Hamiltonians represented by a linear combination of unitaries in a so-called black-box query model can now be simulated with optimal cost using techniques called ``qubitization''~\cite{Low2016Qubitization} and ``quantum signal processing''~\cite{Low2016HamSim}.
In addition,  structure in broad families of Hamiltonians can be exploited to reduce the cost of simulation even more.
This includes Hamiltonians with geometrically local interactions~\cite{Haah2018quantum} and large separations in energy scales~\cite{Low2018IntPicSim}.
Even more recently, very promising tight theoretical bounds on the performance of traditional Lie-Trotter-Suzuki formulas on average-case Hamiltonians have been obtained~\cite{childs2019theory}.

Here, we introduce further refinements to the technique of qubitization applied to molecular systems, which has already been noted~\cite{Babbush2018encoding} to be particularly promising in terms of the dominant quantum Toffoli-gate complexity~\cite{Low2018Trading}.
We assume that the two-electron tensor describing interactions between $N$ molecular orbitals has a low-rank approximation in a so-called double-factorized representation.
This leads to Toffoli gate cost estimates that are orders of magnitude better.
For instance, we previously estimated that obtaining an energy level of a Nitrogen fixation problem to $1$mHartree~\cite{reiher17} cost the equivalent of $1.5\times 10^{14}$ Toffoli gates.
This was reduced to $2.3\times 10^{11}$ Toffoli gates by Berry et al.~\cite{Berry2019CholeskyQubitization} using a so-called single-factorized representation in the qubitization approach.
Under similar assumptions on the rank as Berry et al., we achieve $1.2\times 10^{10}$ Toffoli gates, a further improvement of more than an order of magnitude.

We have also taken into account realistic assumptions for mid-term  quantum hardware. While our previous work~\cite{reiher17} focused on optimistic future devices with  logical gate times of 100ns and all-to-all connectivity between qubits, we here employ  gate times of 10 $\mu$s for fault tolerant gates with nearest neighbor connectivity -- realistic assumptions for mid-term quantum computer architectures.  We therefore also discuss the overhead due to mapping the quantum algorithm to a two-dimensional planar layout, which further increases the overall runtime but makes the estimates more realistic.
\section{A homogeneous carbon dioxide fixation catalyst \label{sec:catalytic_cycle}}
The catalytic process that we selected for our present work is the binding and transformation of carbon dioxide.
The infrared absorption properties of carbon dioxide make it a green-house gas that is a major contributor to climate change.
Naturally, limiting or even inverting rising carbon dioxide levels in the atmosphere is a truly important goal and all possible
ways to accomplish it must be considered. One option, although currently not the most relevant one~\cite{hepburn19},
is carbon dioxide utilization by chemical transformation. In basic research, options have been explored to fix and react
inert carbon dioxide to yield chemicals of higher value (see Refs.~\cite{wu19,li20} for two recent
examples). Also, homogeneous transition metal catalysts have been designed in the past decade, many of them based on ruthenium
as the central metal atom. Despite the fact that eventually heterogeneous catalysis may be preferred over homogeneous catalysis,
understanding the basics of carbon dioxide fixation chemistry is facilitated by well-defined homogeneous systems. While formate is often the
product of such a fixation process, methanol is a chemical of higher value. Not many catalysts have been
reported so far that can transform carbon dioxide directly into methanol~\cite{kar18} and all of these are plagued by a comparatively low turnover number.
To find synthetic catalysts that robustly produce high-value chemicals from carbon dioxide with high turnover number
is therefore an important design challenge and computational catalysis can provide decisive insights as well as
virtual screening to meet this challenge.

For our assessment of quantum computing resource estimates, we chose a catalyst reported by the Leitner group~\cite{wesselbaum15} as extensive density functional theory (DFT) calculations on its mechanism have already been reported by this group. Hence, key molecular
structures have already been identified based on DFT. A detailed
mechanistic picture has emerged, from which we took intermediates and transition state structures for our resource estimate analysis.
Lewis strucures of the eight structures selected for our work are shown in~\cref{cycle}.

\begin{figure*}
\includegraphics[width=0.7\textwidth]{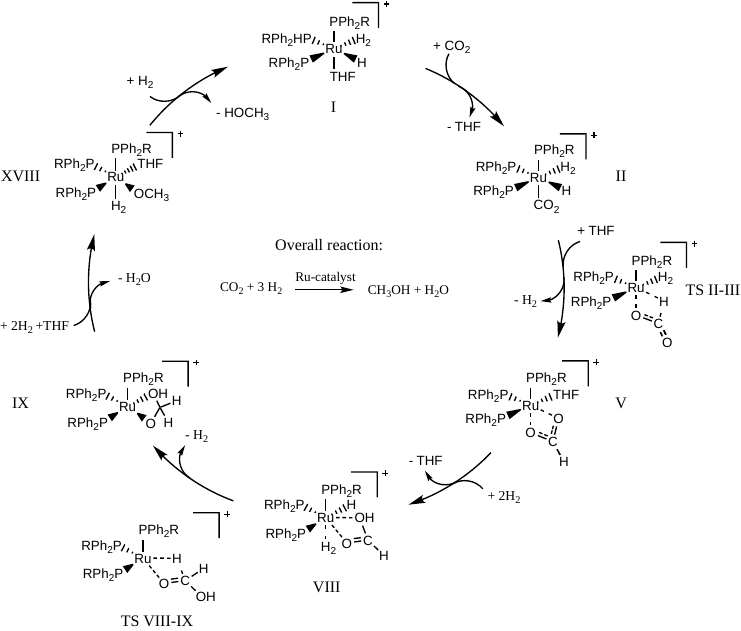}
\caption{Selected steps of the catalytic cycle elucidated in Ref.~\cite{wesselbaum15} on the basis of DFT calculations:
intermediates and transition state structures considered for the present work are shown
(Roman numerals are according to the original publication).
\label{cycle}}
\end{figure*}

In this work, we focus on the electronic structure of isolated complex structures and therefore neglect
any surrounding such as a solvent as well as energy contributions from nuclear dynamics that would be required
for the calculation of free energies (cf.~\cite{reiher17} on how to include them).
The heavy element ruthenium is known to form complexes that are often low-spin, i.e., singlet or doublet, and do not
feature strong multi-configurational character (by contrast to its lighter homolog iron). It is therefore not surprising
that we found no pronounced multi-configurational character as highlighted by the orbital entanglement diagrams provided in the
supplementary material. Hence, the Ru complexes selected for our resource estimate study feature mostly dynamic electron
correlation indicated by small weights of all but one electronic configurations (Slater determinants) that contribute to the
exact wave function in a full configuration interaction expansion. For our resource analysis, this is of little importance
but highlights the importance of a larger and faster universal quantum computer with a few thousand logical qubits for complete state
representation, to accurately capture all relevant dynamical correlations.

The presentation of the CO$_2$-fixating ruthenium catalyst of the Leitner group~\cite{wesselbaum15}
was accompanied by extensive DFT calculations of the potentially relevant molecular structures, which is a routine procedure in
chemistry. However, DFT electronic energies are plagued by approximations made to the exchange-correlation functional~\cite{cohen12,mardirossian17} and
can therefore be of unknown reliability (cf., e.g., Refs.~\cite{weymuth14,simm16,husch18}).
The authors of Ref.~\cite{wesselbaum15} selected the Minnesota functional M06-L~\cite{zhao06a}
because they found good overall agreement with their experimental results. In order to demonstrate the uncertainty
that generally can affect DFT results, which can be
very cumbersome if experimental data are not available for comparison, we supplemented the DFT data of Ref.~\cite{wesselbaum15}
with results obtained with the standard generalized-gradient-approximation functional PBE~\cite{perdew96} and its hybrid variant PBE0~\cite{perdew96b}
for the eight structures of the reaction mechanism considered in this
work (data shown in~\cref{DFT-data}).

As can be seen in~\cref{DFT-data}, whereas structure optimization (compare
results for structures optimized with the same functional with the single-point data denoted by the double slash notation)
has a small effect on the relative electronic energies, the difference between the functionals
can amount to more than 50 kJ/mol (i.e., 19 mHartree per molecule)
and can even reverse the qualitative ordering of the compounds (compare structures `V' und `VIII'). Obviously, in the absence of any additional information (such as experimental data or more accurate
calculations), it is virtually impossible to settle on reliable energetics with potentially dramatic consequences
for the elucidation of a reaction mechanism.

\begin{figure*}
\includegraphics[width=1.0\textwidth]{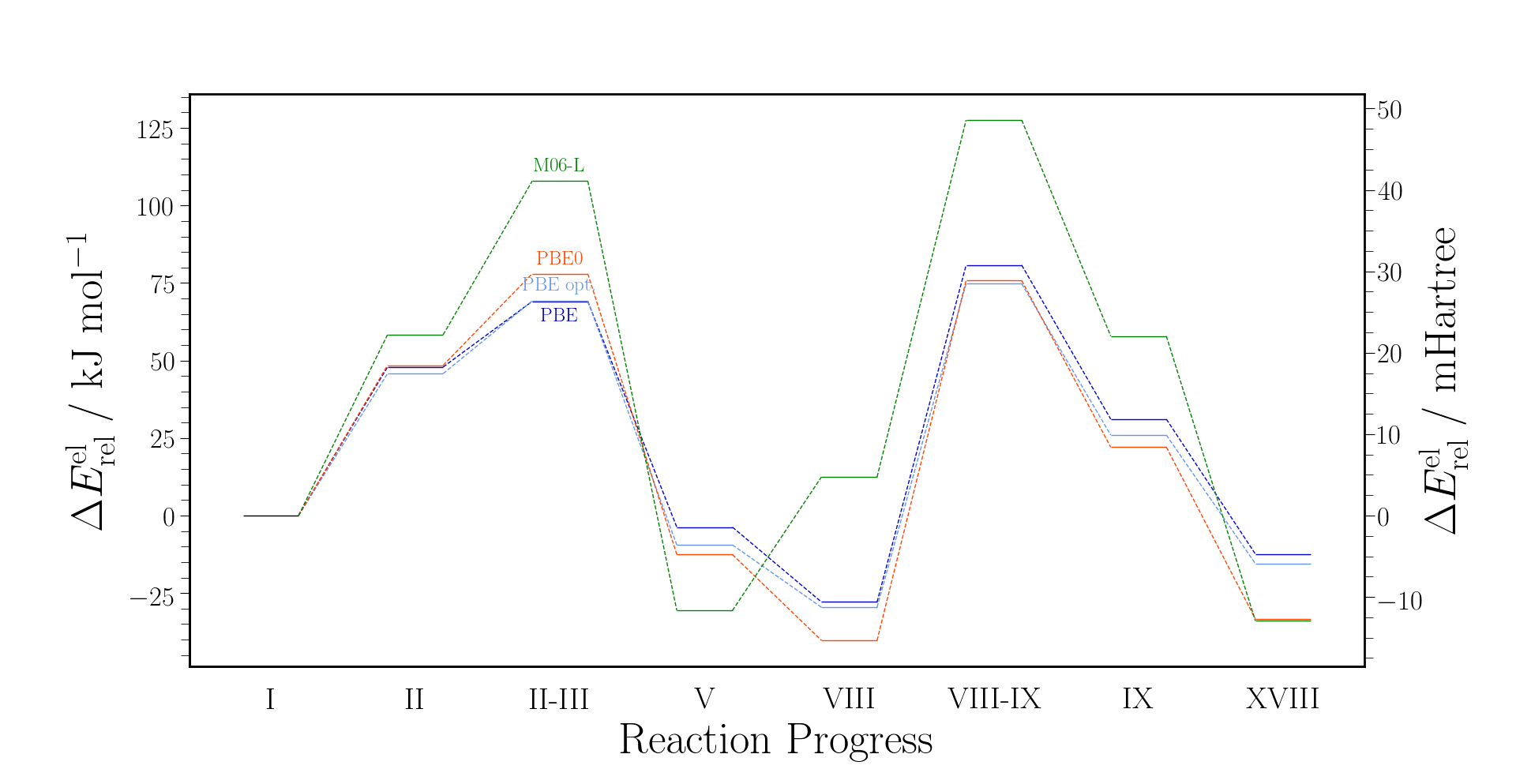}
\caption{\label{DFT-data} Relative DFT electronic reaction energies for the catalytic cycle obtained with a def2-TZVP basis set and
three different approximate density functionals: M06-L~\cite{wesselbaum15}, PBE, and PBE0 on M06-L/def2-SVP-optimized structures taken from
Ref.~\cite{wesselbaum15}. For comparison, we provide relative electronic reaction energies from PBE/def2-TZVP//PBE/def2-TZVP calculations
(in Pople's double slash notation, where the density functional behind the double slash is the one for which the structure was obtained) labeled as ``PBE opt.".}
\end{figure*}

\section{Quantum computing enhanced computational catalysis}
The elucidation of chemical reaction mechanisms is routinely accomplished with
approximate quantum chemical methods~\cite{cramer06,jensen07}, for which
usually stationary structures on Born-Oppenheimer potential energy surfaces are optimized, yielding stable reactants, products,
and intermediates of a chemical process. More importantly, optimized first-order saddle points on such a surface represent
transition state structures, which are key for detailed kinetic studies and hence for the prediction of concentration
fluxes.

In the following, we discuss where in the mechanism elucidation process quantum computing can be efficient, useful, and
decisive, and hence, how quantum computing can efficiently enhance and reinforce computational catalysis to make a difference.
As there are many steps involved that require a deep understanding of various branches of theoretical chemistry,
we provide an overview of the essential steps in~\cref{overview}. Understanding chemical catalysis, and chemical reactions in general, requires an exploration
of relevant molecular structures (\textit{structure exploration} in~\cref{overview}), which is usually done manually and with
DFT approaches (as in Ref.~\cite{wesselbaum15} for our example here), but which
can now also be done in a fully automated and even autonomous way~\cite{simm19}. These structures need to be
assigned an energy which may be conveniently separated into an electronic contribution (steps 1-7 in~\cref{overview})
and a remaining part (\textit{additional free energy caculations} in~\cref{overview}) containing
nuclear and other effects (calculated in a standard rigid-rotor--harmonic-oscillator model accompanied by dielectric
continuum embedding or by explicit molecular dynamics)
to eventually yield a free energy calculated from all related microstates. Relative free-energy differences
will eventually be used as barrier heights in expressions for absolute rate constants (\textit{kinetic modeling} in~\cref{overview})
that can then be used in kinetic
modeling for predicting concentration fluxes through the chemical web of relevant molecular structures. Ultimately, such knowledge
can be exploited to improve on existing or to design new catalysts with enhanced catalytic properties.

\begin{figure*}
	\includegraphics[width=1.0\textwidth]{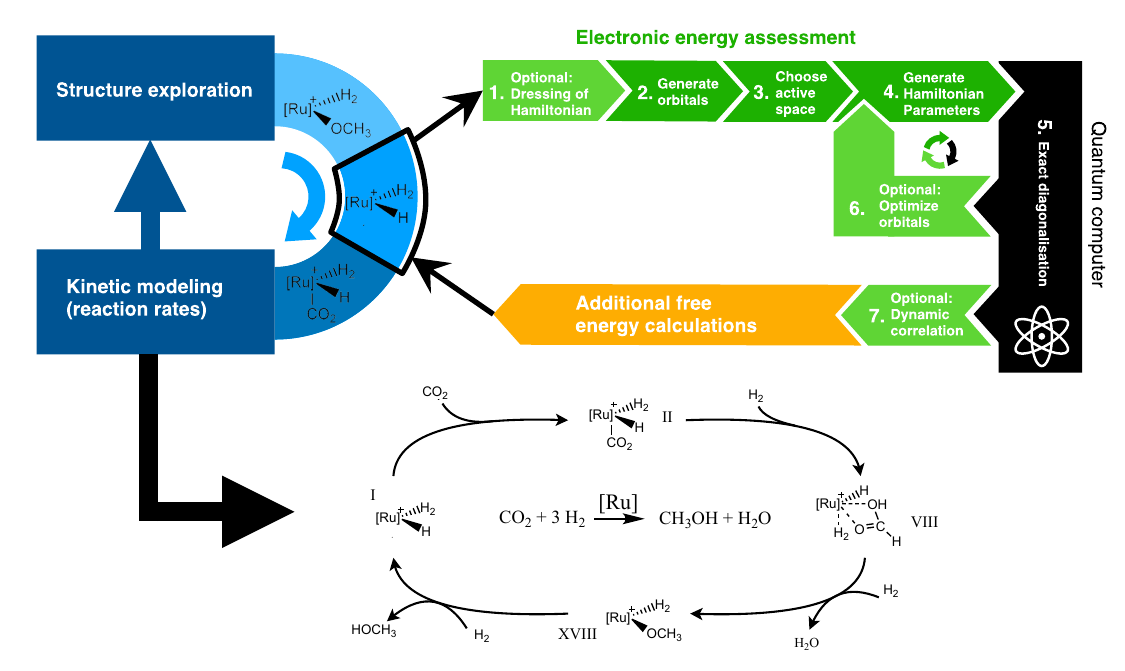}
	\caption{Protocol of computational catalysis with the key step of quantum
		computing embedded in black, which is usually accomplished with traditional
		methods such as CASSCF, DMRG, or FCIQMC (see text for further explanation). \label{overview}}
\end{figure*}

{\bf Accuracy matters:}
A reaction rate depends exponentially on the energy difference between a transition state structure and
its corresponding stable reactants, which are connected by an elementary reaction step. Because of this exponential dependence,
highly accurate energy differences are decisive. While many contributions enter these free energy differences, the
electronic energy difference is the most crucial one in bond breaking and bond making processes.

{\bf Electronic energies are key components (steps 1-7 in~\cref{overview}):}
Electronic energies are notoriously difficult to calculate and standard approximations are affected by unknown errors that can be large.
Only for electronically simple structures, so-called closed-shell
single-determinant electronic structures, well-established methods exist that run efficiently on a classical computer
(such as explicitly correlated local coupled cluster schemes with, at least,
perturbatively treated triple excitations~\cite{ma18}). For general electronic structures, however,
no method of comparable accuracy exists that is at the same time feasible for moderately sized molecules.
In particular for strong correlation cases,
which require more than one Slater determinant for a qualitatively correct description of the electronic wave function,
it can be hard to obtain an accurate total electronic energy, which then enters the calculation of relative energies.

In our previous work~\cite{reiher17}, we considered a quantum computer of moderate size within reach in the not too distant future.
Moreover, we assumed that such a machine might have 100 to 200 logical qubits available for the representation of a quantum state.
Such a state would be constructed from single-particle states, i.e., molecular orbitals for molecular structures.
For decades, it has been the goal of quantum chemistry to devise methods that efficiently construct approximations to a many-electron
state represented in terms of orbitals. If the full many-electron Hamiltonian is expressed and diagonalized in a complete many-electron
(determinantal) basis constructed from such a one-electron basis, then a full-CI calculation may be
carried out.
Such an exact diagonalization is, however, routinely only feasible for about 18 orbitals~\cite{fdez.galvan19} (a record
calculation was recently carried out for 24 orbitals~\cite{vogiatzis17}) due to the exponentially growing number of many-electron states
with the number of orbitals.

{\bf Required accuracy:}
A reasonable target accuracy for relative energies (and therefore also for total electronic energies of individual molecular
species) is about 1 mHartree, if not 0.1 mHartree. This corresponds to about 2.6 and 0.26 kJ/mol, respectively. Note
that thermal energy $RT$ ($T$ being temperature and $R$ the gas constant) at room temperature is on the order of 2.6 kJ/mol,
which may be related to the kinetic energy of a reactant molecule at average velocity (compare this to the spread observed
for different DFT functionals in our example in the last section).

It is important to understand that this target accuracy is important for relative energies, i.e., for the energy
differences that eventually enter the rate constant expressions. For total
electronic energies, however, this accuracy does not match the precision with which these energies are actually
known. In fact, the true total energies
are typically off by a huge energy offset because one does not make an effort to describe core electrons, which contribute
significantly to the total electronic energy, but not to reaction chemistry as they are atomically conserved.
Hence, such calculations rely on significant error cancellation effects that occur when atomic contributions
to the total electronic energy drop out in the calculation of reaction energies (which are relative energies) as they are
conserved during a reaction (consider, e.g.,
a molecular orbital that is mostly of $1s$-atomic orbital character and remains unaffected by a chemical reaction, but
contributes significantly to the total electronic energy).

{\bf Challenges of electronic structure (steps 1-7 in~\cref{overview}):}
It is therefore most important to get the electronic (valence) structure of each relevant molecular structure right.
Here, quantum computing offers an opportunity~\cite{reiher17}.
It is important to realize that a typical chemical catalysis problem does not involve very many valence orbitals in every elementary
reaction step.
As a consequence, the size of the active orbital space, from which the electronic wave function is constructed,
does usually not need to be very large and can be easily handled with methods such as Complete Active Space Self Consistent Field
(CASSCF), DMRG, or FCIQMC.
The latter two are capable of handling active orbital spaces of up to about 100 spatial orbitals owing to efficient approximations.
However, sometimes there may be a price to pay for these approximations and that is a residual uncertainty regarding convergence
of the electronic energy. For instance, a DMRG result will critically depend on -- apart from a fixed finite value of the bond dimension --
proper convergence of the sweeping algorithm, which, at times, might be difficult to determine. In the case of the more recent
FCIQMC approach, which is under continuous development and offers extraordinarily efficient scaling on large traditional parallel computers,
convergence with respect to the number of walkers or extrapolation to an infinite number of walkers may be hard to achieve for certain
molecules. We also note that it is generally hard for any quantum chemical approach (hence, also for DMRG and FCIQMC)
to deliver rigorous yet useful information about the error associated with a specific result.
Moreover, any active-space approach is plagued by a severe drawback discussed below, namely the neglect of dynamic electronic correlation arising
from the majority of virtual orbitals neglected.

{\bf Importance of dynamic electron correlation (steps 1 and 7 in~\cref{overview}):}
As noted already above, moderately sized molecules, such as the ones studied in this work, may easily require on the order of 1000
or more spatial molecular orbitals for a description that may be considered accurate within the chemically relevant accuracy of about 1 mHartree or even 0.1 mHartree for relative energies.
However, the restriction to the valence orbital space from which the active space of the most strongly correlated orbitals
is chosen~\cite{roos80,ruedenberg82} compromises this accuracy
(note that these active orbitals may be identified
based on natural-orbital occupation numbers~\cite{keller15a} or on orbital entanglement measures~\cite{stein16,stein16a}, even in a completely automated fashion~\cite{stein16,stein19}).
The vast majority of orbitals that are weakly entangled and neglected in this procedure give rise to dynamic electron correlations,
which are neglected in a small-CAS calculation. However, the dynamic electron correlation contribution to the total electronic
energy is decisive and so standard recipes exist to approximate it.
The most prominent one is the a posteriori correction (step 7 in~\cref{overview})
provided by multi-reference perturbation theory~\cite{gonzalez19}, which, however, requires elements of the three- and four-body reduced density matrices (3-RDMs and 4-RDMs, respectively) that would be extremely hard to obtain by quantum computing.
Not even approximate approaches that rely on at most some 3-RDM elements will be accessible by quantum computing
for interesting molecules.
To evade such a computational bottleneck, perturb-then-diagonalize approaches (step 1 in~\cref{overview})
such as range-separation DFT for CAS-type methods~\cite{fromager07,Hedegard2015_DMRG-srDFT,hedegard18} or transcorrelation approaches~\cite{luo18} were proposed for quantum computing~\cite{reiher17}.

Quantum computing is supposed to be a valuable and, in the long run, ultimately superior competitor to the aforementioned
traditional methods
(i) because rigorous error estimates are available and (ii) because systems may be accessible that are traditionally not
feasible because of the curse of dimensionality when a total state is to be represented in a large set of
orbitals. Its true benefits will fully unfold if
energy measurements in the full orbital basis become feasible (on this, see the discussion in the Conclusions section).

{\bf Four-index transformation (step 4 in~\cref{overview}):}
A cumbersome technical step, to be carried out by traditional computing on classical computers, is the
four-index transformation of the two-electron interaction integrals in the Hamiltonian, which scales with the fifth power of
the number of basis functions. In this transformation, the
final parameters for the electronic Coulomb Hamiltonian in the molecular orbital basis are produced from the
four-index integrals defined in the atomic orbital basis, i.e., in the basis which is provided for the representation
of all molecular orbitals (typically a set of Gaussian-type functions such as the def2-TZVP basis set used in the DFT
calculations presented above). Note, however, that this transformation, while
being expensive but feasible for a small CAS, becomes a threat to the whole calculation when
the active space grows to eventually incorporate the complete one-electron (atomic orbital) basis
(of say, more than 1000 basis functions). Hence, the sheer
number of two-electron parameters in the electronic Hamiltonian will require us to rethink how to deal with these terms
growing to the fourth power in the number of orbitals, such as in recent work~\cite{Peng2017Cholesky} that generate sparse low-rank approximations to the two-electron integrals. Naturally, this issue has also been discussed in modern traditional approaches~\cite{aquilante11a}.

\section{Quantum algorithms for chemistry \label{sec:quantum_algorithms_for_chemistry}}
For the catalysis problem we require  quantum algorithms that provide reliable results with controlled errors on the electronic energy in a given orbital basis. Uncontrolled approximations in quantum algorithms would negate the advantages offered by quantum computers, which will require  tremendous effort to build and operate even at a moderately sized error-corrected scale.

Though it is without doubt that the popular variational
quantum eigensolver (VQE)~\cite{peruzzo14} can obtain parameters of a unitary coupled cluster (UCC) parametrization of the electronic wavefunction that is likely to be accurate (especially when higher than double excitations are considered~\cite{grimsley19}), this scheme unavoidably generates a residual unknown uncertainty in the true electronic energy. Reducing these errors by improving the ansatz will require significantly more gate operations than are expected to be possible on non-error corrected noisy intermediate scale quantum devices~\cite{preskill18,kuehn2019accuracy}. %
Moreover, the number of repetitions required to estimate energies with sufficient accuracy of 1 mHartee or better is enormous~\cite{Wecker2015Variational}.

We also note that knowledge of reliable and controllable errors in quantum algorithms
is a decisive advantage over classical methods such as DMRG and FCIQMC, for which convergence control with respect to
their parameters is not necessarily easy or even feasible.

Therefore, we turn to one of the most promising applications of quantum computers, which is a bounded-error simulation of quantum systems using quantum phase estimation.
The main idea is to synthesize a quantum circuit
that implements the real time-evolution operator $W=e^{-iH/\alpha}$ by a given Hamiltonian $H$ for some normalizing factor $\alpha$, which henceforth are always in atomic units of Hartree.
When applied $n$ times to an eigenstate $H\ket{\psi_k}=E_k\ket{\psi_k}$, a phase  $nE_k/\alpha$ is accumulated.
Quantum phase estimation then estimates the energy $E_k$ with a standard deviation $\Delta_E=\mathcal{O}(\alpha/n)$~\cite{Wiebe2016PhaseEstimation}.
If one prepares an arbitrary trial state $\ket{\psi_{\text{trial}}}$ rather than an eigenstate, phase estimation collapses the trial state to the $k$-th eigenstate with probability $p_k=|\braket{\psi_k}{\psi_{\text{trial}}}|^2$ and it returns an estimate to the corresponding energy $E_k$~\cite{Nielsen2004}.

The phase estimation procedure is executed on a quantum computer by applying a sequence of quantum gates.
If the unitary $W$ is implemented using a number $c_W$ of quantum gates, the overall quantum gate cost of obtaining a single estimate $\hat{E}_k$ is then
\begin{align}
\label{eq:QPE_repetitions}
c_W\times\frac{\pi\alpha}{2\Delta_E},
\end{align}
where the factor $\frac{\pi\alpha}{2\Delta_E}$ arises from previous analyses on the performance of phase estimation~\cite{Babbush2018encoding,Wojciech2020Heisenberglimit} (combined with a so-called phase-doubling trick~\cite{Wecker2015Strongly,Babbush2018encoding}).
In general, the quantum circuit only approximates $W$ to some bounded error $\Delta_W$ in spectral norm.
This adds a systematic bias of $\alpha\Delta_W$ to the estimate  $\hat{E}_k$.
Thus, we budget for this error by making the somewhat arbitrary choice of performing phase estimation to an error of $0.9 \Delta_E$, and compiling $W$ so that $\Delta_W \le 0.1 \Delta_E / \alpha$.
To date, there are several prominent quantum algorithms for approximating real time-evolution, such as Lie-Trotter-Suzuki product formulas~\cite{Suzuki1990Fractal},
sparse Hamiltonian simulation~\cite{Low2018SpectralNorm}, linear-combination
of unitaries~\cite{Childs2012LCU}, qubitization~\cite{Low2016Qubitization},
and quantum signal processing~\cite{Low2016HamSim}.
In this work, we consider the electronic Hamiltonian in its non-relativistic form with Coulomb interactions (in Hartree atomic units),
\begin{align}
\label{eq:hamiltonian_unfactorized}
H &= \sum_{ij,\sigma}h_{ij}a_{(i,\sigma)}^{\dagger}a_{(j,\sigma)}
\\ &\phantom{=}  + \frac{1}{2}\sum_{ijkl,\sigma\rho}h_{ijkl}a_{(i,\sigma)}^{\dagger}a_{(k,\rho)}^{\dagger}a_{(l,\rho)}a_{(j,\sigma)}, \nonumber
\end{align}
which is parametrized through the one- and two electron integrals $h_{ij}$ and $h_{ijkl}$ of the molecular orbitals $\{\psi_i \}$,

\begin{align}
h_{ij} &= \int\psi^*_{i}(x_{1})\left(-\frac{\nabla^{2}}{2}-\sum_{m}\frac{Z_{m}}{|x_{1}-r_{m}|}\right) \\
 &\phantom{=} \quad \quad \quad \times\psi_{j}(x_{1})\mathrm{d^{3}}x_{1},\nonumber\\
 h_{ijlk} &= \int\psi^*_{i}(x_{1})\psi_{j}(x_{1})\left(\frac{1}{|x_{1}-x_{2}|}\right) \nonumber\\
 &\phantom{=} \quad \quad \quad \times \psi^*_{k}(x_{2})\psi_{l}(x_{2})\mathrm{d^{3}}x_{1}\mathrm{d^{3}}x_{2},
\end{align}
where $x_i$ denote electronic coordinates and $Z_m$ is the charge number of nucleus $m$ at position $r_m$ (note that a relativistic generalization is straightforward~\cite{reiher15}).
We explicitly separate the fermion indices $p\equiv(i,\sigma)$ into an index where $i\in\{1,...,N\}$ enumerates the $N$ spatial molecular orbitals, and $\sigma\in\{0,1\}$ indexes spin-up and spin-down.
Hence, the fermion operators satisfy the usual anti-commutation relations
\begin{align}
\left\{ a_{p},a_{q}\right\} =0,\;\;\left\{ a_{p}^{\dagger},a_{q}^{\dagger}\right\} =0,\;\;\left\{ a_{p},a_{q}^{\dagger}\right\} =\delta_{pq}\ii,
\end{align}
and the coefficients $h_{ij}$, $h_{ijkl}$ are real and satisfy the symmetries
\begin{align}
\label{eq:symmetry}
h_{ij} &= h_{ji},\\\nonumber
h_{ijkl} &= h_{jikl}=h_{ijlk}=h_{jilk}= h_{lkij} \\ &= h_{lkji}=h_{klij}=h_{klji}. \nonumber
\end{align}

For the purposes of phase estimation through the so-called \textit{qubitization} approach, it is not necessary to simulate time-evolution $e^{-iHt}$.
Whereas our previous work~\cite{reiher17} approximated the time-evolution operation using Lie-Trotter-Suzuki product formulas, it can be advantageous in some cases to implement the unitary walk operator $W=e^{i\sin^{-1}(H/\alpha)}$~\cite{Low2016Qubitization,Poulin2018Spectral,babbush2017sorting}
instead, which has some normalizing constant $\alpha\ge \|H\|$ that ensures the arcsine is real.
This walk operator can be implemented exactly, assuming access to arbitrary single-qubit rotations, in contrast to all known quantum algorithms where time-evolution $e^{-iHt}$ can only be approximated.
After estimating the phase $\hat{\theta}=\sin^{-1}(E_k/\alpha)$ with phase estimation, we may obtain $\hat{E}_k$ by applying $\sin(\hat{\theta})$ in a classical postprocessing step.

We focus on the qubitization technique applied to electronic structure~\cref{eq:hamiltonian_unfactorized}, with the goal of minimizing the quantum gate costs in~\cref{eq:QPE_repetitions}.
In fault-tolerant architectures, quantum gate costs reduce to the number of so-called primitive Clifford gates (e.g. Hadamard$=\frac{1}{\sqrt{2}}\begin{bsmallmatrix}
	1 & 1 \\ 1 & -1
\end{bsmallmatrix}$, phase$=\begin{bsmallmatrix}
1 & 0 \\ 0 & i
\end{bsmallmatrix}$, and controlled-$\textsc{Not}$) and non-Clifford gates (e.g. $T=\begin{bsmallmatrix}
1 & 0 \\ 0 & \sqrt{i}
\end{bsmallmatrix}$ and Toffoli, which applies a $\textsc{Not}$ gate controlled on two input bits being in the `one' state).
As the physical resources needed to implement a single error-corrected primitive non-Clifford gate, such as the $\t$ gate is on the order of $100$ to $10000$ times more costly than a single two-qubit Clifford gate~\cite{Maslov2016,Gidney2019efficientmagicstate} --
and one Toffoli gate may be implemented by four $\t$ gates  --
much recent work has focused on optimizing the non-Clifford-gate cost of $W$.
Within an atom-centered basis set (such as the ones employed in this work), this focus has reduced the $\t$ gate cost of obtaining a single estimate $\hat{E}_k$ from $\tilde{\mathcal{O}}(N^5/\Delta_E^{3/2})$ using Trotter methods to $\tilde{\mathcal{O}}(\alpha_{\text{CD}} N^{3/2}/\Delta_E)$ Toffoli gates and $\mathcal{O}(N^{3/2}\log{N/\Delta_E})$ qubits using qubitization combined with a so-called single-factorized Hamiltonian representation, where $\alpha_{\text{CD}}=\mathcal{O}( N^{3})$~\cite{Berry2019CholeskyQubitization} is a certain norm of the Hamiltonian.
Our main technical contribution is a further reduction of the gate cost of $W$ and
the normalizing factor to $\alpha_{\text{DF}}=\mathcal{O}(N^{3/2})$ as described in the next Section.

\section{Efficient encoding of double-factorized electronic structure \label{sec:encoding}}
Our main algorithmic advance is based on  a quantum algorithm to `qubitize' the so-called double-factorized representation  $H_{\text{DF}}$ of the electronic Hamiltonian $H$.
On the one hand, the double-factorized representation is sparse, and therefore minimizes the cost $c_W$ of qubitization, which generally scales with the number of terms needed to represent the Hamiltonian.
On the other hand, the double-factorized representation is a partial diagonalization of the original Hamiltonian, and hence has a small normalizing constant $\alpha_{\text{DF}}$.
Despite these favorable properties, previous quantum simulations using this representation are based on Trotter methods~\cite{Poulin2015Trotter,Motta2018LowRankTrotter}.

The double-factorized form builds upon the single-factorized representation $H_{\text{CD}}$ of the Hamiltonian $H$ where
\begin{align}
\label{eq:HamiltonianFactorized}
H_{\text{CD}}  &\doteq  {\sum_{ij,\sigma}\tilde{h}_{ij}a_{(i,\sigma)}^{\dagger}a_{(j,\sigma)}}
 \\\nonumber &\phantom{\doteq} +{\frac{1}{2}\sum_{r\in[R]}\left(\sum_{ij,\sigma}L_{ij}^{(r)}a_{(i,\sigma)}^{\dagger}a_{(j,\sigma)}\right)^{2}},
\\  \quad \tilde{h}_{ij}\doteq & h_{ij}-\frac{1}{2}\sum_{l}h_{illj}.
\end{align}
Note that the rank-$R$ factorization of the two-electron tensor $h_{ijkl}=\sum_{r\in[R]}L^{(r)}_{ij}L^{(r)\top}_{kl}$~\cref{eq:hamiltonian_unfactorized} into the $N\times N$ symmetric matrices $L^{(r)}$ always exists due the symmetry constraints~\cref{eq:symmetry}, and may be computed using either a singular-value decomposition or a Cholesky decomposition.
This representation also facilitates a low-rank approximation by truncating the rank $R$.
In the worst-case, $R\le N^2$.
However, it was noted by Peng et al.~\cite{Peng2017Cholesky} that rank $R\sim N\log{N}$ for typical molecular systems when $N$ is proportional to the number of atoms, which is a provable statement for 1D systems~\cite{Motta2019LowRank}.
This reduces the number of terms needed to describe the second-quantized Hamiltonian from $\mathcal{O}(N^4)$ in~\cref{eq:hamiltonian_unfactorized} to $\mathcal{O}(RN^2)$ in~\cref{eq:HamiltonianFactorized}.
This representation was first exploited by Berry et al.~\cite{Berry2019CholeskyQubitization} to qubitize $H_{\text{CD}}$ with a normalizing constant
$
\alpha_{\text{CD}}\doteq 2\normEW{\tilde{h}}+2\sum_{r\in[R]}\normEW{L^{(r)}}^{2}$ expressed using the entry-wise norm $
\normEW{L^{(r)}}\doteq\sum_{ij\in[N]}|L^{(r)}_{ij}|.
$

The technical innovation in our approach is a quantum circuit, detailed in
the supplementary material, for qubitzing the double-factorized Hamiltonian
\begin{align}
\label{eq:HamiltonianDoubleFactorized}
H_{\text{DF}} &= {\sum_{ij,\sigma}\tilde{h}_{ij}a_{(i,\sigma)}^{\dagger}a_{(j,\sigma)}}
\\ & +{\frac{1}{2}\sum_{r\in[R]}\Biggl(\sum_{ij,\sigma}\sum \limits_{\substack{m \\ \in[M^{(r)}]}}\lambda^{(r)}_m \vec{R}_{m,i}^{(r)} \vec{R}_{m,j}^{(r)}a_{(i,\sigma)}^{\dagger}a_{(j,\sigma)}\Biggr)^{2}}.\nonumber
\end{align}

This is obtained by a rank-$M^{(r)}$ eigendecomposition of the symmetric matrices
$L^{(r)}=\sum_{m\in[M^{(r)}]}\lambda^{(r)}_m \vec{R}_m^{(r)}\cdot( \vec{R}_m^{(r)})^\top$ into a total of $M\doteq\sum_{r\in[R]}M^{(r)}$, each normalized to be of unit length $\|\vec{R}_m^{(r)}\|_2 = 1$.
The number of terms in~\cref{eq:HamiltonianDoubleFactorized} is even further reduced to $\mathcal{O}(MN)$ as Peng et al.~\cite{Peng2017Cholesky} also noted that for typical molecular systems where $N\gg10^3$ scales with the number of atoms, one may retain $M\sim N\log{N}$ eigenvectors and truncate the rest.
A key technical step in our approach is to work in the Majorana representation of  fermion operators
\begin{align}
\label{eq:Majorana}
\gamma_{p,0} &= a_{p}+a_{p}^{\dagger},\\
\gamma_{p,1} &= -i\left(a_{p}-a_{p}^{\dagger}\right),\\
\left\{ \gamma_{p,x},\gamma_{q,y}\right\}  &= 2\delta_{pq}\delta_{xy}\ii.
\end{align}
As we show in the supplementary material, this representation maps
\begin{align}
\label{eq:Ham_df_Majorana}
H_{\text{DF}}\rightarrow & \left(\sum_{i}h_{ii}-\frac{1}{2}\sum_{il}h_{illi}+\frac{1}{2}\sum_{il}h_{llii}\right)\mathcal{I}
\\\nonumber & +\onebody_{L^{(-1)}}
+\frac{1}{2}\sum_{r\in[R]}\onebody^2_{L^{(r)}},
\end{align}
which is expressed as a sum of squares of one-body Hamiltonians
\begin{align}
\onebody_{L}&\doteq\frac{i}{2}\sum_{ij}\sum_{\sigma}L_{ij}\gamma_{i,\sigma,0}\gamma_{j,\sigma,1},\\
 L^{(-1)}_{ij} &\doteq h_{ij}-\frac{1}{2}\sum_{l}h_{illj}+\sum_{l}h_{llij}.
\end{align}
This leads to a normalizing constant
\begin{align}
\alpha_{\text{DF}}=\normSH{L^{(-1)}}+\frac{1}{4}\sum_{r\in[R]}\normSH{L^{(r)}}^2.
\end{align}
Note that the difference between $L^{(-1)}$ and $\tilde{h}$ is the mean-field Coulomb repulsion contribution term $\sum_{l}h_{llij}$.
Note that $\alpha_{\text{DF}}$ is also significantly smaller than $\alpha_{\text{CD}}$ from previous approaches~\cite{Berry2019CholeskyQubitization}.
In addition to the factor of eight reduction in the prefactor of the two-electron norms, the dependence of $\alpha_{\text{DF}}$ on Schatten norms $\normSH{L^{(r)}}$ is beneficial as they can be up to a factor of $N$ smaller than the entry-wise norms $\normEW{L^{(r)}}$ that $\alpha_{\text{CD}}$ depends on, as follows from the following tight inequalities for any Hermitian $N\times N$ matrix $h$,
\begin{align}
&\normSH{h}\doteq\sum_{k\in[N]}|\text{Eigenvalues}[h]_k|, \\
&\frac{1}{N}\normEW{h}\le\normSH{h}\le\normEW{h},
\end{align}
which we prove in the supplementary material and may be of independent interest.

The Toffoli gate complexity of synthesizing the walk operator to an error $\Delta_W=0.1\Delta_E/\alpha_{\text{DF}}$, as detailed in the supplementary material, is then
\begin{align}
\label{eq:Toffoli_cost}
c_W &\le   \frac{2M}{1+\lambda}+2\lambda N\beta+8N\beta+4N \nonumber \\ &\phantom{\le } + \mathcal{O}(\sqrt{R\log{M}}+\sqrt{M \log{(1/\Delta_W)}}),
\end{align}
for any choice of integer $\lambda\ge0$, and where the parameter $\beta=\left\lceil5.652+\log_{2}\left(\frac{N}{\Delta_W}\right)\right\rceil$.
This also uses the following number of qubits
\begin{align}
\label{eq:Qubit_cost}
N\beta(1+\lambda)+2N+\mathcal{O}(\log{(N/\Delta_W)}).
\end{align}
Note that $\lambda$ controls the number of ancillary qubits.
By choosing $\lambda=\mathcal{O}(\sqrt{\frac{M}{N\beta}})$, the Toffoli cost is $c_W=\mathcal{O}(\sqrt{MN\beta}+N\beta)$, which is advantageous when $M\gg N\beta$, using $\mathcal{O}(\sqrt{MN\beta}+N+\beta)$ qubits,
Assuming the empirical scaling of $M$ and $R$ by Peng et al., our algorithm encodes electronic spectra for atom-centered basis sets into the walk operator using only $c_W=\tilde{\mathcal{O}}(N)$ Toffoli gates, and with a normalizing constant $\alpha_{\text{DF}}$, which improves upon the $\tilde{\mathcal{O}}(N^{3/2})$ Toffoli gates and larger normalizing constant $\alpha_{\text{CD}}$ required by the single-factorized approach~\cite{Berry2019CholeskyQubitization}.

In the following examples we consider, typical values of $\lambda$ that minimize the Toffoli costs are between $1$ and $5$, and the Toffoli gate counts and qubit counts we quote exclude the sub-dominant big-$\mathcal{O}(\cdot)$ component of~\cref{eq:Toffoli_cost,eq:Qubit_cost}.

\subsection{Truncation \label{eq:truncation}}

We now describe our procedure for obtaining, by truncating eigenvalues, low-rank approximations $\tilde{H}_{\text{DF}}$ of the double-factorized Hamiltonian~\cref{eq:Ham_df_Majorana} from an initial numerically-exact representation $H_{\text{DF}}$.
We focus on truncating the two-electron terms as those tend to dominate the cost of block-encoding.
Given a target approximation error $\epsilon$, we consider two truncations schemes, which we call {\textit{coherent}} and {\textit{incoherent}}.
The coherent scheme upper bounds the actual error
\begin{align}
\|\tilde{H}_{\text{DF}}-H_{\text{DF}}\| \le \epsilon_{\text{co}},
\end{align}
in spectral norm for any given choice of $\epsilon_{\text{co}}$.
This upper bound is obtained by a triangle inequality, which assumes that all truncated terms have an error that adds linearly.
That is, if the difference $\tilde{H}_{\text{DF}}-H_{\text{DF}}=\sum_{j}H_j$ is a sum of terms, then we truncate so that
\begin{align}
\sum_{j}\|H_j\|\le \epsilon_{\text{co}}.
\end{align}
We find that this bound is often quite loose, and the gap between the actual error and $\epsilon$ grows with system size.
This motivates the \textit{incoherent} scheme, which assumes that the error of truncated terms add incoherently by a sum-of-squares.
Thus, we truncate so that
\begin{align}
\sqrt{\sum_{j}\|H_j\|^2}\le \epsilon_{\text{in}}.
\end{align}
Suppose we remove a single eigenvalue $\lambda^{(r)}_m$ from~\cref{eq:Ham_df_Majorana} during truncation.
Using the identity $A^2-(A-B)^2=AB+BA-B^2$,
the difference
\begin{align}
\label{eq:HamiltonianDoubleFactorizedError}
&H_{\text{DF}}-\tilde{H}_{\text{DF}} =  \\ & \frac{1}{2}\Biggl(\left\{\onebody_{L^{(r)}},\frac{\lambda^{(r)}_m}{2}\sum_\sigma \gamma_{\vec{R}^{(r)}_m,\sigma,0}\gamma_{\vec{R}^{(r)}_m,\sigma,1}\right\}\nonumber\\ & \phantom{=} - \left(\frac{\lambda^{(r)}_m}{2}\sum_\sigma \gamma_{\vec{R}^{(r)}_m,\sigma,0}\gamma_{\vec{R}^{(r)}_m,\sigma,1}\right)^2\Biggr)\nonumber,
\end{align}
can be bounded in $\|\cdot\|$ as follows:
\begin{align}
\|H_{\text{DF}}-\tilde{H}_{\text{DF}}\|
&\le |\lambda_{m}^{(r)}|\sum_{n\neq m}|\lambda_{n}^{(r)}|+\frac{1}{2}|\lambda_{m}^{(r)}|^2  \\\nonumber
&= |\lambda^{(r)}_m|\left(\normSH{L^{(r)}}-\frac{1}{2}|\lambda^{(r)}_m|\right)\\
&\le \normSH{L^{(r)}}|\lambda^{(r)}_m|. \nonumber
\end{align}
In the coherent scheme, we truncate eigenvalues in the index set $\mathcal{T}$ such that $\sum_{(r,m)\in\mathcal{T}}\normSH{L^{(r)}}|\lambda^{(r)}_m|\le \epsilon_{\text{co}}$. In the incoherent scheme, we truncate such that $\sqrt{\sum_{(r,m)\in\mathcal{T}}(\normSH{L^{(r)}}|\lambda^{(r)}_m|)^2}\le \epsilon_{\text{in}}$.
In both cases we may maximize the number of truncated eigenvalues by deleting those with the smallest value of $\normSH{L^{(r)}}|\lambda^{(r)}_m|$ first.

In the following, we truncate according to the sum-of-square procedure.
While this approach does not rigorously bound the total error, we find that it better matches the error we observed for our benchmark systems, see~\cref{II-III-smallAS-truncation}, but still tends to significantly overestimate the actual error.

\begin{figure*}
	\centering
	\includegraphics{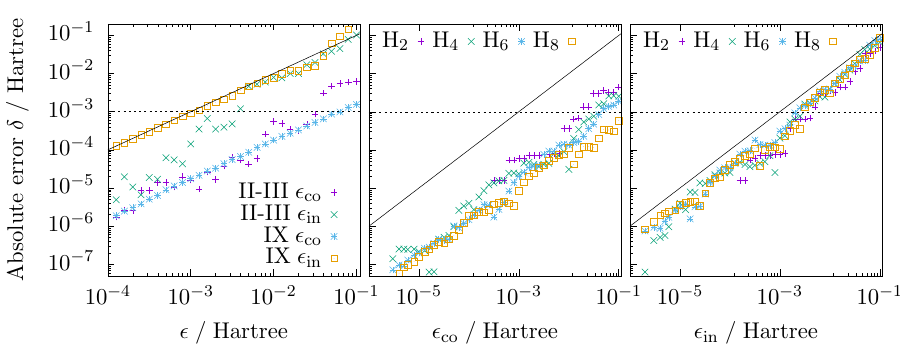}
	\caption{Empirical absolute error $\delta$ (in Hartree) of the ground-state electronic energy resulting from applying the two truncation schemes to the two-electron integrals of the Hamiltonian at various thresholds $\epsilon_{\text{co}}, \epsilon_{\text{in}}$, evaluated from DMRG-CI calculations.
		The lines $\delta = \epsilon$ and $\delta = 1\text{mHartree}$ are meant to guide the eye.
		(Left) Absolute error of the ground-state electronic DMRG-CI energy for the two truncation schemes for complex II-III with $6$ active orbitals and complex IX with $16$ active orbitals.
		(Middle) Absolute error in the ground-state electronic DMRG-CI energy of linear Hydrogen chains of length 2, 4, 6, and 8 at different truncation thresholds $\epsilon_{\text{co}}$ for the coherent truncation scheme and (Right) $\epsilon_{\text{in}}$ for the incoherent truncation.}
	\label{II-III-smallAS-truncation}
\end{figure*}

\section{Results}
For the intermediates depicted in the catalytic cycle in ~\cref{cycle}, we carried out density
functional calculations that delivered the data presented in ~\cref{DFT-data} and wave function calculations
with various active orbital spaces for the analyses discussed below. For the sake of brevity, we refer the
reader to the supplementary material for all technical details on these standard quantum chemical calculations.

For each of these key intermediates and transition states we then evaluate the cost of performing quantum phase estimation to chemical accuracy, both
for small active spaces of  $52$--$65$ orbitals and larger ones, and then discuss runtimes and qubit requirements.

\subsection{Selection of active orbitals}
As mentioned in~\cref{sec:quantum_algorithms_for_chemistry},
the electronic Hamiltonian in~\cref{eq:hamiltonian_unfactorized} is parametrized by one- and two-electron integrals, $h_{ij}$ and $h_{ijkl}$, respectively, over the molecular orbitals
of a restricted orbital subspace, the active space.
For the selection of the active space (cf.~\cref{overview}), all strongly correlated orbitals
must be identified (see the detailed discussion for our carbon dioxide fixation process in the supporting
information).
For this, we relied on orbital entanglement measures which we applied in an automated procedure~\cite{stein16,stein19}, but we emphasize that the importance of this choice for selecting the different orbital space sizes for this work is not crucial at all. Analysis of the resulting states in terms of these orbital entanglement measures, pair-orbital mutual information, and occupation numbers shows that the electronic ground states of the catalyst structures selected for our study are not dominated by static correlation. For such cases, the traditional quantum chemistry toolbox offers efficient and very accurate coupled cluster methods at the basis set limit, which represent a true challenge for quantum algorithms to compete with. Still, we seek to analyze a general approach in quantum computation that
can deal with strong (static) correlation (and possibly with any correlation problem in the not too distant future).
It is therefore not decisive for our resource analysis that the active space sizes chosen (see supporting information)
do neither obey nor exploit typical patterns of electron correlation. It was, however, still possible to
select a small active space of strongly correlated orbitals on the size of five to sixteen orbitals. The main purpose of this work is understanding the performance of quantum algorithms on an actual
chemical problem of varying size. The varying size is the growing size of the  orbital space in which the exact
wave function is constructed. In order to construct a proper test bed, we are therefore forced to choose
active-space sizes that are arbitrary with respect to the proper balance of static and dynamic correlation. In other words, to grow the active space sizes beyond the classical limit, we had to add weakly correlated orbitals. Naturally, it then became increasingly difficult to
select the active spaces based on orbital entanglement alone (because all such orbitals show similarly weak entanglement entropy measures).
In this study, we therefore resorted to criteria which ensure a reproducible selection that is, at the same time,
reasonable from a chemical point of view. We chose three different active space sizes (small, intermediate, and large)
for each molecular structure of the catalytic cycle. The first active space comprises those orbitals selected based on
orbital entanglement criteria. Note that this small active space is identical to the one employed in the CASSCF calculations
that produced the molecular orbitals from which the Hamiltonian parameters $h_{ij}$ and $h_{ijkl}$ were calculated.
The intermediate active space includes all the orbitals of the small active space and then in addition
the valence orbitals of ruthenium and the ligands (excluding the triphos ligand and, if present, solvent molecules), as well as the orbitals involved in ligand-metal bonding.
For the largest active space, we further supplemented these orbitals with the $\pi$ and $\pi^*$ orbitals
on the triphos ligand, as they form a well-defined and easily identifiable set of orbitals.
Apart from these three active spaces per catalyst structure, we created even larger active spaces for structure XVIII,
for which we additionally selected three active spaces of size $n=$100, 150, and 250 spatial molecular orbitals.
Since a manual selection is pointless for such large active spaces of mostly dynamically correlated orbitals,
we selected the $n/2$ occupied and $n/2$ virtual orbitals around the Fermi level.
For each of these active spaces, we obtained the one- and two-electron integrals $h_{ij}$ and $h_{ijkl}$
via a four-index transformation (see also~\cref{overview}). These integrals then served as input to the quantum algorithms.

\subsection{Resource estimates}

We now evaluate the cost of phase estimation to chemical accuracy for various structures in the carbon capture catalytic cycle of~\cref{sec:catalytic_cycle}.
The active spaces of the molecules considered in this section are in the range of $52$--$65$ orbitals, although we tabulate more results for $2$--$250$ orbitals in the supplementary material.
In addition, we also evaluate the scaling of cost with respect to the active space size $N$ while keeping the number of atoms fixed, and find a different scaling law of $M\sim N^{2.5}$ for the number of eigenvectors compared to that of $M\sim N \log{N}$~\cite{Peng2017Cholesky} when increasing the number of atoms.
We find, as shown in~\cref{tab:carbon_costs}, that the gate cost for systems in the $52$-orbital to $65$-orbital range is on the order of $ 10^{10}$  Toffoli gates, using about $4000$ qubits.
\begin{table*}[h!tb]
	\caption{\label{tab:carbon_costs} Number of Toffoli gates for estimating an energy level to an error of $1$~mHartree using a truncation threshold of $\epsilon_{\text{in}}=1\text{mHartree}$ for the largest active spaces of structures in the catalytic cycle considered here. Our approach allows for a trade-off between the number of logical qubits required and the Toffoli count. }
\begin{ruledtabular}
\begin{tabular}{llllllllll}
		\multirow{2}{*}{Structure} & \multirow{2}{*}{Orbitals} &\multirow{2}{*}{Electrons} & \multirow{2}{*}{$R$} &\multirow{2}{*}{$M$} & $\alpha_{\text{DF}}$ & \multicolumn{2}{l}{Using fewer qubits} & \multicolumn{2}{l}{Using fewer Toffolis} \\
		&&&&& $/ \text{Hartree}$& Qubits & Toffolis$/10^{10}$ & Qubits & Toffolis$/10^{10}$\\
		\hline
		I& 52& 48& 613& 23566& 177.3& 3400& 1.3 & 6900 & 1.1\\
		II& 62& 70&734& 33629& 374.4& 4200& 3.6 & 8400 & 3.1 \\
		II-III& 65& 74& 783& 38122& 416.0& 4400& 4.5 & 8900 & 3.7\\
		V& 60& 68&670& 29319& 371.1& 4100& 3.3 & 8200 & 2.9\\
		VIII& 65& 76&794& 39088& 425.7& 4400& 4.6 & 8900 & 3.8 \\
		VIII-IX& 59&72& 666& 29286& 384.4& 4000& 3.4 & 8000 & 2.9 \\
		IX& 62& 68&638& 28945& 396.6& 4200& 3.5 & 8400 & 3.1\\
		XVIII& 56& 64 & 705& 29594& 293.5& 3700& 2.5 & 7400 & 2.1
\end{tabular}
\end{ruledtabular}
\end{table*}

The cost depends on the truncation threshold $\epsilon_{\text{in}}$ as plotted in~\cref{fig:carbon_cost}.
\begin{figure*}
	\includegraphics{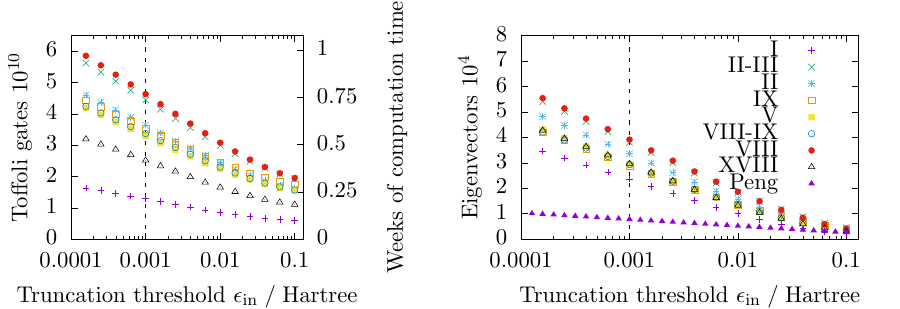}
	\caption{\label{fig:carbon_cost} (Left) Toffoli cost versus truncation threshold for phase estimation to a precision of $1$mHartree in the double-factorized representation of the catalytic cycle in~\cref{cycle} with a number of orbitals listed in~\cref{tab:carbon_costs}. Computation time is assumed to be 10$\mu$s between each Toffoli, or $100\text{kHz}$. (Right) Number of eigenvector in the double-factorized representation. The data points by Peng et al. are for three-dimensional hydrocarbons~\cite{Peng2017Cholesky} with $54$ orbitals.}
\end{figure*}
At the highly aggressive threshold of $\epsilon_{\text{in}}=100\text{mHartree}$, the rank of these examples roughly matches the parameters fit by Peng et al.~\cite{Peng2017Cholesky} to three-dimensional hydrocarbons for $54$ orbitals.
However, the numerically computed shift in energy at that threshold can exceed chemical accuracy, following~\cref{II-III-smallAS-truncation}, and should be interpreted as a most optimistic cost estimate.
Therefore, we choose  $\epsilon_{\text{in}}=1\text{mHartree}$, which our numerical simulations indicate closely reflect chemical accuracy.
Irrespective of the truncation scheme, we note that the logarithmic scaling with $1/\epsilon_{\text{in}}$ means that the Toffoli costs vary by at most factor of five between these extremes.
By varying the active space size between $2$ and $250$ orbitals for the different catalyst structures, we find in~\cref{fig:carbon_cost_vs_orbitals} that the Toffoli cost of phase estimation scales with $\sim N^{3.25}$.
The exponent is a combination of two factors: Across all configurations of the catalytic cycle, the number of eigenvectors scales with $M\sim N^{2.5}$, hence $c_W\sim\sqrt{MN}\sim N^{1.75}$, and the normalizing constant $\alpha_{\text{DF}}\sim N^{1.5}$\;Hartree.
\begin{figure*}
	\includegraphics{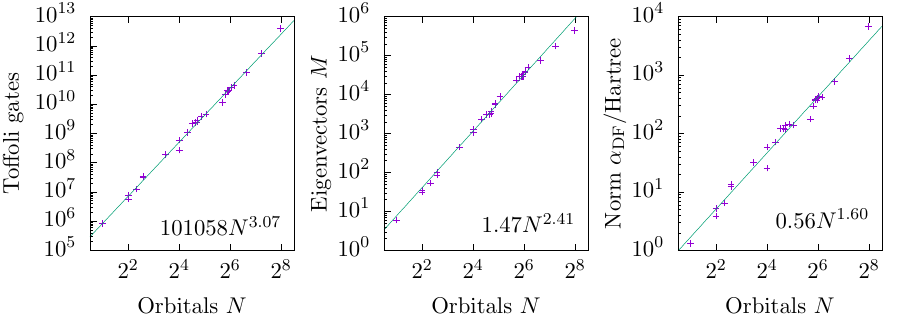}
	\caption{\label{fig:carbon_cost_vs_orbitals} (Left) Scaling of minimized Toffoli cost across all configurations of the catalytic cycle with respect to active space size, for performing phase estimation to a precision of $1$mHartree in the double-factorized representation. The fitted power-law curve has an exponent arising from the product $\alpha_{\text{DF}}c_W\sim\alpha_{\text{DF}}\sqrt{MN}/\Delta_{E}$ in~\cref{eq:QPE_repetitions,eq:Toffoli_cost} (middle) where $M$ is the number of eigenvectors and (right) the normalizing constant is $\alpha_{\text{DF}}$.}
\end{figure*}

We next demonstrate the extent of our algorithmic improvements by applying our techniques to the $54$-orbital representation of the FeMoco active site of nitrogenase that we considered previously~\cite{reiher17}, and comparing costs with prior art based on Trotterization and qubitization of the single-factorized representation.
As seen in~\cref{tab:Nitrogenase}, our Toffoli cost following~\cref{eq:Toffoli_cost} is on the order of $1.22\times 10^{10}$ using $3600$ qubits.
\begin{table*}[htb]
	\caption{\label{tab:Nitrogenase} Comparison of our new double-factorization approach for $H_{\text{DF}}$ applied to the FeMoco active site of nitrogenase ($N=54$) with prior approaches based on Trotterization~\cite{reiher17} or qubitization~\cite{Berry2019CholeskyQubitization} using the unfactorized $H$ or single-factorized $H_{\text{CD}}$ Hamiltonian, and also for the VIII structure in the catalytic cycle ($N=65$) where all examples apply the incoherent truncation scheme with the same  threshold of $\epsilon_{\text{in}}=1$mHartree.
}
\begin{ruledtabular}
\begin{tabular}{llllllll}
		Structure & Approach   & $\alpha$ / Hartree & Terms & Qubits & Toffoli gates & Comments
		\\
		\hline
		\multirow{5}{*}{FeMoco} &Qubitization $H_{\text{DF}}$  & 300.5 & $1.3\times 10^6$  & 3600  & $2.3\times 10^{10}$ & $\epsilon_{\text{in}}=1\text{mHartree}$.
		\\
		&Qubitization $H_{\text{DF}}$   & 296.9  & $2.8\times 10^5$ & 3600 & $1.22\times 10^{10}$ & Optimistic $\epsilon_{\text{in}}=73\text{mHartree}$.
		\\
		&Trotterization $H$~\cite{reiher17} & - & - &  142 & $1.5\times 10^{14}$ & Optimistic Trotter number.
		\\
		&Qubitization  $H$~\cite{Berry2019CholeskyQubitization} & $9.9\times 10^3$ & $4.4\times 10^5$  & 5100 & $2.3\times 10^{11}$ & Truncation evaluated by CCSD.
		\\
		&Qubitization  $H_{\text{CD}}$~\cite{Berry2019CholeskyQubitization} & $3.6\times 10^4$ & $4.0\times 10^5$  & 3000 & $1.2\times 10^{12}$ & Truncation evaluated by CCSD.
		\\
		\hline
		\multirow{3}{*}{VIII}&Qubitization $H_{\text{DF}}$  &425.7 & $2.5\times 10^6$ & 4600 & $4.6\times 10^{10}$ & $\epsilon_{\text{in}}=1\text{mHartree}$.
		\\
		&Qubitization $H$  & $1.1\times 10^4$ & $2.2\times 10^6$& 11000 & $9.3\times 10^{11}$ & $\epsilon_{\text{in}}=1\text{mHartree}$.
		\\
		&Qubitization $H_{\text{CD}}$ & $4.2\times 10^4$ & $1.3\times 10^{6}$ &  5800 & $2.1\times 10^{12}$ & $\epsilon_{\text{in}}=1\text{mHartree}$.
\end{tabular}
\end{ruledtabular}

\end{table*}
Assuming that the number of logical qubits is not a limiting factor, this is a dramatic improvement over the $6.0\times 10^{14}$ $T$ gates and $142$ qubits in our original estimate~\cite{reiher17},
and a significant improvement over the $1.2\times 10^{12}$ Toffoli gates of the single-factorized approach~\cite{Berry2015Truncated}, where $\alpha_{\text{CD}}=3.6\times 10^4$ Hartree with $3.0\times 10^5$ unique non-zero terms, and a rank of $R\sim 200$  was claimed to be sufficient to achieve chemical accuracy with respect to a CCSD ansatz.
This is seen in~\cref{tab:Nitrogenase_threshold} where, for the sake of comparison, we choose a truncation threshold of $\epsilon_{\text{in}}=73\text{mHartree}$  to normalize the rank between our results and the single-factorized approach.
Moreover, our double-factorized approach also improves upon the $2.3\times 10^{11}$ Toffoli gates of highly-optimized implementations~\cite{Berry2015Truncated} based on the unfactorized Hamiltonian~\cref{eq:hamiltonian_unfactorized}, which has $\alpha=9.9\times 10^3$ Hartree with $4.4\times 10^5$ unique non-zero terms.
Our improvement largely stems from a normalizing factor $\alpha_{\text{DF}}$ that is $33$ to $120$ times smaller as the number of terms in all approaches at this threshold are roughly equal.
This trend also applies to the various carbon fixation catalyst structures that we consider when we perform a similar comparison but instead use the same incoherent truncation scheme for all examples at a more conservative error threshold of $\epsilon_{\text{in}}=1$mHartree.

\begin{table*}[h!tb]
\caption{\label{tab:Nitrogenase_threshold} Scaling of cost in our double-factorization approach with truncation threshold for the FeMoco active site of nitrogenase ($N=54$). For comparison, the last line has $R=200$ which matches that used Berry et al.~\cite{Berry2019CholeskyQubitization}.}
\begin{center}
\begin{ruledtabular}
\begin{tabular}{lllllll}
		\multirow{2}{*}{$\epsilon_{\text{in}}\;/\; \text{mHartree}$} & Rank & Eigenvectors & Terms& \multirow{2}{*}{$\alpha_{\text{DF}} \;/\; \text{Hartree}$} & \multirow{2}{*}{Qubits} & $\#$Toffoli \\
		&  $R$  & $M$ &$M\times N$ & & & gates\\
		\hline
		1 &	567	&$2.4\times 10^{4}$&$1.30\times 10^6$&	 300.5 & $3600$ & $2.3\times 10^{10}$\\
		10 &	371&$1.33\times 10^{4}$&$7.2\times 10^5$& 300.0	& $3600$ & $1.67\times 10^{10}$ \\
		100 &	178&$4.2\times 10^{3}$&$2.3\times 10^5$&	 295.8	& $3600$ & $1.16\times 10^{10}$\\
		\hline
		73 &	200&$5.2\times 10^{3}$&$2.8\times 10^5$&	 296.9	& $3600$ & $1.22\times 10^{10}$\\
\end{tabular}
\end{ruledtabular}
\end{center}
\end{table*}


\subsection{Runtimes and qubit counts}

We finally relate these gate count estimates to expected runtimes on future quantum computers. Exact runtime will, of course, depend on details of the system and error correction schemes. In our previous analysis~\cite{reiher17}, we optimistically assumed gate times of 100ns for fault tolerant logical gate operations, which may be a long-term achievable goal.  As a more realistic assumption for mid-term fault tolerant quantum computers we now expect that the physical gate times in current quantum computer architectures range from tens of nanoseconds for solid state qubits to tens of microseconds for ion traps. Realizing fault tolerance by using quantum error correction with the surface code~\cite{fowler12}, will lead to logical gate times for a Toffoli gate of about 10$\mu$s to $10$ms depending on the architecture. The lower estimate of 10$\mu$s means that $10^{10}$ Toffoli gates correspond to a runtime of 28 hours or about a day, while the upper estimate of 10ms would correspond to several years. These considerations show that fast gate times will be essential for realistic quantum computation for chemical catalysis.

The above estimates do not explicitly consider the cost of layout, i.e. the mapping onto a nearest neighbor planar square lattice topology of error corrected logical qubits. We argue that this overhead is negligible, since the subroutine with dominant cost, the table lookup discussed in the supplementary material is based on \textsc{Fanout} operators~\cite{Pham2013Fanout} and maps well to this topology. Moreover, we here assume that any overhead of layout as well as Clifford gates are included in the assumed gate time for a Toffoli gate, as these dominant \textsc{Fanout} operations may be implemented in a constant Clifford depth of $4$ in parallel with the sequentially applied Toffolis.\\
Fault tolerant gates also have an overhead in the number of qubits, with hundreds to thousands of physical qubits needed per logical qubit~\cite{fowler12,reiher17}, depending on the quality of the qubits. 4000 logical qubits will thus correspond to millions of physical qubits. This implies the need for a scalable quantum computer architecture, scaling to millions of qubits.

\subsection{State preparation}
To determine the ground state energy using phase estimation, it is required to prepare a trial state $\ket{\psi_{\text{trial}}}$ which has a high overlap with the true ground state $\ket{\psi_0}$ of the Hamiltonian $H$ as discussed in~\cref{sec:quantum_algorithms_for_chemistry}.
While the exact ground state is unknown, we used state-of-the-art DMRG calculations to obtain an approximate ground state $\ket{\tilde{\psi_0}}$ for each of our systems (see supplementary material for further details). An approximate configuration interaction wave function obtained by reconstruction~\cite{boguslawski11}
from the corresponding matrix product state wave function optimized with DMRG
served to provide an overlap with the trial state prepared on the quantum computer. Since the trial state is chosen to be HF determinant
for this case, the overlap is given by the square of the coefficient in front of the Hartree-Fock determinant in the configuration interaction
expansion, i.e., the one with the first $N/2$ orbitals occupied where $N$ is the number of electrons.
We found that there is a large overlap for all of our systems, see~\cref{tab:Overlap}. We expect this overlap to not shrink substantially with increasing precision of the approximate ground state obtained by DMRG and hence preparing the dominant single-determinant state is sufficient for the molecular structures considered in this work.
Recall that the Ru complexes selected for our study do not exhibit strong multi-reference character. If this overlap had turned out to be small, we would defer to Ref.~\cite{tubman2018postponing} for options on how to prepare a multi-determinant initial state in order to boost the success probability of phase estimation.

\begin{table*}[htb]
\caption{\label{tab:Overlap} Overlap $|\braket{\tilde{\psi_0}}{\psi_{\text{trial}}}|^2$ of the dominant single-determinant state $\ket{\psi_{\text{trial}} }$  with the approximate ground state $\tilde{\ket{\psi_0}}$ obtained with DMRG for the complexes considered in this work. The type of molecular orbital (MO) basis is given as well as the size of active space of the DMRG calculation in terms of orbitals and electrons.}
\begin{center}
\begin{ruledtabular}
\begin{tabular}{lllll}
		Complex & MO basis & Orbitals & Electrons &  $|\braket{\tilde{\psi_0}}{\psi_{\text{trial}}}|^2$ \\\hline
                I  & CAS(4,5)SCF & 52 & 48 & $0.847$ \\
                II & CAS(8,6)SCF & 62 & 70 & $0.848$ \\
                II-III & CAS(8,6)SCF & 65 & 74 & $0.848$ \\
                V & CAS(12,11)SCF & 60 & 68 & $0.841$ \\
                VIII & CAS(2,2)SCF & 65 & 76 & $0.869$ \\
		        VIII-IX & CAS(4,4)SCF & 59 & 72 & $0.863$ \\
				IX & CAS(16,16)SCF & 62 & 68 & $0.807$ \\
                XVIII & CAS(4,4)SCF & 56 & 64 & $0.889$ \\
\end{tabular}
\end{ruledtabular}
\end{center}
\end{table*}

\section{Conclusions}
In this work, we considered computational catalysis leveraged by quantum computing. We rely on accurate error bounds on
the electronic energy accessible in quantum algorithms to obtain sufficiently accurate data for intermediate and
transition state structures of a catalytic cycle. In particular,
i) we considered decisive steps of a synthetic catalyst that is known to convert carbon dioxide to methanol (and for which
traditional work based on density functional theory had already been reported alongside the experimental results~\cite{wesselbaum15}),
ii) we presented a new quantum algorithm in the qubitization framework that exploits a double-factorized electronic structure representation, which significantly reduces the runtime,
iii) we validated the truncation schemes by comparison with density matrix renormalization group calculations,
iv) we confirmed that starting the quantum algorithm from a single determinant initial state has high success probability for the carbon dioxide functionalization process studied in this work, and
v) we calculated realistic resource estimates for mid-term scalable quantum computers and related them to expected runtimes.

A pressing challenge in the cost of quantum simulation for large molecular systems is the rapid growth in the number of four-index two-electron integrals.
By using integral decomposition techniques in the double-factorized representation~\cite{beebe77,aquilante11a}, we simultaneously minimize two key parameters governing the cost of qubitization: the number of coefficients that must be loaded into the quantum computer and a certain bound $\alpha_{\text{DF}}$ on the spectral norm of the Hamiltonian.
Compared to other simulation techniques such as Trotterization, qubitization also enables a large trade-off in non-Clifford gate count by using additional qubits.
In the case of estimating an energy level to chemical accuracy, we found that our approach has a Toffoli gate cost with respect to active space size $N$, while keeping the number of atoms fixed, that scales like $\sim N^{3.25}$.
When the number of atoms also grows with $N$, our technique promises a Toffoli cost scaling like $\sim \alpha_{\text{DF}}N\log{(N)}$, assuming the empirically fitted sparsity of double-factorization by other authors~\cite{Peng2017Cholesky}, and where $\alpha_{\text{DF}}\sim N^{1.5}$ based on Hydrogen chain benchmarks.
When applied to structures in the catalytic cycle with $52$--$65$ orbitals, our approach requires roughly $10^{10}$--$10^{11}$ Toffoli gates and $\sim 4000$ logical qubits.
Note that the benefit of our reduced Toffoli count can far outweigh the cost of using more logical qubits --
the physical qubit count of the overall algorithm is typically dominated by that required for the fault-tolerant distillation of non-Clifford gates.
Under realistic assumptions on the performance of mid-term quantum hardware, this corresponds to a runtime of a few weeks for a single calculation. While this seems similar to the results of our previous paper~\cite{reiher17}, we here use much more conservative and realistic estimates for the clock speeds of mid-term quantum computers, while Ref.~\cite{reiher17} assumed ambitious specifications for the long term.

From our results it is evident that a future universal quantum computer can only lead the range of electronic structure
methods if further developments provide faster algorithms that can treat much larger active spaces in order to be competitive.
We emphasize that, although no logical qubit has been realized experimentally so far, a universal quantum computer
that can represent a state on a few thousand logical qubits would revolutionize electronic structure theory as it bears the
potential to provide an accurate exact-diagonalization energy for a moderately sized molecular system
of on the order of 100 atoms
in the full single-particle basis set. As a consequence, any necessarily approximate a-priori or a-posteriori correction for the
nagging dynamic correlation problem
(that arises solely from the choice of a reduced-dimensional active space) would no longer be needed.
It is clear that such a calculation would be an ultimate goal, but also that it would present new technical challenges.

Going to an order of magnitude larger active space sizes as considered here would increase the quantum computational requirements by three orders of magnitude.
Hence, despite the advances reported in this work, it is obvious that further algorithmic improvements of quantum algorithms
are urgently needed in order to get timings down so that a quantum computer can become superior to traditional approaches, i.e.,
to eventually demonstrate a quantum advantage in real-world chemistry applications. Given the innovation rate in quantum algorithms for chemistry of many orders of magnitude over the past years, we are confident for this to happen. One  direction is exploring novel sparse representations. Another promising avenue is exploiting the fact that the number of electrons may grow much smaller than the size of the single-particle basis set when expanding the latter to incorporate dynamical correlations. That can lead to significant savings compared to the algorithms presented here.

\section{Comments on the integral files}
The integral files used in this work are available upon request from the corresponding author and will be made publicly available upon publication.
\begin{acknowledgments}
We are grateful to Markus H\"olscher for valuable discussions on the catalyst synthesized in the Leitner lab and for providing us with the correct structure of complex I and to Ali Alavi for sharing his insights into the FCIQMC algorithm with us. VvB and MR gratefully acknowledge financial support by the Swiss National Science Foundation (project no. 200021 182400). 
We also thank Dominic Berry, Garnet Chan, Joonho Lee, Hongbin Liu, and Mathias Soeken for comments on the manuscript.
\end{acknowledgments}



\end{document}


\title{Supplementary Material for \textit{Quantum Computing Enhanced Computational Catalysis}}
	\author{Vera von Burg}
	\affiliation{Laboratorium f\"ur Physikalische Chemie, ETH Z\"urich, Vladimir-Prelog-Weg 2, 8093
		Zürich, Switzerland}
	\author{Guang Hao Low}
	\affiliation{\Microsoft}
	\author{Thomas H\"aner}
	\author{Damian S. Steiger}
	\affiliation{Microsoft Quantum, 8038 Z\"urich, Switzerland}
	\author{\\Markus Reiher}
	\affiliation{Laboratorium f\"ur Physikalische Chemie, ETH Z\"urich, Vladimir-Prelog-Weg 2, 8093 Zürich, Switzerland}
	\author{Martin Roetteler}
	\author{Matthias Troyer}
	\affiliation{\Microsoft}
	\maketitle
	\tableofcontents
	
	\newpage
	\section{General notes on the [Ru] catalyst}
	We chose seven key intermediates and transition states from the supporting information of Wesselbaum et al.~\cite{wesselbaum15}, therein labeled as complex structures II, II-III, V, VIII, VIII-IX, IX, and XVIII. These structures had been optimized with the M06-L/def2-SVP combination of exchange-correlation density functional and basis set, and single-point M06-L/def2-TZVP energies were obtained~\cite{wesselbaum15} which we used for the diagram in the main text. In this work, we continue to use the roman numerals as labels. Hyphenated labels (i.e., II-III and VIII-IX) correspond to transition states whereas the others are stable intermediates. The Cartesian coordinates of intermediate I were erroneous in the supporting information of the original paper and we obtained the correct ones through a private communication~\cite{holscher19} (see~\cref{sec:coords-I}). All of the complexes are monocations and were considered in the lowest-energy singlet state.
	
	\section{DFT calculations}
	
	To be able to obtain relative reaction energies (see~\cref{tab:dftEFormula}), we optimized the small molecules CO\textsubscript{2}, H\textsubscript{2}, H\textsubscript{2}O, THF, and methanol with \textsc{Gaussian09}, revision D.01 including density fitting and the \textit{UltraFine} keyword for the integration grid in accordance with the supporting information of Ref.~\cite{wesselbaum15}. For all other density functional theory (DFT) calculations reported herein we employed \textsc{Turbomole}~\cite{ahlrichs89}, version 7.0.2. We then carried out single-point, unrestricted PBE~\cite{perdew96} and PBE0~\cite{perdew96b} calculations in the def2-TZVP~\cite{weigend05} basis for each complex and the small molecules on the M06-L/def2-SVP structures reported in the supporting information of Ref.~\cite{wesselbaum15} or the newly obtained ones in the case of the small molecules, respectively. To ensure that we obtained the correct global minima in the molecular orbital coefficient parameter space, we perturbed the $\alpha$ and $\beta$ orbitals of the converged calculation with the orbital steering protocol of Ref.~\cite{vaucher17}. We repeated the respective calculation with the perturbed orbitals as starting orbitals to probe whether a solution of lower energy could be obtained.

	We additionally optimized the structures with the
	PBE/def2-TZVP combination of density functional and basis set and carried out frequency calculations to ensure we had obtained the correct type of stationary points (no imaginary frequencies for intermediates, one for transition states).
	\subsection{Single-point M06-L energies for small molecules}
	For comparison, we report here the M06-L/def2-SVP electronic energies for the small molecules which we re-optimized.
	\begin{table}[htb]
		\caption{M06-L/def2-SVP electronic energies for the optimized structures given in Hartree atomic units.}
		\begin{tabular}{l @{\qquad} r}
			\hline \hline
			Compound & $E^{\mathrm{el}}_{\mathrm{M06-L}}$\\\colrule
			H$_2$ & -1.16721 \\
			H$_2$O & -76.35053 \\
			CO$_2$ &  -188.43534 \\
			CH$_3$OH &  -115.61382 \\
			THF & -232.24308 \\
			\botrule
		\end{tabular}
	\end{table}
	
	\newpage
	\subsection{Single-point PBE and PBE0 electronic energies and PBE electronic energies for fully optimized structures}
	The electronic energies of the unrestricted, single-point PBE/def2-TZVP and PBE0/def2-TZVP DFT calculations for the complexes and remaining compounds as well as the relative reaction energies of the complexes are collected in~\cref{tab:dftEnergies-sp} and the ones of the re-optimized structures in~\cref{tab:dftEnergies-opt}. The relative reaction energies reported in these Tables were obtained in accordance with the supporting information of Ref.~\cite{wesselbaum15} and their formulaic calculation is given in~\cref{tab:dftEFormula}. We employ the double slash notation where the labels before the double slash denote the combination of exchange-correlation functional and basis set for the single-point calculation whereas the exchange-density functional and basis set employed for the structure optimization is given after the double slash.
	
	\begin{table}[htb]
		\caption{\label{tab:dftEFormula} Formulae for the calculation of relative reaction energies $\Delta E_{\mathrm{rel}}^{\mathrm{el}}$ for the complexes, where $E^{\mathrm{el}}(i)$ is the electronic energy of compound $i$ computed with a certain combination of exchange-correlation density functional and basis set.}
		\begin{ruledtabular}
			\begin{tabular}{ll}
				Complex & Calculation of $\Delta E_{\mathrm{rel}}^{\mathrm{el}}$ \\\hline
				I & 0.0 \\
				II & $E^{\mathrm{el}}(\mathrm{II}) + E^{\mathrm{el}}(\mathrm{THF}) - E^{\mathrm{el}}(\mathrm{I})  - E^{\mathrm{el}}(\mathrm{CO_2})$ \\
				II-III & $E^{\mathrm{el}}$(II-III) + $E^{\mathrm{el}}(\mathrm{THF}) - E^{\mathrm{el}}(\mathrm{I})  - E^{\mathrm{el}}(\mathrm{CO_2})$ \\
				V & $E^{\mathrm{el}}(\mathrm{V}) + E^{\mathrm{el}}(\mathrm{H}_2) - E^{\mathrm{el}}(\mathrm{I})  - E^{\mathrm{el}}(\mathrm{CO_2})  $ \\
				VIII & $E^{\mathrm{el}}(\mathrm{VIII}) + E^{\mathrm{el}}(\mathrm{THF}) - E^{\mathrm{el}}(\mathrm{H}_2) - E^{\mathrm{el}}(\mathrm{I})  - E^{\mathrm{el}}(\mathrm{CO_2})  $ \\
				VIII-IX & $E^{\mathrm{el}}$(VIII-IX) + $E^{\mathrm{el}}(\mathrm{THF}) - E^{\mathrm{el}}(\mathrm{I})  - E^{\mathrm{el}}(\mathrm{CO_2})  $ \\
				IX & $E^{\mathrm{el}}(\mathrm{IX}) + E^{\mathrm{el}}(\mathrm{THF}) - E^{\mathrm{el}}(\mathrm{I})  - E^{\mathrm{el}}(\mathrm{CO_2})$ \\
				XVIII & $E^{\mathrm{el}}(\mathrm{XVIII}) + E^{\mathrm{el}}(\mathrm{H_2O}) - 2 \cdot E^{\mathrm{el}}(\mathrm{H_2})   - E^{\mathrm{el}}(\mathrm{I})  - E^{\mathrm{el}}(\mathrm{CO_2})$ \\
			\end{tabular}
		\end{ruledtabular}
	\end{table}

	\begin{table}[htb]
		\centering
		\caption{PBE/def2-TZVP//M06-L/def2-SVP and PBE0/def2-TZVP//M06-L/def2-SVP electronic single-point energies $E^{\mathrm{el}}$ and relative reaction energies $\Delta E_{\mathrm{rel}}^{\mathrm{el}}$ given in Hartree atomic units.  \label{tab:dftEnergies-sp}}
		\begin{ruledtabular}
			\begin{tabular}{ddddd}
				\text{Compound} & \text{E}^{\mathrm{el}}_{\mathrm{PBE}} & \Delta \text{E}_{\mathrm{rel}}^{\mathrm{el,PBE}}  &  \text{E}^{\mathrm{el}}_{\mathrm{PBE0}} & \Delta \text{E}_{\mathrm{rel}}^{\mathrm{el,PBE0}}  \\\hline
				\text{I} & -2936.85035 & \multicolumn{1}{c}{0.0} & -2937.09522& \multicolumn{1}{c}{0.0} \\
				\text{II} & -2893.07375 & 0.01824933 & -2893.27797& 0.01845670 \\
				\text{II-III} & -2893.06593 & 0.02630769 & -2893.26673 & 0.02969426 \\
				\text{V} &  -3124.16485 & -0.00142320  &-3124.39830 & -0.00479493 \\
				\text{VIII} &  -2894.26852 & -0.01061211&-2894.47992 &  -0.01533001 \\
				\text{VIII-IX} &   -2893.06125&  0.03075055 & -2893.26757 & 0.02886085 \\
				\text{IX} & -2893.08016 & 0.01184458 &-2893.28802 & 0.00840579  \\
				\text{XVIII} &   -3051.28933 & -0.00472200 &-3051.53350 & -0.01268354 \\
				\text{H}_2 &  -1.16590&  \multicolumn{1}{c}{ --- } &-1.16817 & \multicolumn{1}{c}{ --- } \\
				\text{H}_2\text{O} & -76.37653 & \multicolumn{1}{c}{ --- }  &-76.37717 & \multicolumn{1}{c}{ --- } \\
				\text{CO}_2 & -188.47898 & \multicolumn{1}{c}{ --- } &-188.46644 & \multicolumn{1}{c}{ --- } \\
				\text{CH}_3\text{OH} &  -115.62907 & \multicolumn{1}{c}{ --- }  &-115.63697 & \multicolumn{1}{c}{ --- } \\
				\text{THF} & -232.23733& \multicolumn{1}{c}{ --- } &-232.26523 & \multicolumn{1}{c}{ --- } \\
			\end{tabular}
		\end{ruledtabular}
	\end{table}

	\begin{table}[htb]
		\caption{PBE/def2-TZVP//PBE/def2-TZVP electronic energies $E^{\mathrm{el,opt}}$ after structure optimization and corresponding relative reaction energies $\Delta E_{\mathrm{rel}}^{\mathrm{el,opt}}$ given in Hartree atomic units. \label{tab:dftEnergies-opt}}
		\begin{ruledtabular}
			\begin{tabular}{ddd}
				\text{Compound} & \text{E}^{\mathrm{el,opt}}_{\mathrm{PBE}} & \Delta \text{E}_{\mathrm{rel}}^{\mathrm{el,PBE,opt}}  \\\hline
				\text{I} &  -2936.85482 & \multicolumn{1}{c}{0.0} \\
				\text{II} & -2893.07854 &  0.01743994 \\
				\text{II-III} & -2893.06962 & 0.02635768 \\
				\text{V} & -3124.17151 & -0.00331186  \\
				\text{VIII} & -2894.27315  & -0.01126172 \\
				\text{VIII-IX} & -2893.06751  & 0.02846465 \\
				\text{IX} & -2893.08609 & 0.00988447 \\
				\text{XVIII} &  -3051.29506  &  -0.00590696 \\
				\text{H}_2 & -1.16591 &  \multicolumn{1}{c}{ --- }  \\
				\text{H}_2\text{O} & -76.37676 & \multicolumn{1}{c}{ --- }  \\
				\text{CO}_2 & -188.47928 & \multicolumn{1}{c}{ --- } \\
				\text{CH}_3\text{OH} & -115.62962  & \multicolumn{1}{c}{ --- }   \\
				\text{THF} & -232.23812 & \multicolumn{1}{c}{ --- } \\
			\end{tabular}
		\end{ruledtabular}
	\end{table}

	\clearpage
	\section{HF and post-HF calculations}\label{sec:postHFCompDet}
	We obtained Hartree-Fock (HF) and Complete Active Space Self-Consistent Field (CAS-SCF)~\cite{roos80,werner85,knowles85,ruedenberg82} molecular orbitals (MOs) in an ANO-RCC-VTZP~\cite{widmark90,roos04} atomic orbital (AO) basis for the light elements and ANO-RCC-VQZP~\cite{roos05} for Ruthenium with \textsc{OpenMolcas}~\cite{fdez.galvan19}.
	We employed a Cholesky Decomposition (CD) of the two-electron repulsion integrals and generated two sets of integrals with decomposition thresholds of $10^{-4}$ and $10^{-8}$, respectively (note that this decomposition which occurs during the quantum chemical calculations is separate from a later truncation of the two-electron integrals in the context of the quantum computing algorithms). To distinguish the two sets of integrals, those that were obtained with a threshold of $10^{-8}$ are labeled as ``highCD'' (e.g., I-highCD-cas5-fb-48e52o). We observed that the choice of this threshold has a negligible effect on the resource estimates, as seen in~\cref{tab:RE-1,tab:RE-2}. The choices for the molecular and atomic orbital bases in this study are summarized in~\cref{tab:bases}.
	
	We performed exploratory calculations in a so-called minimal basis (mb) (ANO-RCC-MB~\cite{widmark90,roos04,roos05} for this study). For quantitative results, it is necessary to employ a much larger AO basis, in this study, the ANO-RCC-VTZP for light elements and ANO-RCC-VQZP for Ruthenium, which we will term the full atomic orbital basis (fb). \cref{tab:NumFunctions} lists the number of one-electron basis functions resulting from the choice of these bases for each complex. Note that the number of MOs equals the number of AOs and the numbers in~\cref{tab:NumFunctions} therefore apply to both types of one-electron basis sets.
	
	The MOs are generally labeled by the type of method with which they have been obtained, i.e. HF and CASSCF. When necessary, the active space employed in the CASSCF calculation is given explicitly through the notation ($N$,$L$) (e.g., CAS(6,6)SCF), where $N$ is the number of electrons and $L$ the number of orbitals in the active space.
	\begin{table}[htb]
		\caption{Overview of different types of one-electron basis sets employed in this study.\label{tab:bases}}
		\begin{ruledtabular}
			\begin{tabular}{lll}
				Basis & Basis set  & Comment\\\hline
				Atomic orbitals & ANO-RCC-MB &  Abbreviation: mb  \\ &$\left.
				\begin{tabular}{@{}l l} ANO-RCC-VTZP for light elements\\ ANO-RCC-VQZP for
				Ruthenium \end{tabular}\right\}$ & \multirow{1}{*}{Abbreviation: fb} \\
				Molecular orbitals & HF  &   \\  & CAS($N$,$L$)SCF & ($N$,$L$): Active space
				of \\ & & $N$ electrons in $L$ orbitals \\
			\end{tabular}
		\end{ruledtabular}
	\end{table}

	\begin{table}[ht]
		\caption{Number of one-electron basis functions for each complex for the two choices of atomic orbital bases. `mb' and `fb' denote the minimal and full atomic orbital basis, respectively, as described in the text and in~\cref{tab:bases}.\label{tab:NumFunctions}}
		\begin{ruledtabular}
			\begin{tabular}{lcc}
				& \multicolumn{2}{c}{Number of basis functions }  \\
				Complex & mb & fb\\\hline
				I & 334 & 2286\\
				II & 316 & 2114\\
				II-III & 316 & 2114\\
				V & 347 & 2348\\
				VIII & 318 & 2142\\
				VIII-IX & 316 & 2114\\
				IX & 316 & 2114\\
				XVIII & 346 & 2374 \\
			\end{tabular}
		\end{ruledtabular}
	\end{table}
	
	\newpage
	\subsection{HF, CASSCF, and DMRG-CI electronic energies}
	The CASSCF molecular orbitals from calculations with a Cholesky Decomposition threshold of the two-electron integrals of $10^{-4}$ were obtained in the following manner: HF orbitals in the ANO-RCC-MB atomic orbital basis were split-localized~\cite{olivares-amaya15} with the Pipek-Mezey method~\cite{pipek89}. From these localized orbitals, we selected the orbitals corresponding to the $4d$ orbitals of Ruthenium, the $2p$ and $2s$ orbitals of Oxygen and Carbon atoms (i.e. the bonding and antibonding $\pi$ orbitals of carbon dioxide and derivatives), the bonding and antibonding $\sigma$ orbitals of H\textsubscript{2} as well as the $s$-orbitals of Hydrogen atoms to evaluate their orbital entanglement and pair-orbital mutual information in an approximate Density Matrix Renormalization Group-Configuration Interaction (DMRG-CI) calculation with maximum bond order $m=800$ and $n=5$ with the \textsc{QCMaquis}~\cite{keller15a} program. We selected the orbitals corresponding to a threshold in \textsc{autoCAS}~\cite{stein16,stein19} as the active space for a subsequent CASSCF calculation. We repeated the approximate DMRG calculation on these new CASSCF orbitals and chose those orbitals selected by the \textsc{autoCAS} program as the final active orbital space. We then expanded these orbitals to the full atomic orbital basis with the \textsc{EXPBAS} module of \textsc{OpenMolcas} and re-optimized them in a final CASSCF calculation.
	For the calculations involving a Cholesky Decomposition threshold of $10^{-8}$, we started the CASSCF calculations directly from the corresponding CASSCF orbital file of the calculations with a threshold of $10^{-4}$ to save computational resources. 
	\newline
	The one- and two-electron integrals of a subset of these orbitals then served as parameters for the Coulomb Hamiltonian, for which the quantum-algorithm resource estimates were obtained. The general procedure for the selection of these active spaces is detailed in~\cref{subsec:ASSelection} whereas the molecular orbitals resulting from this selection for each complex are reproduced in~\cref{sec:ElStruc}.\newline

	\begin{table}[htb]
		\caption{HF, CAS($N_c$,$L_c$)SCF, and DMRG($N_d,L_d$)CI electronic energies in Hartree for the complexes in the full atomic orbital basis obtained with different values of the Cholesky Decomposition threshold. ($N_c$,$L_c$) denotes the active space of the CASSCF calculations and ($N_d,L_d$) the one of the DMRG-CI calculations. For these, $N_c$ and $N_d$ refer to the number of electrons and $L_c$ and $L_d$ to the number of orbitals, respectively.\label{tab:postHFEnergies}}
		\begin{ruledtabular}
			\begin{tabular}{lllllll}
				Complex  & CD & HF & CAS($N_c$,$L_c$)SCF & ($N_c$,$L_c$) & DMRG($N_d,L_d$)CI & $(N_d,L_d)$ \\
				& threshold & Energy & Energy &  & Energy & \\\hline
				I 	& $10^{-4}$ &-7361.315677 & -7361.357885 & (4,5)&  -7361.461381 & (48,52)\\
				I     & $10^{-8}$ &  \multicolumn{1}{c}{ --- }   & -7361.360402 & (4,5)&  \multicolumn{1}{c}{ --- }  & \multicolumn{1}{c}{ --- }\\
				II & $10^{-4}$ & -7317.956458 & -7318.035225 & (8,6) & -7318.123548   & (70,62)\\
				II & $10^{-8}$ &  \multicolumn{1}{c}{ --- }  & -7318.037547 & (8,6) &   \multicolumn{1}{c}{ --- }    & \multicolumn{1}{c}{ --- }\\
				II-III 	 & $10^{-4}$ & -7317.931739 & -7317.999907 & (8,6) & -7318.099062 & (74,65)\\
				II-III         & $10^{-8}$ &  \multicolumn{1}{c}{ --- }  & -7318.002150 & (8,6) &   \multicolumn{1}{c}{ --- }   & \multicolumn{1}{c}{ --- }\\
				V & $10^{-4}$ & -7548.043678 & -7548.214683 & (12,11) &  -7548.296446 & (68,60)	\\
				V & $10^{-8}$ &  \multicolumn{1}{c}{ --- }  & -7548.217104  & (12,11) &   \multicolumn{1}{c}{ --- }   & \multicolumn{1}{c}{ --- }    \\
				VIII 	& $10^{-4}$ & -7319.116649 & -7319.140267 & (2,2) & -7319.234732 & (76,65) \\
				VIII  & $10^{-8}$ &  \multicolumn{1}{c}{ --- }  & -7319.142596 & (2,2) &   \multicolumn{1}{c}{ --- }   & \multicolumn{1}{c}{ --- } \\
				VIII-IX & $10^{-4}$ & -7317.937509 & -7317.971210 & (4,4) &  -7318.066197 & (72,59) \\
				VIII-IX & $10^{-8}$ &  \multicolumn{1}{c}{ --- }  & -7317.973426 & (4,4) &   \multicolumn{1}{c}{ --- }   & \multicolumn{1}{c}{ --- } \\
				IX 	& $10^{-4}$ & -7317.970467 &  -7318.209829	& (16,16) & -7318.303045  & (68,62) \\
				IX    & $10^{-8}$ &  \multicolumn{1}{c}{ --- }  &    -7318.212256     & (16,16) &   \multicolumn{1}{c}{ --- }    & \multicolumn{1}{c}{ --- } \\
				XVIII 	& $10^{-4}$ &  -7475.314796 &  -7475.367378	& (4,4) &-7475.439228 &  (64,65)  \\
				XVIII         & $10^{-8}$ &  \multicolumn{1}{c}{ --- }  &    -7475.369956    & (4,4) &   \multicolumn{1}{c}{ --- }   &  \multicolumn{1}{c}{ --- } \\
			\end{tabular}
		\end{ruledtabular}
	\end{table}
	For each structure, we additionally carried out DMRG-Configuration-Interaction (CI) calculations of the orbitals corresponding to the integral file with the largest active space in order to gain qualitative information about the electronic structure. These calculations were performed with a number of sweeps $n=10$ and a maximum bond dimension $m=1000$ and Fiedler ordering and on the integral files with a CD truncation theshold of $10^{-4}$. From the matrix product states generated in this way, we obtained the largest CI-coefficient through a reconstruction of an approximate CI wave function expansion through the sampling-reconstruction algorithm of Ref.~\cite{boguslawski11}. For the latter algorithm, we employed a CI-threshold of $10^{-6}$ and a CI completeness measure of $10^{-6}$. We confirmed that the choice of the CD theshold does not have a substantial influence on these overlap values as shown in \cref{tab:overlap} for complex IX (which features the smallest overlap of all structures).
	
	\begin{table}[!htb]
		\caption{ Comparison of the overlap $|\braket{\tilde{\psi_0}}{\psi_{\text{trial}}}|^2$ of the dominant single-determinant state $\ket{\psi_{\text{trial}} }$  with the approximate ground state $\tilde{\ket{\psi_0}}$ obtained with DMRG-CI for complex IX. The active space is given by the number of electrons $N$ and orbitals $L$. The maximum bond order dimension $m$ and the number of sweeps $n$ of the parent DMRG-CI calculation is given as well. \label{tab:overlap}}
		\centering
		\begin{ruledtabular}
			\begin{tabular}{llllll}
				Catalyst & Active space & CD             & $m$ & $n$ & Overlap \\  
				structure & ($N,L$)     &  threshold     &     &     &  $|\braket{\tilde{\psi_0}}{\psi_{\text{trial}}}|^2$ \\\hline
				IX & (68,62) & $10^{-4}$ & 500 & 8 & 0.8108  \\
				IX & (68,62) & $10^{-8}$ & 500 & 8 & 0.8107 \\ 
				IX & (68,62) & $10^{-4}$ & 1000 & 10  & 0.8069 \\
			\end{tabular}
		\end{ruledtabular}
	\end{table}
	
	The energies of the HF, CASSCF, and DMRG-CI calculations are reported in~\cref{tab:postHFEnergies}. Note that no HF energies are tabulated for the calculations involving a tight Cholesky Decomposition threshold since these were started directly from the corresponding CASSCF orbital file of the calculations with the default threshold of $10^{-4}$.\newline

	We studied the effect of the accuracy of the wave function on the overlap calculated after 
	applying the reconstruction algorithm by comparing with DMRG results obtained for smaller bond dimesions and
	energetical orderings of the orbitals on the lattics (we refrain from reporting the explicit total energies of these
	calculations in order not to confuse the reader with less converged energy data). 
	As expected, Fiedler ordering of the orbitals turned out to improve the energy by about
	0.6 mHartree to 0.45 Hartree when comparing calculations with the same bond order but different
	orbital ordering. For our catalyst structures, changing $m$ from a value of 500 to 1000 (with Fiedler ordering) 
	lowers the energy by 1-2 mHartree.
	With respect to the question whether a bond order of 1000 is actually sufficient, we note that
	increasing the value of $m$ from 1000 to 2048 leads to a change in energy of -1.7 mHartree for catalyst structure IX,
	whereas the overlap changes by 0.004. 
	We therefore conclude that the convergence of the DMRG calculations beyond values of 
	the bond dimension $m$ chosen for this work has a negligible effect on the qualitative structure of
	the wave function and hence on the overlap.
	
	\newpage
	\subsection{DMRG-CI electronic energies of linear chains of Hydrogen atoms, II-III, and IX for two-electron integrals at different truncation levels}
	To obtain reasonable thresholds for the truncation parameters $\epsilon_{\text{in}}$ and $\epsilon_{\text{co}}$, we carried out highly accurate DMRG calculations with \textsc{QCMaquis} on integral files where the two-electron integrals had been truncated at varying thresholds as well as the full integrals and report the resulting absolute error in the energy. We performed these calculations on linear chains of Hydrogen atoms of length 2, 4, 6, and 8 with an internuclear separation of 1.4 times the equilibrium H$_2$ bond distance (1.037 $\mathrm{\mathring{A}}$) as well as on the complexes II-III with an active space of (8,6) and IX with an active space of (16,16), respectively. For the catalyst structures, this means we employed the integral files II-III-highCD-cas6-fb-8e6o and IX-highCD-cas16-fb-16e16o of~\cref{tab:integralFiles1,tab:integralFiles2}. The HF orbitals for the integrals of the linear chains of Hydrogen atoms were obtained with \textsc{OpenMolcas} in an ANO-RCC-VDZ basis. To determine the optimal parameters for these calculations, we carried out several ones at varying values of the maximum bond dimension $m$ and number of sweeps $n$, the results of which are summarized in~\cref{tab:conv-II,tab:conv-IX}. The settings of the final production calculations, derived from these Tables, are given in~\cref{tab:trunc-settings}.

	\begin{table}[!htb]
		\caption{DMRG($8$,$6$)-CI energies and their convergence with respect to the bond dimension $m$ and the number of sweeps $n$ for catalyst structure II-III with integrals obtained with a Cholesky Decomposition threshold of $10^{-8}$. The column ``Energy change" collects, for each row, the difference in energy of the given calculation with respect to the energy of the calculation listed in the previous row. The final settings are a maximum bond dimension of 500 and 5 sweeps (second to last column). \label{tab:conv-II}}
		\begin{ruledtabular}
			\begin{tabular}{lcccc}\toprule \toprule
				Integral file &  $m$ & $n$ & DMRG($8$,$6$)-CI & Energy change\\
				&     &     & energy / Hartree &  / mHartree        \\ \hline
				II-III-highCD-cas6-8e6o & 10 & 2 & -7318.000681  & n.a. \\
				II-III-highCD-cas6-8e6o & 250 & 5 & -7318.002150 & -1.5 \\
				II-III-highCD-cas6-8e6o & 500 & 5 & -7318.002150  & 0.0 \\
				II-III-highCD-cas6-8e6o & 500 & 10 & -7318.002150  & 0.0 \\
				\bottomrule \bottomrule
			\end{tabular}
		\end{ruledtabular}
	\end{table}

	\begin{table}[!htb]
		\caption{DMRG($16$,$16$)-CI energies and their convergence with respect to the bond dimension $m$ and the number of sweeps $n$ for catalyst structure IX with integrals obtained with a Cholesky Decomposition threshold of $10^{-8}$.  The column ``Energy change" collects, for each row, the difference in energy of the given calculation with respect to the energy of the calculation listed in the previous row. The final settings are a maximum bond dimension of 2048 and 16 sweeps (second to last column).\label{tab:conv-IX}}
		\begin{ruledtabular}
			\begin{tabular}{lcccc}\toprule \toprule
				Integral file &  $m$ & $n$ & DMRG($16$,$16$)-CI & Energy change \\
				&     &     & energy / Hartree &  / mHartree     \\ \hline
				IX-highCD-cas16-fb-16e16o & 512 & 8 & -7318.21108886 & n.a. \\
				IX-highCD-cas16-fb-16e16o & 1024 & 8 & -7318.21198760  & $-0.90$\\
				IX-highCD-cas16-fb-16e16o & 2048 & 8 & -7318.21221455 &  $-0.23$ \\
				IX-highCD-cas16-fb-16e16o & 2048 & 16 & -7318.21221456 & $-10^{-6}$ \\
				IX-highCD-cas16-fb-16e16o & 4096 & 16 & -7318.21225438 & $-0.04$ \\
				\bottomrule\bottomrule
			\end{tabular}
		\end{ruledtabular}
	\end{table}

	\begin{table}[htb]
		\caption{Summary of the parameters of the DMRG calculations performed to evaluate the error in the ground-state electronic energy due to a truncation of the two-electron integrals of the catalyst structures II-III and IX and the linear chains of Hydrogen atoms H$_2$, H$_4$, H$_6$, and H$_8$. The active space for the DMRG calculations for II-III and IX corresponds to the small and intermediate active space of the respective catalyst structure (see~\cref{tab:integralFiles1,tab:integralFiles2}, the `highCD'  integrals) whereas for the linear chains of Hydrogen atoms, the full orbital space was employed. The maximum bond dimension and number of sweeps of the DMRG calculations are given by $m$ and $n$, respectively. \label{tab:trunc-settings}}
		\begin{ruledtabular}
			\begin{tabular}{lllll}
				System & Active electrons & Active orbitals & $m$ & $n$ \\\hline
				H$_2$ & 2 & 4 & 1000 & 10 \\
				H$_4$ & 4 & 8 & 1000 & 10 \\
				H$_6$ & 6 & 12 & 1000 & 10 \\
				H$_8$ & 8 & 16 & 1000 & 10 \\
				II-III & 8 & 6 & 500 & 5 \\
				IX & 16 & 16 & 2048 & 16 \\
			\end{tabular}
		\end{ruledtabular}
	\end{table}
	
	The resulting DMRG-CI electronic ground-state energies are given in~\cref{tab:truncation-II-III,tab:truncation-IX,tab:truncation-Hydrogen-chains}.
	
	\newpage
	\setlength{\LTcapwidth}{\linewidth}
	\begin{longtable}{clllclll}
		\caption{DMRG-CI electronic energies in Hartree atomic units for the linear chains of Hydrogen atoms H$_2$, H$_4$, H$_6$, and H$_8$ evaluated from low-rank approximations to the double-factorized Hamiltonian at two different truncation schemes. $\epsilon_{\text{in}/\text{co}}$ denotes the value of the truncation parameter $\epsilon$ of either the incoherent truncation ($\epsilon_{\text{in}}$) of the two-electron integrals or the coherent one ($\epsilon_{\text{co}})$. The resulting ground-state DMRG-CI energy of the integral file at a given truncation level is denoted by $E_{\mathrm{el}}^{\mathrm{DMRG-CI}} (\epsilon_{\text{in}})$ for the incoherent truncation scheme and $E_{\mathrm{el}}^{\mathrm{DMRG-CI}} (\epsilon_{\text{co}})$ for the coherent truncation scheme. The energy associated with the non-truncated integrals is given as the energy at the truncation value of $0.0$ Hartree.\label{tab:truncation-Hydrogen-chains}}\\
		\hline\hline
		\multirow{2}{*}{System} & $\epsilon_{\text{in}/\text{co}}$ & $E_{\mathrm{el}}^{\mathrm{DMRG-CI}} (\epsilon_{\text{in}})$ &  $E_{\mathrm{el}}^{\mathrm{DMRG-CI}} (\epsilon_{\text{co}})$ & \multirow{2}{*}{System} & $\epsilon_{\text{in}/\text{co}}$ & $E_{\mathrm{el}}^{\mathrm{DMRG-CI}} (\epsilon_{\text{in}})$ &  $E_{\mathrm{el}}^{\mathrm{DMRG-CI}} (\epsilon_{\text{co}})$ \\
		& / Hartree  & / Hartree & / Hartree  & & / Hartree & / Hartree & / Hartree \\\hline
		H$_2$ & 0.000000 & -1.12158913 & -1.12158913 & H$_4$ & 0.000000 & -2.22021372 & -2.22021372 \\
		H$_2$ & 0.000001 & -1.12158913 & -1.12158913 & H$_4$ & 0.000001 & -2.22021393 & -2.22021374 \\
		H$_2$ & 0.000002 & -1.12158913 & -1.12158913 & H$_4$ & 0.000002 & -2.22021378 & -2.22021374 \\
		H$_2$ & 0.000003 & -1.12158913 & -1.12158913 & H$_4$ & 0.000003 & -2.22021414 & -2.22021386 \\
		H$_2$ & 0.000004 & -1.12158913 & -1.12158913 & H$_4$ & 0.000004 & -2.22021422 & -2.22021397 \\
		H$_2$ & 0.000005 & -1.12158913 & -1.12158913 & H$_4$ & 0.000005 & -2.22021428 & -2.22021397 \\
		H$_2$ & 0.000006 & -1.12158913 & -1.12158913 & H$_4$ & 0.000006 & -2.22021468 & -2.22021396 \\
		H$_2$ & 0.000008 & -1.12158913 & -1.12158913 & H$_4$ & 0.000008 & -2.22021653 & -2.22021397 \\
		H$_2$ & 0.000010 & -1.12158913 & -1.12158913 & H$_4$ & 0.000010 & -2.22021770 & -2.22021393 \\
		H$_2$ & 0.000013 & -1.12158913 & -1.12158913 & H$_4$ & 0.000013 & -2.22021642 & -2.22021394 \\
		H$_2$ & 0.000016 & -1.12158913 & -1.12158913 & H$_4$ & 0.000016 & -2.22022156 & -2.22021378 \\
		H$_2$ & 0.000020 & -1.12158913 & -1.12158913 & H$_4$ & 0.000020 & -2.22022139 & -2.22021378 \\
		H$_2$ & 0.000025 & -1.12158913 & -1.12158913 & H$_4$ & 0.000025 & -2.22022737 & -2.22021417 \\
		H$_2$ & 0.000032 & -1.12158913 & -1.12158913 & H$_4$ & 0.000032 & -2.22022689 & -2.22021422 \\
		H$_2$ & 0.000040 & -1.12158913 & -1.12158913 & H$_4$ & 0.000040 & -2.22022945 & -2.22021428 \\
		H$_2$ & 0.000050 & -1.12158913 & -1.12158913 & H$_4$ & 0.000050 & -2.22023611 & -2.22021468 \\
		H$_2$ & 0.000063 & -1.12158913 & -1.12158913 & H$_4$ & 0.000063 & -2.22023841 & -2.22021633 \\
		H$_2$ & 0.000079 & -1.12158913 & -1.12158913 & H$_4$ & 0.000079 & -2.22023851 & -2.22021711 \\
		H$_2$ & 0.000100 & -1.12158913 & -1.12158913 & H$_4$ & 0.000100 & -2.22023926 & -2.22021770 \\
		H$_2$ & 0.000126 & -1.12158913 & -1.12158913 & H$_4$ & 0.000126 & -2.22025864 & -2.22021642 \\
		H$_2$ & 0.000158 & -1.12160508 & -1.12158913 & H$_4$ & 0.000158 & -2.22026109 & -2.22021956 \\
		H$_2$ & 0.000200 & -1.12160508 & -1.12158913 & H$_4$ & 0.000200 & -2.22025525 & -2.22022226 \\
		H$_2$ & 0.000251 & -1.12164293 & -1.12158913 & H$_4$ & 0.000251 & -2.22025040 & -2.22022578 \\
		H$_2$ & 0.000316 & -1.12164896 & -1.12160508 & H$_4$ & 0.000316 & -2.22025288 & -2.22022734 \\
		H$_2$ & 0.000398 & -1.12165993 & -1.12160508 & H$_4$ & 0.000398 & -2.22025826 & -2.22022919 \\
		H$_2$ & 0.000501 & -1.12165993 & -1.12160508 & H$_4$ & 0.000501 & -2.22026298 & -2.22022954 \\
		H$_2$ & 0.000631 & -1.12165993 & -1.12164293 & H$_4$ & 0.000631 & -2.22028388 & -2.22023802 \\
		H$_2$ & 0.000794 & -1.12166545 & -1.12164293 & H$_4$ & 0.000794 & -2.22023929 & -2.22023851 \\
		H$_2$ & 0.001000 & -1.12166545 & -1.12164896 & H$_4$ & 0.001000 & -2.22041577 & -2.22023742 \\
		H$_2$ & 0.001259 & -1.12167101 & -1.12164896 & H$_4$ & 0.001259 & -2.22045034 & -2.22025864 \\
		H$_2$ & 0.001585 & -1.12195357 & -1.12165993 & H$_4$ & 0.001585 & -2.22075130 & -2.22026109 \\
		H$_2$ & 0.001995 & -1.12223966 & -1.12165993 & H$_4$ & 0.001995 & -2.22085576 & -2.22025509 \\
		H$_2$ & 0.002512 & -1.12223266 & -1.12165993 & H$_4$ & 0.002512 & -2.22097745 & -2.22025040 \\
		H$_2$ & 0.003162 & -1.12223266 & -1.12166545 & H$_4$ & 0.003162 & -2.22179419 & -2.22025566 \\
		H$_2$ & 0.003981 & -1.12226773 & -1.12166545 & H$_4$ & 0.003981 & -2.22188997 & -2.22025826 \\
		H$_2$ & 0.005012 & -1.12293938 & -1.12167101 & H$_4$ & 0.005012 & -2.22213743 & -2.22029173 \\
		H$_2$ & 0.006310 & -1.12456443 & -1.12167101 & H$_4$ & 0.006310 & -2.22269565 & -2.22028388 \\
		H$_2$ & 0.007943 & -1.12457007 & -1.12195357 & H$_4$ & 0.007943 & -2.22264060 & -2.22030357 \\
		H$_2$ & 0.010000 & -1.12521364 & -1.12194658 & H$_4$ & 0.010000 & -2.22413315 & -2.22024464 \\
		H$_2$ & 0.012589 & -1.12478454 & -1.12223266 & H$_4$ & 0.012589 & -2.22570303 & -2.22041314 \\
		H$_2$ & 0.015849 & -1.12599035 & -1.12226773 & H$_4$ & 0.015849 & -2.22713724 & -2.22055674 \\
		H$_2$ & 0.019953 & -1.12599176 & -1.12293938 & H$_4$ & 0.019953 & -2.22977827 & -2.22075130 \\
		H$_2$ & 0.025119 & -1.12789093 & -1.12293938 & H$_4$ & 0.025119 & -2.23257326 & -2.22085576 \\
		H$_2$ & 0.031623 & -1.13295508 & -1.12456443 & H$_4$ & 0.031623 & -2.23862926 & -2.22097745 \\
		H$_2$ & 0.039811 & -1.13508558 & -1.12457007 & H$_4$ & 0.039811 & -2.25219293 & -2.22172800 \\
		H$_2$ & 0.050119 & -1.15524671 & -1.12521364 & H$_4$ & 0.050119 & -2.26286296 & -2.22184603 \\
		H$_2$ & 0.063096 & -1.15524671 & -1.12478454 & H$_4$ & 0.063096 & -2.26411098 & -2.22280707 \\
		H$_2$ & 0.079433 & -1.16115222 & -1.12478867 & H$_4$ & 0.079433 & -2.27453848 & -2.22285026 \\
		H$_2$ & 0.100000 & -1.16973021 & -1.12599176 & H$_4$ & 0.100000 & -2.28592300 & -2.22277314 \\\hline
		H$_6$ & 0.000000 & -3.32158066 & -3.32158066 & H$_8$ & 0.000000 & -4.34617684 & -4.34617684 \\
		H$_6$ & 0.000001 & -3.32158106 & -3.32158069 & H$_8$ & 0.000001 & -4.34617730 & -4.34617686 \\
		H$_6$ & 0.000002 & -3.32158141 & -3.32158070 & H$_8$ & 0.000002 & -4.34617766 & -4.34617687 \\
		H$_6$ & 0.000003 & -3.32158160 & -3.32158073 & H$_8$ & 0.000003 & -4.34617815 & -4.34617688 \\
		H$_6$ & 0.000004 & -3.32158154 & -3.32158076 & H$_8$ & 0.000004 & -4.34617880 & -4.34617690 \\
		H$_6$ & 0.000005 & -3.32158200 & -3.32158075 & H$_8$ & 0.000005 & -4.34617900 & -4.34617692 \\
		H$_6$ & 0.000006 & -3.32158238 & -3.32158079 & H$_8$ & 0.000006 & -4.34617932 & -4.34617694 \\
		H$_6$ & 0.000008 & -3.32158325 & -3.32158084 & H$_8$ & 0.000008 & -4.34617947 & -4.34617696 \\
		H$_6$ & 0.000010 & -3.32158386 & -3.32158086 & H$_8$ & 0.000010 & -4.34618052 & -4.34617700 \\
		H$_6$ & 0.000013 & -3.32158445 & -3.32158093 & H$_8$ & 0.000013 & -4.34618092 & -4.34617706 \\
		H$_6$ & 0.000016 & -3.32158222 & -3.32158095 & H$_8$ & 0.000016 & -4.34618121 & -4.34617709 \\
		H$_6$ & 0.000020 & -3.32158388 & -3.32158102 & H$_8$ & 0.000020 & -4.34618138 & -4.34617716 \\
		H$_6$ & 0.000025 & -3.32158395 & -3.32158112 & H$_8$ & 0.000025 & -4.34618025 & -4.34617723 \\
		H$_6$ & 0.000032 & -3.32158764 & -3.32158109 & H$_8$ & 0.000032 & -4.34618418 & -4.34617720 \\
		H$_6$ & 0.000040 & -3.32159475 & -3.32158111 & H$_8$ & 0.000040 & -4.34619035 & -4.34617730 \\
		H$_6$ & 0.000050 & -3.32160131 & -3.32158169 & H$_8$ & 0.000050 & -4.34619473 & -4.34617740 \\
		H$_6$ & 0.000063 & -3.32161010 & -3.32158159 & H$_8$ & 0.000063 & -4.34619620 & -4.34617764 \\
		H$_6$ & 0.000079 & -3.32161480 & -3.32158184 & H$_8$ & 0.000079 & -4.34620736 & -4.34617808 \\
		H$_6$ & 0.000100 & -3.32161479 & -3.32158215 & H$_8$ & 0.000100 & -4.34621712 & -4.34617875 \\
		H$_6$ & 0.000126 & -3.32162599 & -3.32158266 & H$_8$ & 0.000126 & -4.34622564 & -4.34617877 \\
		H$_6$ & 0.000158 & -3.32165812 & -3.32158330 & H$_8$ & 0.000158 & -4.34623466 & -4.34617920 \\
		H$_6$ & 0.000200 & -3.32167318 & -3.32158434 & H$_8$ & 0.000200 & -4.34624784 & -4.34617919 \\
		H$_6$ & 0.000251 & -3.32167606 & -3.32158422 & H$_8$ & 0.000251 & -4.34625387 & -4.34617953 \\
		H$_6$ & 0.000316 & -3.32176286 & -3.32158239 & H$_8$ & 0.000316 & -4.34628688 & -4.34618063 \\
		H$_6$ & 0.000398 & -3.32171852 & -3.32158342 & H$_8$ & 0.000398 & -4.34621573 & -4.34618099 \\
		H$_6$ & 0.000501 & -3.32171439 & -3.32158481 & H$_8$ & 0.000501 & -4.34627616 & -4.34618128 \\
		H$_6$ & 0.000631 & -3.32171118 & -3.32158902 & H$_8$ & 0.000631 & -4.34629912 & -4.34618081 \\
		H$_6$ & 0.000794 & -3.32173600 & -3.32159451 & H$_8$ & 0.000794 & -4.34629439 & -4.34618028 \\
		H$_6$ & 0.001000 & -3.32189021 & -3.32160091 & H$_8$ & 0.001000 & -4.34628547 & -4.34618546 \\
		H$_6$ & 0.001259 & -3.32193984 & -3.32160671 & H$_8$ & 0.001259 & -4.34640554 & -4.34619096 \\
		H$_6$ & 0.001585 & -3.32202636 & -3.32161480 & H$_8$ & 0.001585 & -4.34649770 & -4.34619721 \\
		H$_6$ & 0.001995 & -3.32251898 & -3.32161759 & H$_8$ & 0.001995 & -4.34663895 & -4.34620128 \\
		H$_6$ & 0.002512 & -3.32283885 & -3.32162047 & H$_8$ & 0.002512 & -4.34653313 & -4.34621092 \\
		H$_6$ & 0.003162 & -3.32301884 & -3.32165725 & H$_8$ & 0.003162 & -4.34747627 & -4.34621651 \\
		H$_6$ & 0.003981 & -3.32324022 & -3.32167319 & H$_8$ & 0.003981 & -4.34793027 & -4.34622811 \\
		H$_6$ & 0.005012 & -3.32388399 & -3.32167315 & H$_8$ & 0.005012 & -4.34849946 & -4.34623794 \\
		H$_6$ & 0.006310 & -3.32508229 & -3.32171495 & H$_8$ & 0.006310 & -4.34944662 & -4.34625163 \\
		H$_6$ & 0.007943 & -3.32543283 & -3.32172159 & H$_8$ & 0.007943 & -4.34993224 & -4.34625462 \\
		H$_6$ & 0.010000 & -3.32541755 & -3.32171487 & H$_8$ & 0.010000 & -4.34921153 & -4.34629028 \\
		H$_6$ & 0.012589 & -3.32710745 & -3.32173575 & H$_8$ & 0.012589 & -4.35177484 & -4.34621821 \\
		H$_6$ & 0.015849 & -3.33011387 & -3.32174741 & H$_8$ & 0.015849 & -4.35336244 & -4.34625426 \\
		H$_6$ & 0.019953 & -3.33237130 & -3.32180721 & H$_8$ & 0.019953 & -4.35525855 & -4.34629940 \\
		H$_6$ & 0.025119 & -3.33587433 & -3.32193798 & H$_8$ & 0.025119 & -4.35984793 & -4.34629900 \\
		H$_6$ & 0.031623 & -3.33689545 & -3.32205172 & H$_8$ & 0.031623 & -4.36466842 & -4.34629188 \\
		H$_6$ & 0.039811 & -3.34292326 & -3.32242166 & H$_8$ & 0.039811 & -4.36972018 & -4.34638023 \\
		H$_6$ & 0.050119 & -3.34354251 & -3.32278773 & H$_8$ & 0.050119 & -4.37083690 & -4.34652382 \\
		H$_6$ & 0.063096 & -3.37708344 & -3.32295246 & H$_8$ & 0.063096 & -4.37667930 & -4.34650941 \\
		H$_6$ & 0.079433 & -3.39476016 & -3.32307469 & H$_8$ & 0.079433 & -4.39344403 & -4.34648766 \\
		H$_6$ & 0.100000 & -3.40138205 & -3.32346661 & H$_8$ & 0.100000 & -4.43190074 & -4.34675691 \\\hline \hline
	\end{longtable}
	
	\begin{table}[htb]
		\caption{DMRG
			-CI electronic energies in Hartree atomic units for the catalyst structure II-III with a CAS(8,6), evaluated from low-rank approximations to the double-factorized Hamiltonian at two different truncation schemes. The base integral file with the non-truncated integrals is II-III-highCD-cas6-fb-8e6o. $\epsilon_{\text{in}/\text{co}}$ denotes the value of the truncation parameter $\epsilon$ of either the incoherent truncation ($\epsilon_{\text{in}}$) of the two-electron integrals or the coherent one ($\epsilon_{\text{co}})$. The resulting ground-state DMRG(8,6)-CI energy of the integral file at a given truncation level is denoted by $E_{\mathrm{el}}^{\mathrm{DMRG-CI}} (\epsilon_{\text{in}})$ for the incoherent truncation scheme and $E_{\mathrm{el}}^{\mathrm{DMRG-CI}} (\epsilon_{\text{co}})$ for the coherent one. The energy associated with the non-truncated integrals is given in the first row at a truncation value of $0.0$ Hartree.\label{tab:truncation-II-III}}
		\begin{ruledtabular}
			\begin{tabular}{lll}
				$\epsilon_{\text{in}/\text{co}}$ & $E_{\mathrm{el}}^{\mathrm{DMRG-CI}} (\epsilon_{\text{in}})$ &  $E_{\mathrm{el}}^{\mathrm{DMRG-CI}} (\epsilon_{\text{co}})$ \\
				/ Hartree  & / Hartree & / Hartree   \\\hline\hline
				0.000000 & -7318.00215013 & -7318.00215013 \\
				0.000100 & -7318.00215864 & -7318.00215187 \\
				0.000126 & -7318.00215503 & -7318.00215187 \\
				0.000158 & -7318.00213043 & -7318.00215274 \\
				0.000200 & -7318.00216079 & -7318.00215274 \\
				0.000251 & -7318.00215679 & -7318.00215864 \\
				0.000316 & -7318.00213123 & -7318.00215864 \\
				0.000398 & -7318.00216705 & -7318.00216409 \\
				0.000501 & -7318.00221220 & -7318.00216409 \\
				0.000631 & -7318.00220470 & -7318.00213948 \\
				0.000794 & -7318.00219296 & -7318.00216986 \\
				0.001000 & -7318.00229250 & -7318.00216625 \\
				0.001259 & -7318.00250109 & -7318.00214070 \\
				0.001585 & -7318.00280990 & -7318.00217604 \\
				0.001995 & -7318.00249959 & -7318.00216705 \\
				0.002512 & -7318.00254140 & -7318.00218788 \\
				0.003162 & -7318.00262988 & -7318.00221220 \\
				0.003981 & -7318.00333674 & -7318.00220470 \\
				0.005012 & -7318.00680238 & -7318.00219296 \\
				0.006310 & -7318.00737955 & -7318.00220866 \\
				0.007943 & -7318.00814157 & -7318.00241234 \\
				0.010000 & -7318.01010856 & -7318.00271304 \\
				0.012589 & -7318.00987156 & -7318.00263355 \\
				0.015849 & -7318.01217928 & -7318.00249959 \\
				0.019953 & -7318.01212307 & -7318.00247483 \\
				0.025119 & -7318.01908891 & -7318.00262988 \\
				0.031623 & -7318.02160123 & -7318.00301686 \\
				0.039811 & -7318.03507773 & -7318.00526335 \\
				0.050119 & -7318.04053844 & -7318.00683148 \\
				0.063096 & -7318.04732502 & -7318.00769067 \\
				0.079433 & -7318.08057598 & -7318.00814157 \\
				0.100000 & -7318.10412106 & -7318.00829295 \\
			\end{tabular}
		\end{ruledtabular}
	\end{table}

	\begin{table}[htb]
		\caption{DMRG-CI electronic energies in Hartree atomic units for the catalyst structure IX with a CAS(16,16), evaluated from low-rank approximations to the double-factorized Hamiltonian at two different truncation schemes.  The base integral file with the non-truncated integrals is IX-highCD-cas16-fb-16e16o. $\epsilon_{\text{in}/\text{co}}$ denotes the value of the truncation parameter $\epsilon$ of either the incoherent truncation ($\epsilon_{\text{in}}$) of the two-electron integrals or the coherent one ($\epsilon_{\text{co}})$. The resulting ground-state DMRG(16,16)-CI energy of the integral file at a given truncation level is denoted by $E_{\mathrm{el}}^{\mathrm{DMRG-CI}} (\epsilon_{\text{in}})$ for the incoherent truncation scheme and $E_{\mathrm{el}}^{\mathrm{DMRG-CI}} (\epsilon_{\text{co}})$ for the coherent one. The energy associated with the non-truncated integrals is given in the first row at a truncation value of $0.0$ Hartree.\label{tab:truncation-IX}}
		\begin{ruledtabular}
			\begin{tabular}{lll}
				$\epsilon_{\text{in}/\text{co}}$ & $E_{\mathrm{el}}^{\mathrm{DMRG-CI}} (\epsilon_{\text{in}})$ &  $E_{\mathrm{el}}^{\mathrm{DMRG-CI}} (\epsilon_{\text{co}})$ \\
				/ Hartree  & / Hartree & / Hartree   \\\hline
				0.000000 & -7318.21221456 & -7318.21221456 \\
				0.000100 & -7318.21230964 & -7318.21221600 \\
				0.000126 & -7318.21234289 & -7318.21221644 \\
				0.000158 & -7318.21237029 & -7318.21221704 \\
				0.000200 & -7318.21240692 & -7318.21221761 \\
				0.000251 & -7318.21246425 & -7318.21221840 \\
				0.000316 & -7318.21251053 & -7318.21221965 \\
				0.000398 & -7318.21259187 & -7318.21222093 \\
				0.000501 & -7318.21268462 & -7318.21222356 \\
				0.000631 & -7318.21281515 & -7318.21222591 \\
				0.000794 & -7318.21297009 & -7318.21222853 \\
				0.001000 & -7318.21312759 & -7318.21223208 \\
				0.001259 & -7318.21330936 & -7318.21223602 \\
				0.001585 & -7318.21365065 & -7318.21224246 \\
				0.001995 & -7318.21397040 & -7318.21224753 \\
				0.002512 & -7318.21431504 & -7318.21225796 \\
				0.003162 & -7318.21484663 & -7318.21227019 \\
				0.003981 & -7318.21581205 & -7318.21228159 \\
				0.005012 & -7318.21688967 & -7318.21230107 \\
				0.006310 & -7318.21781895 & -7318.21233012 \\
				0.007943 & -7318.21914577 & -7318.21235753 \\
				0.010000 & -7318.22001812 & -7318.21239445 \\
				0.012589 & -7318.22199964 & -7318.21243973 \\
				0.015849 & -7318.22428406 & -7318.21248711 \\
				0.019953 & -7318.22385766 & -7318.21254718 \\
				0.025119 & -7318.22690256 & -7318.21263369 \\
				0.031623 & -7318.22795625 & -7318.21273693 \\
				0.039811 & -7318.24170762 & -7318.21286327 \\
				0.050119 & -7318.27394958 & -7318.21307670 \\
				0.063096 & -7318.30682217 & -7318.21319792 \\
				0.079433 & -7318.35577244 & -7318.21351128 \\
				0.100000 & -7318.43110749 & -7318.21376772 \\
			\end{tabular}
		\end{ruledtabular}
	\end{table}

	\newpage
	\subsection{General procedure for the selection of molecular orbitals for the integral files}\label{subsec:ASSelection}
	
	During the selection of the small active spaces for the CASSCF calculations, it became evident that the intermediates and transition states in this study are mainly dynamically correlated, as indicated by e.g. low values of our multi-configurational measure $Z_{s(1)}$~\cite{stein17a}. While the orbitals are sufficiently correlated to allow for the determination of a small active space for the CASSCF calculations, the selection becomes ambiguous for larger active space sizes because of the dominant dynamic correlation that renders the active space concept arbitrary. For the selection of the active spaces for the integral files, we therefore chose a scheme that is as reproducible as possible while still maintaining chemical relevance. For each complex, we selected at least three different active spaces from the valence space of the CAS($N$,$L$)SCF orbitals, termed the \textit{small}, \textit{intermediate}, and \textit{large} active spaces. The small active space consists of the $L$ orbitals of the underlying CAS($N$,$L$)SCF calculation. The intermediate one includes the small active space plus the orbitals which are primarily localized on CO\textsubscript{2}, HCO\textsubscript{2}\textsuperscript{-}, HCOOH, CH\textsubscript{2}OOH\textsuperscript{-}, CH\textsubscript{3}O\textsuperscript{-}, H\textsubscript{2}, H, Ruthenium or in the bonding region between the metal and the ligands. The large one comprises the intermediate active space plus the 18 $\pi$ and 18 $\pi^*$ orbitals of the 6 phenyl groups that are part of the triphos ligand.
	
	To investigate the performance of the quantum algorithms in the limit of very large active spaces, we additionally generated integral files with 100, 150, and 250 active orbitals for the complex XVIII. As a selection guided by manual inspection becomes unfeasible for these active space sizes, we chose to include the $n/2$ orbitals below (and including) the HOMO and the $n/2$ orbitals above it (with $n$=100, 150, 250).

	For complex II, we also generated integral files from HF orbitals to understand the effect of the orbital type on the resource estimates. We selected the active space analogously to the CASSCF active spaces but without the small active space (as it is not defined for the HF orbitals).

	\section{Overview of the integral files}
	
	The files containing the one- and two-electron integrals for a selection of molecular orbitals as defined in~\cref{subsec:ASSelection} are summarized in~\cref{tab:integralFiles1,tab:integralFiles1}, which also lists the atomic and molecular orbital bases used as well as the active spaces for the integrals. The notation for the files is \textit{catalyst-integralAccuracy-MObasis-AObasis-$n_e$e$n_o$o}, where \textit{catalyst} is one of the catalyst structure identifiers I, II, II-III, V, VIII, VIII-IX, IX, or XVIII, \textit{integralAccuracy} denotes the tightness of the CD threshold, with \textit{highCD} standing for a threshold of $10^{-8}$ and no string for a threshold of $10^{-4}$; \textit{MObasis} denotes the molecular orbitals basis, \textit{AObasis} the atomic orbital basis, and $n_e$ and $n_o$ are the number of electrons and orbitals, respectively.
	
	\begin{table}[htb]
		\caption{List of the integral files resulting from the active space selection outlined in~\cref{subsec:ASSelection} for complexes I to VIII with the notation explained in the text. The MO active space is the one of the underlying CASSCF calculation through which the orbitals have been obtained, whereas the integral active space denotes the size of the active space chosen for the integral file under consideration. Calculations were performed with a Cholesky Decomposition of the two-electron integrals with two thresholds of $10^{-4}$ and $10^{-8}$. Files that were obtained with the high-accuracy threshold $10^{-8}$ are labeled ``highCD'', whereas the absence of that label indicates the default threshold of $10^{-4}$. \label{tab:integralFiles1}}
		\begin{ruledtabular}
			\begin{tabular}{lllll}
				Name of file & AO    & MO    & MO           & Integral \\
				& basis & basis & active space & active space \\\hline
				I-cas5-fb-4e5o & fb & CASSCF & (4,5)  & (4,5)  \\
				I-cas5-fb-14e16o & fb & CASSCF & (4,5)  & (14,16)  \\
				I-cas5-fb-48e52o & fb & CASSCF & (4,5)  & (48,52)    \\
				I-highCD-cas5-fb-4e5o & fb & CASSCF & (4,5) & (4,5) \\
				I-highCD-cas5-fb-14e16o & fb & CASSCF & (4,5) & (14,16) \\
				I-highCD-cas5-fb-48e52o & fb & CASSCF & (4,5) & (48,52) \\
				\hline
				II-cas6-fb-8e6o & fb & CASSCF & (8,6) & (8,6) \\
				II-cas6-fb-34e26o & fb & CASSCF & (8,6) & (34,26) \\
				II-cas6-fb-70e62o & fb & CASSCF & (8,6) & (70,62) \\
				II-highCD-cas6-fb-8e6o & fb & CASSCF & (8,6) & (8,6) \\
				II-highCD-cas6-fb-34e26o & fb & CASSCF & (8,6) & (34,26) \\
				II-highCD-cas6-fb-70e62o & fb & CASSCF & (8,6) & (70,62) \\
				II-hf-fb-38e33o  & fb & HF & n.a. & (38,33) \\
				II-hf-fb-74e71o & fb & HF & n.a. & (74,71) \\
				\hline
				II-III-cas6-fb-8e6o & fb & CASSCF & (8,6) & (8,6) \\
				II-III-cas6-fb-38e29o & fb & CASSCF & (8,6) & (38,29)  \\
				II-III-cas6-fb-74e65o & fb &  CASSCF & (8,6) & (74,65) \\
				II-III-highCD-cas6-fb-8e6o & fb & CASSCF & (8,6) & (8,6) \\
				II-III-highCD-cas6-fb-38e29o & fb & CASSCF & (8,6) & (38,29)  \\
				II-III-highCD-cas6-fb-74e65o & fb &  CASSCF & (8,6) & (74,65)\\
				\hline
				V-cas11-fb-12e11o & fb & CASSCF & (12,11) & (12,11)  \\
				V-cas11-fb-32e24o & fb & CASSCF & (12,11) & (32,24)  \\
				V-cas11-fb-68e60o & fb & CASSCF & (12,11) & (68,60)  \\
				V-highCD-cas11-fb-12e11o & fb & CASSCF & (12,11) & (12,11)  \\
				V-highCD-cas11-fb-32e24o & fb & CASSCF & (12,11) & (32,24)  \\
				V-highCD-cas11-fb-68e60o & fb & CASSCF & (12,11) & (68,60)  \\
				\hline
				VIII-cas2-fb-2e2o & fb & CASSCF & (2,2) & (2,2) \\
				VIII-cas2-fb-40e29o & fb & CASSCF & (2,2) & (40,29) \\
				VIII-cas2-fb-76e65o & fb & CASSCF & (2,2) & (76,65) \\
				VIII-highCD-cas2-fb-2e2o & fb & CASSCF & (2,2) & (2,2) \\
				VIII-highCD-cas2-fb-40e29o & fb & CASSCF & (2,2) & (40,29) \\
				VIII-highCD-cas2-fb-76e65o & fb & CASSCF & (2,2) & (76,65) \\
			\end{tabular}
		\end{ruledtabular}
	\end{table}

	\begin{table}[htb]
		\caption{List of the integral files resulting from the active space selection outlined in~\cref{subsec:ASSelection} for complexes VIII-IX to XVIII with the notation explained in the text. The MO active space is the one of the underlying CASSCF calculation through which the orbitals have been obtained, whereas the integral active space denotes the size of the active space chosen for the integral file under consideration. Calculations were performed with a Cholesky decomposition of the two-electron integrals with two thresholds of $10^{-4}$ and $10^{-8}$. Files that were obtained with the high-accuracy threshold $10^{-8}$ are labeled ``highCD'', whereas the absence of that label indicates the default threshold of $10^{-4}$. \label{tab:integralFiles2}}
		\begin{ruledtabular}
			\begin{tabular}{lllll}
				VIII-IX-cas4-fb-4e4o & fb & CASSCF & (4,4) & (4,4) \\
				VIII-IX-cas4-fb-36e23o & fb & CASSCF & (4,4) & (36,23) \\
				VIII-IX-cas4-fb-72e59o & fb & CASSCF & (4,4) & (72,59) \\
				VIII-IX-highCD-cas4-fb-4e4o & fb & CASSCF & (4,4) & (4,4) \\
				VIII-IX-highCD-cas4-fb-36e23o & fb & CASSCF & (4,4) & (36,23) \\
				VIII-IX-highCD-cas4-fb-72e59o & fb & CASSCF & (4,4) & (72,59) \\
				\hline
				IX-cas16-fb-16e16o & fb & CASSCF & (16,16) & (16,16) \\
				IX-cas16-fb-32e26o & fb & CASSCF & (16,16) & (32,26) \\
				IX-cas16-fb-68e62o & fb & CASSCF & (16,16) & (68,62) \\
				IX-highCD-cas4-cas16-fb-16e16o & fb & CASSCF & (16,16) & (16,16) \\
				IX-highCD-cas4-cas16-fb-32e26o & fb & CASSCF & (16,16) & (32,26) \\
				IX-highCD-cas4-cas16-fb-68e62o & fb & CASSCF & (16,16) & (68,62) \\
				\hline
				XVIII-cas4-fb-4e4o & fb & CASSCF & (4,4) & (4,4) \\
				XVIII-cas4-fb-28e20o & fb & CASSCF & (4,4) & (28,20) \\
				XVIII-cas4-fb-64e56o & fb & CASSCF & (4,4) & (64,56) \\
				XVIII-cas4-fb-100e100o & fb &  CASSCF & (4,4) & (100,100) \\
				XVIII-cas4-fb-150e150o & fb &  CASSCF & (4,4) & (150,150) \\
				XVIII-cas4-fb-250e250o & fb &  CASSCF & (4,4) & (250,250) \\
				XVIII-highCD-cas4-fb-4e4o & fb & CASSCF & (4,4) & (4,4) \\
				XVIII-highCD-cas4-fb-28e20o & fb & CASSCF & (4,4) & (28,20) \\
				XVIII-highCD-cas4-fb-64e56o & fb & CASSCF & (4,4) & (64,56) \\
				XVIII-highCD-cas4-fb-100e100o & fb &  CASSCF & (4,4) & (100,100) \\
				XVIII-highCD-cas4-fb-150e150o & fb &  CASSCF & (4,4) & (150,150) \\
				XVIII-highCD-cas4-fb-250e250o & fb &  CASSCF & (4,4) & (250,250) \\
			\end{tabular}
		\end{ruledtabular}
	\end{table}

	\clearpage
	\section{Graphical representation of active orbitals chosen for integral files} \label{sec:ElStruc}
	We provide in the following figures of the active orbitals chosen for each catalyst structure.
	
	\subsection{CAS(4,5)SCF active molecular orbitals chosen for integral files of stable intermediate I}

	\begin{figure*}[htb]
		\centering
		\includegraphics[width=0.9\textwidth]{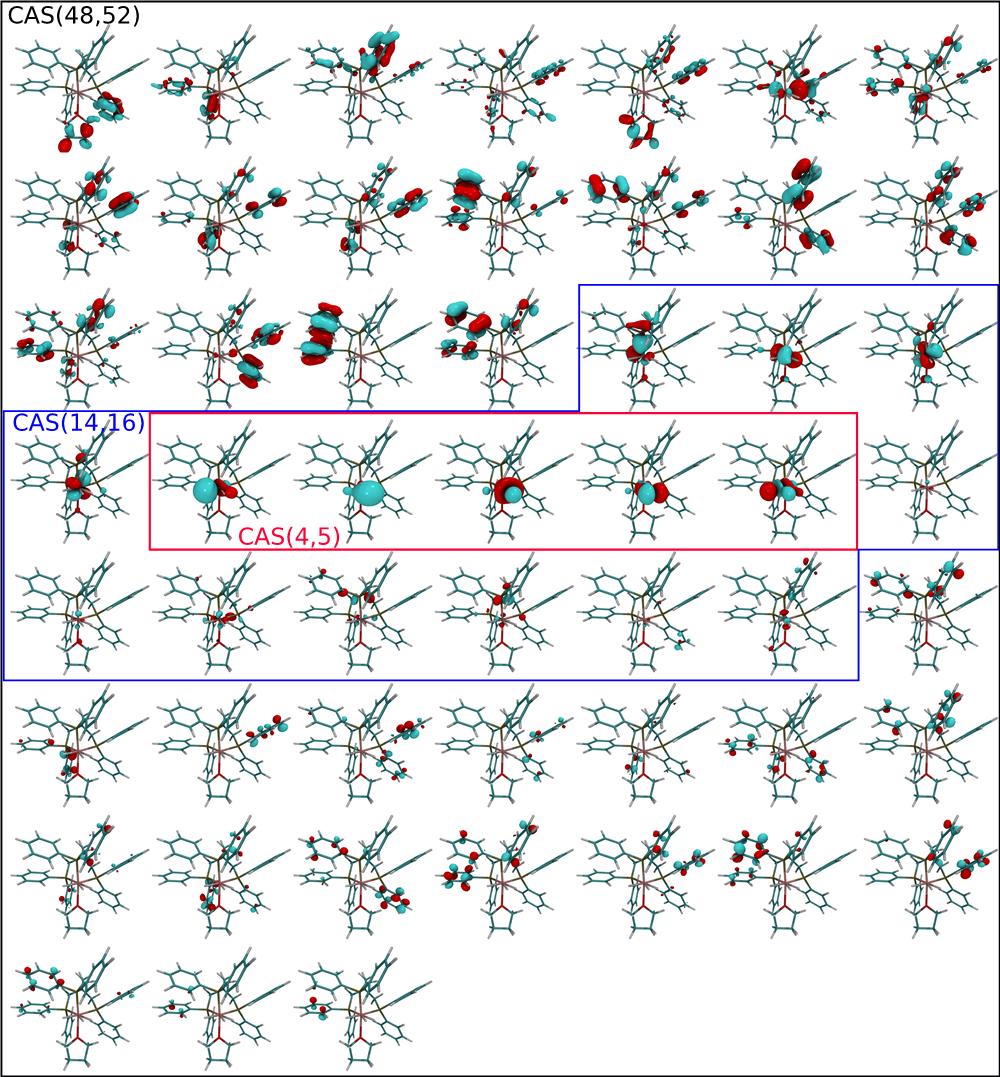}
		\caption{Active spaces yielding integral files for complex structure  I produced from CAS(4,5)SCF molecular orbitals in the full atomic orbital basis.}
	\end{figure*}

	\clearpage
	\subsection{CAS(8,6)SCF active molecular orbitals chosen for integral files of stable intermediate II}
	\begin{figure*}[htb]
		\centering
		\includegraphics[width=0.9\textwidth]{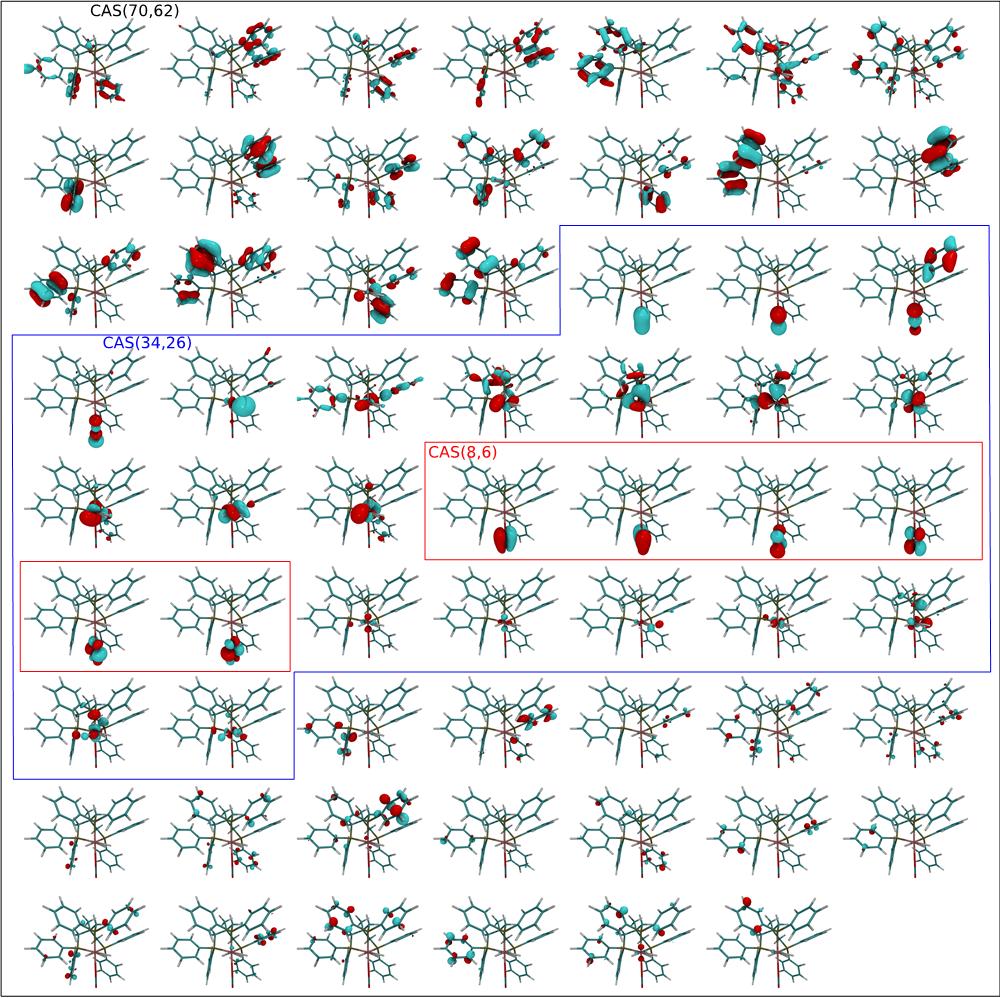}
		\caption{Active spaces yielding integral files for complex structure II produced from CAS(8,6)SCF molecular orbitals in the full atomic orbital basis.}
	\end{figure*}
	
	\clearpage

	\subsection{CAS(8,6)SCF active molecular orbitals chosen for integral files of transition state II-III}
	\begin{figure*}[htb]
		\centering
		\includegraphics[width=0.9\textwidth]{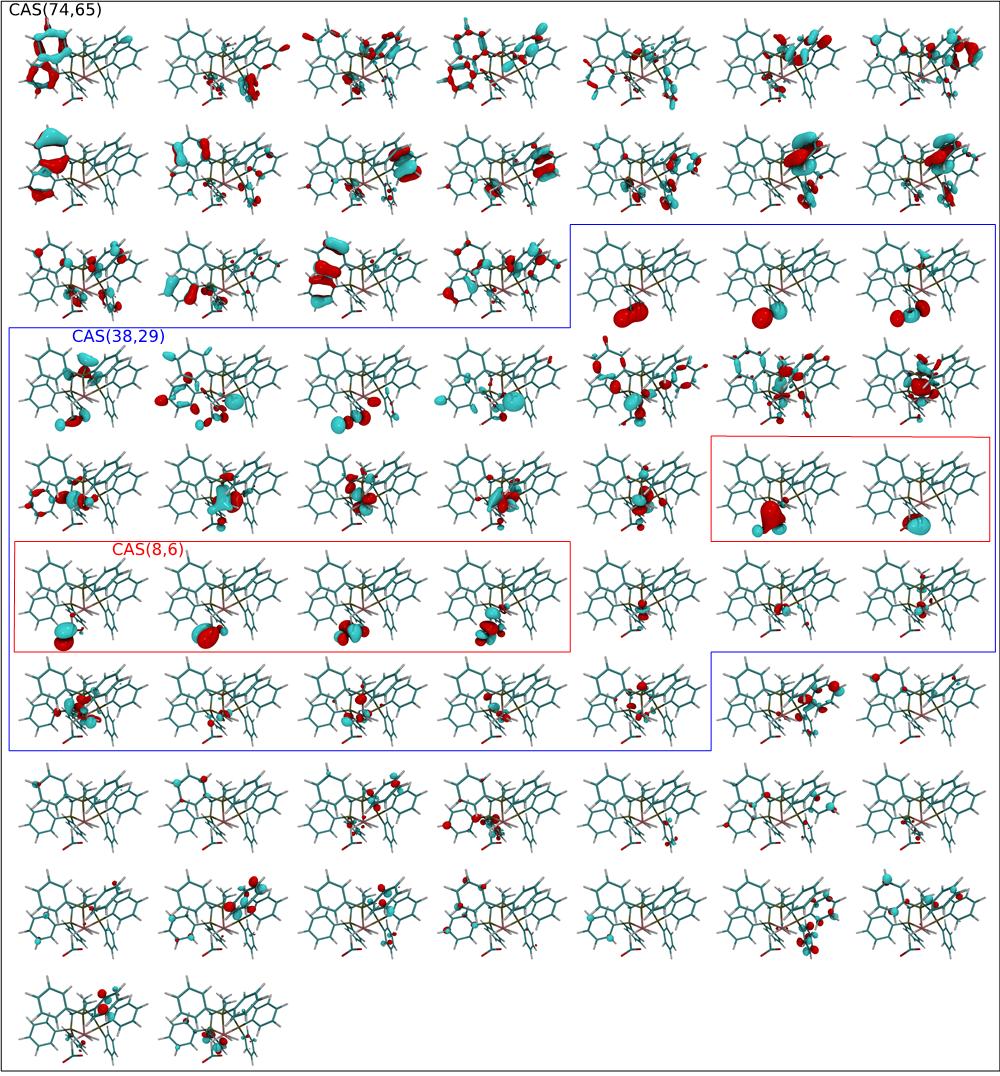}
		\caption{Active spaces yielding integral files for complex structure II-III produced from CAS(8,6)SCF molecular orbitals in the full atomic orbital basis.}
	\end{figure*}
	
	\newpage

	\subsection{CAS(12,11)SCF active molecular orbitals chosen for integral files of stable intermediate V}
	\begin{figure*}[htb]
		\centering
		\includegraphics[width=0.9\textwidth]{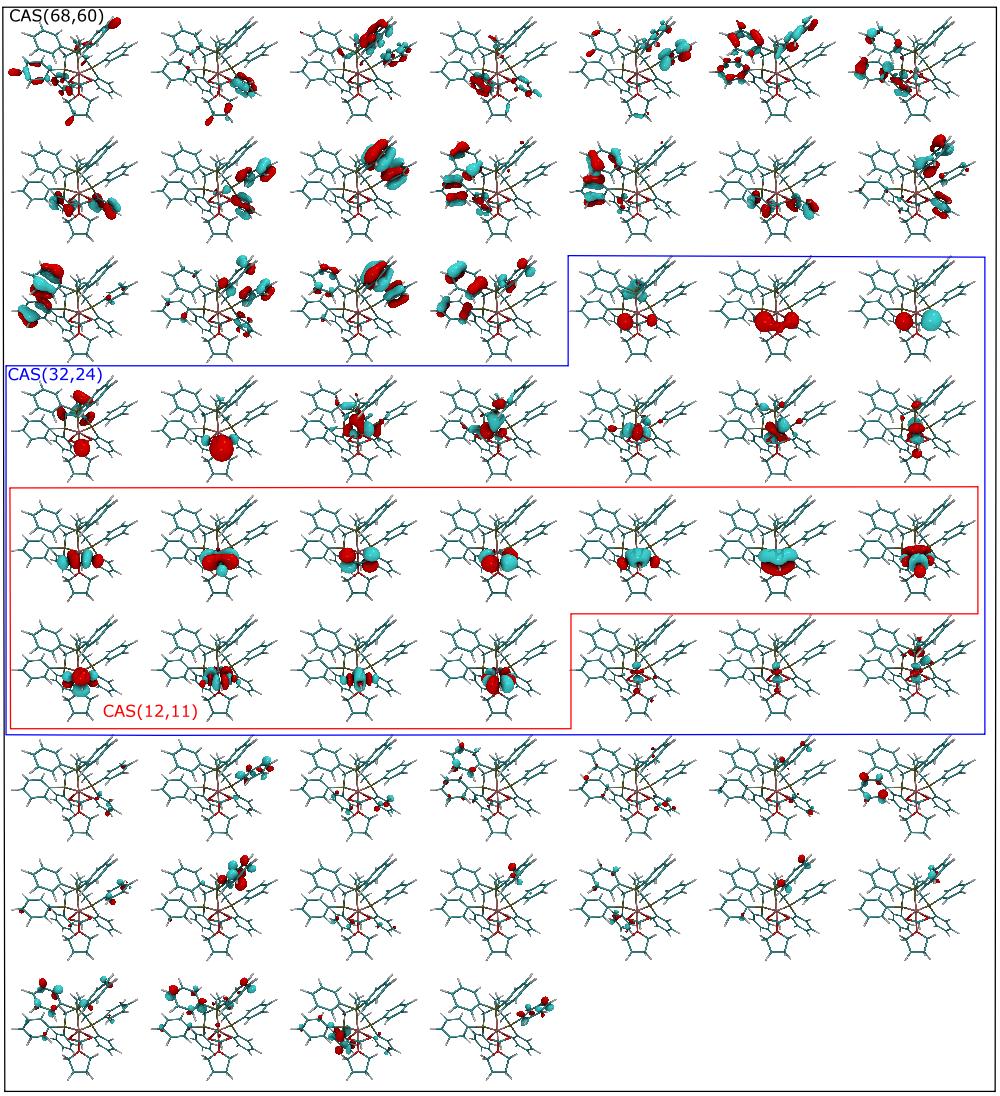}
		\caption{Active spaces yielding integral files for complex structure V produced from CAS(12,11)SCF molecular orbitals in the full atomic orbital basis.}
	\end{figure*}
	
	\newpage

	\subsection{CAS(2,2)SCF active molecular orbitals chosen for integral files of stable intermediate VIII}
	\begin{figure*}[htb]
		\centering
		\includegraphics[width=0.9\textwidth]{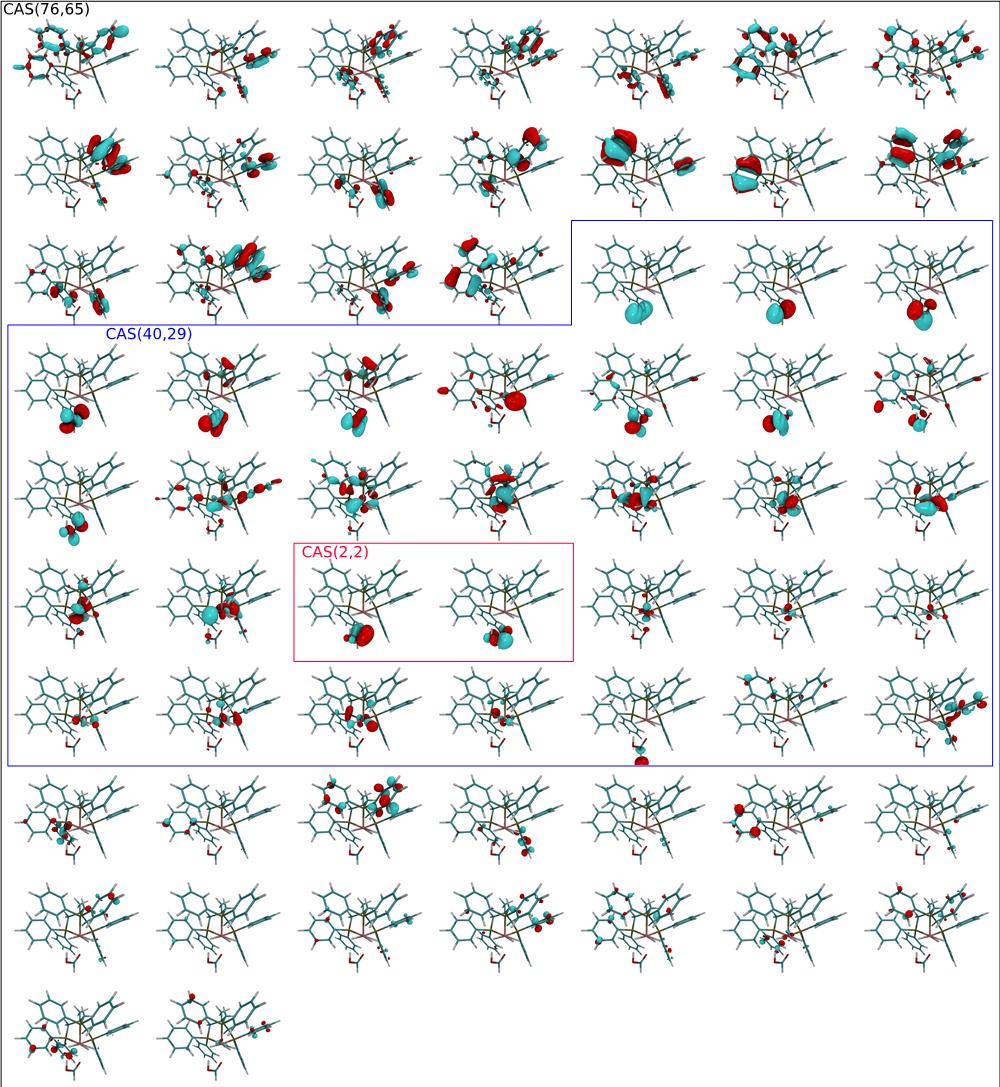}
		\caption{Active spaces yielding integral files for complex structure VIII produced from CAS(2,2)SCF molecular orbitals in the full atomic orbital basis.}
	\end{figure*}
	
	\newpage

	\subsection{CAS(4,4)SCF active molecular orbitals chosen for integral files of transition state VIII-IX}
	\begin{figure*}[htb]
		\centering
		\includegraphics[width=0.9\textwidth]{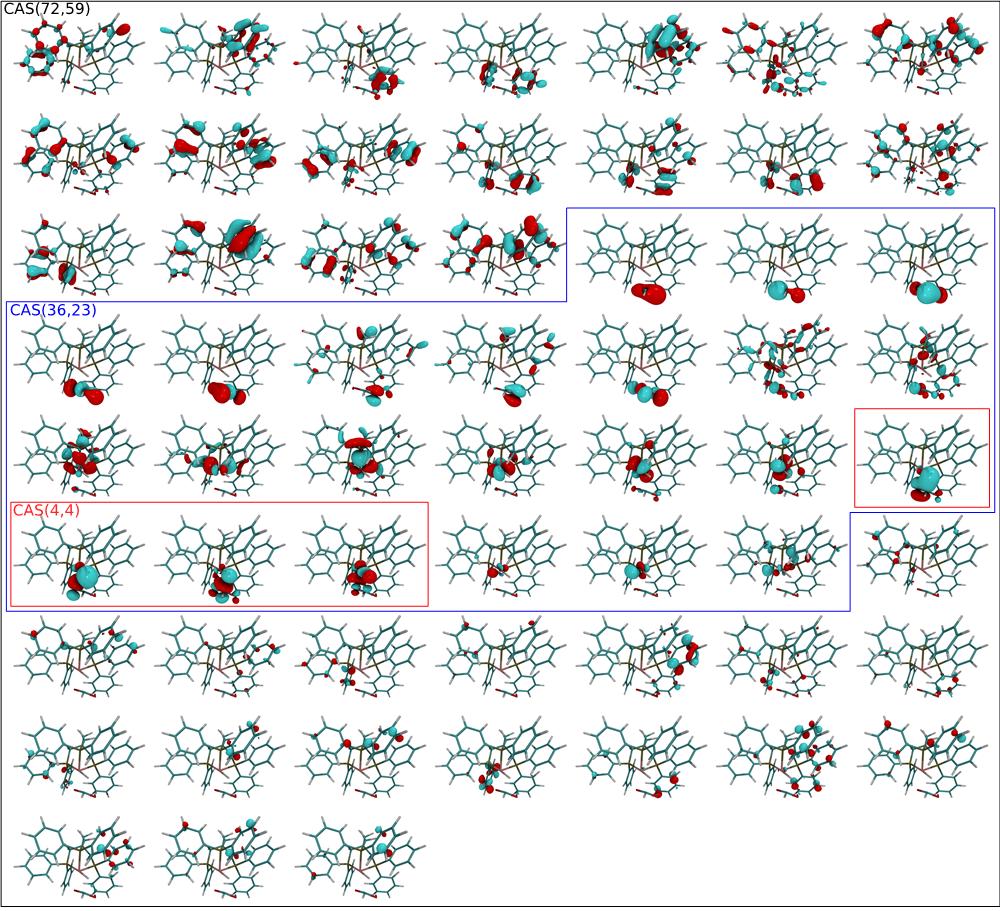}
		\caption{Active spaces yielding integral files for complex structure VIII-IX produced from
			CAS(4,4)SCF molecular orbitals in the full atomic orbital basis.}
	\end{figure*}

	\newpage

	\subsection{CAS(16,16)SCF active molecular orbitals chosen for integral files of stable intermediate IX}
	\begin{figure*}[htb]
		\centering
		\includegraphics[width=0.85\textwidth]{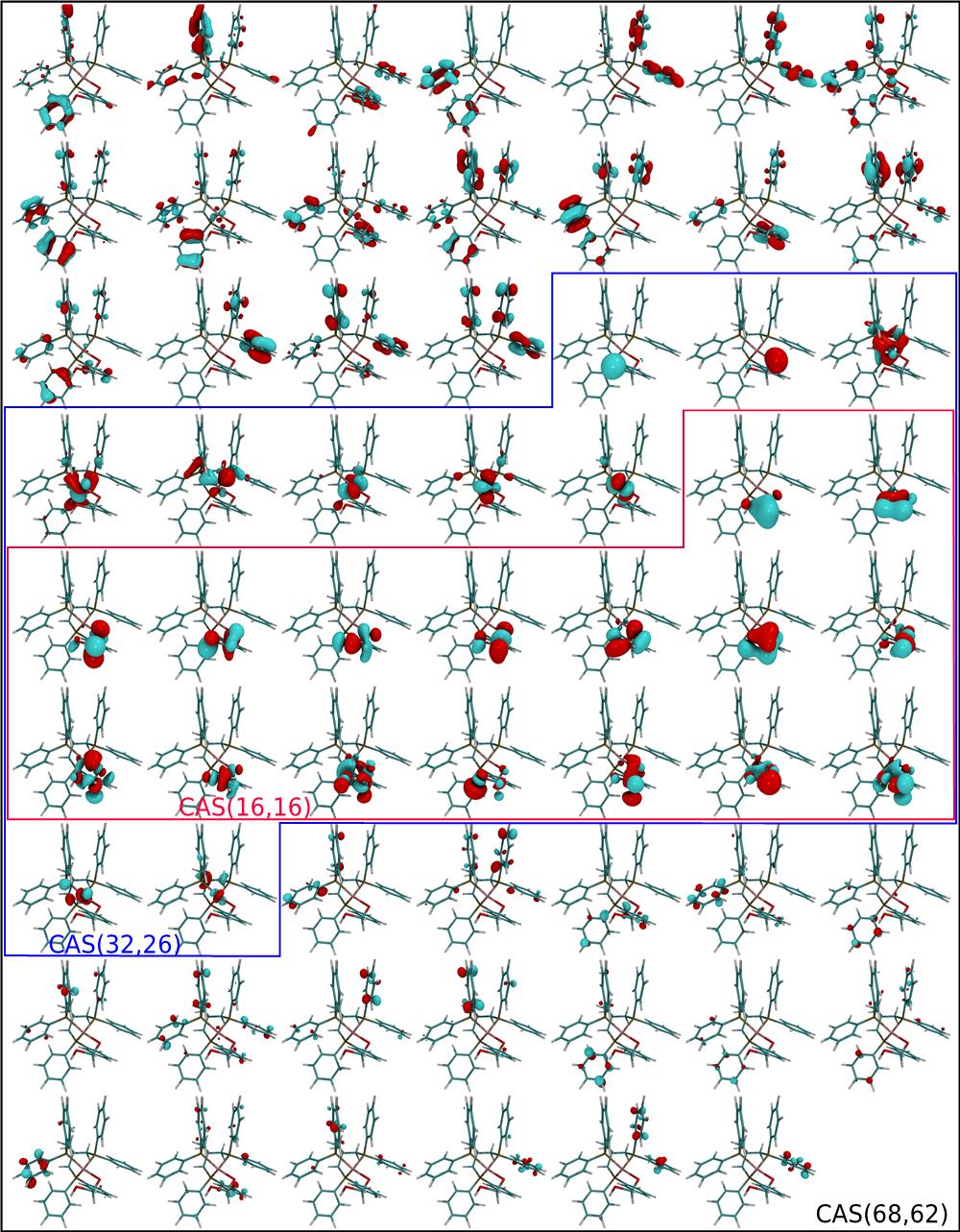}
		\caption{Active spaces yielding integral files for complex structure IX produced from
			CAS(16,16)SCF molecular orbitals in the full atomic orbital basis.}
	\end{figure*}

	\newpage

	\subsection{CAS(4,4)SCF active molecular orbitals chosen for integral files of stable intermediate XVIII}
	\begin{figure*}[htb]
		\centering
		\includegraphics[width=0.9\textwidth]{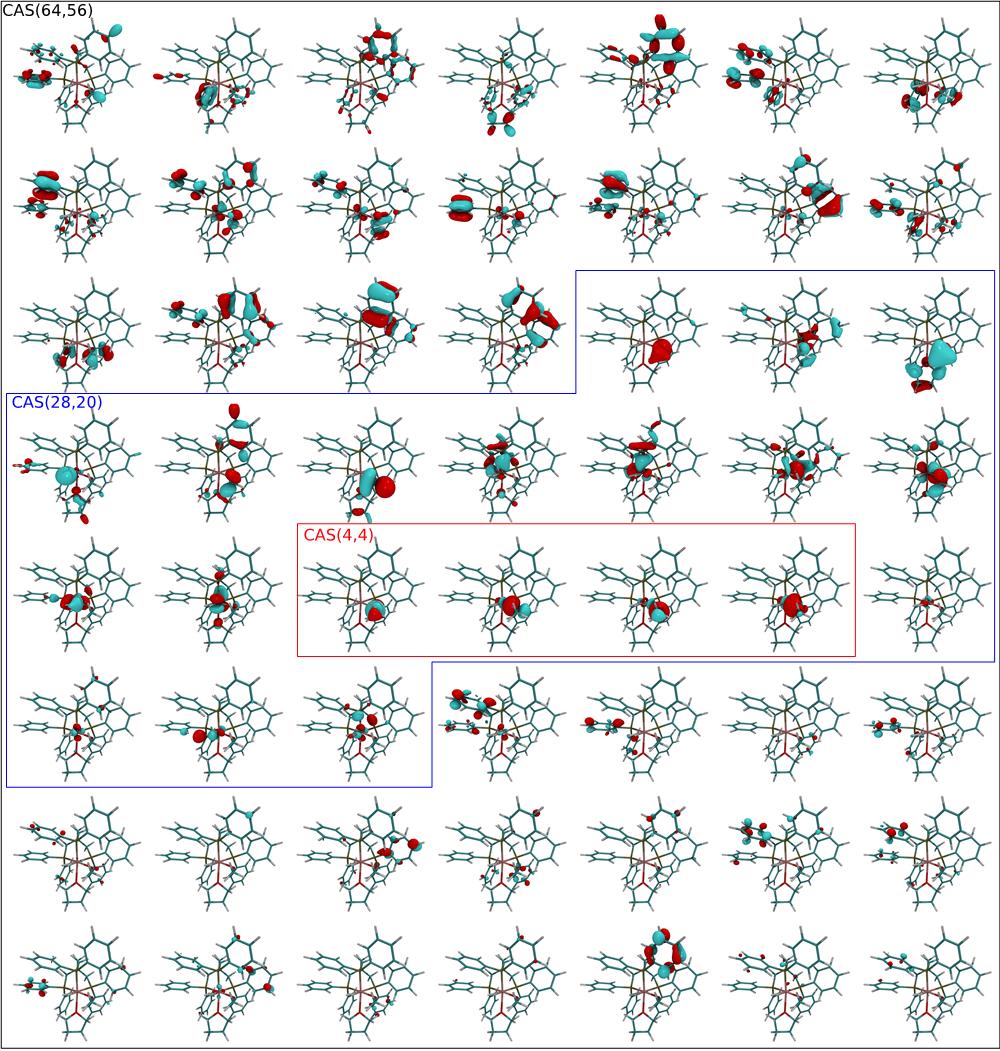}
		\caption{Active spaces yielding integral files for complex structure XVIII produced from CAS(4,4)SCF molecular orbitals in the full atomic orbital basis.}
	\end{figure*}

	\newpage

	\section{Resource analysis}
	
	There exists a plethora of quantum algorithms that allow one to obtain energies of a given Hamiltonian. Our previous work~\cite{reiher17} focused on a product-formula based implementation of the time evolution operator, where the Hamiltonian was given in a standard second-quantized representation. In this work, we evaluated a variety of different approaches and focused on the ones that yield the lowest resource costs.
	
	An initial screening of methods left us with the following methods to consider:
	\begin{itemize}
		\item Trotter based implementation of the low-rank factorized Hamiltonian~\cite{Motta2018LowRankTrotter}
		\item Qubitization of the standard (unfactorized) second-quantized representation~\cite{low2019qsharp}
		\item Qubitization of the single-factorized Hamiltonian~\cite{Berry2019CholeskyQubitization}
		\item Qubitization of the double-factorized Hamiltonian (this work)
	\end{itemize}
	We found that a qubitization based implementation of a doubly-factorized Hamiltonian resulted in the lowest costs.
	This method is discussed in the main text with a comparison to the single-factorized Hamiltonian, and also thoroughly detailed in~\cref{sec:qubitization_electronic_structure} of this supporting document.
	We also ruled out qubitization in the standard and single-factorized second-quantized representation based on highly-optimized algorithms targeting these representations~\cite{Berry2019CholeskyQubitization}.
	In the remainder of this section, we summarize our evaluation of remaining Trotter based implementation, and compare resource estimates in~\cref{tab:resest} with the best qubitization implementations.
	Finally, we tabulate resource estimates of all carbon dioxide fixation complexes mentioned in the main text.
	%
	
	\subsection{Product formula and low-rank factorization of the Hamiltonian}
	\label{sec:productformulas}
	This method was introduced in~\cite{Motta2018LowRankTrotter}. The idea is to approximate the Hamiltonian using a low-rank factorization, i.e., using $R < N^2$ components, and to then simulate time-evolution using a product formula. Each step of the product formula based implementation consists of a sequence of rotations that diagonalize the current component of the low-rank factorization, followed by a sequence of controlled phase gates that correspond to the eigenvalues of the component. This is repeated for all $R$ components of the low-rank factorization. As the simulation basis size $N$ is increased, it has been shown that $R\sim\mathcal O(N)$~\cite{Peng2017Cholesky}. Additionally, the eigenvalues of each component in the factorization may be truncated, allowing one to further reduce the resource requirements.
	For a more detailed description of this method, we refer the reader to~\cite{Motta2018LowRankTrotter}.
	
	We provide gate counts for estimating the electronic energies of some of our chosen systems using this method in~\cref{tab:resest}. The resource estimation uses the same truncation method as discussed in Section 5.1 of the main text. We note that this product formula based implementation naturally offers more parallelism than, e.g., a qubitization based implementation at comparable qubit numbers. We leave the detailed investigation of such space/time tradeoffs for future work.
	
	The time step for the second order Trotter method was chosen using the following upper bound~\cite{kivlichan2019improved}
	\begin{align} \|{e^{-iHt}-e^{-iH_{\text{eff}}t}}\|\le\frac{t^{3}}{12}\sum_{p}\left(\sum_{c>b}\sum_{a>b}\|{\left[H_{a},\left[H_{b},H_{c}\right]\right]}\|+\frac{1}{2}\sum_{c>b}\|{\left[H_{b},\left[H_{b},H_{c}\right]\right]}\|\right),
	\end{align}
	which applies to any Hamiltonian $H=\sum_{j=0}^{M-1}H_{j}$ where time-evolution by each term can be applied without any error.
	In the context of electronic structure, time-evolution by the one-electron term and each rank component of the two-electron terms can be applied exactly by diagonalizing them with fermionic basis transformations that are implemented by Givens rotations~\cite{Wecker2015Strongly}.
	
	Since such bounds have been shown to be loose by many orders of magnitude, we carried out an empirical study of the required Trotter step size for linear chains of Hydrogen atoms. We depict the bound and empirically determined step sizes in Fig.~\ref{fig:hchainsbound}. The results for linear chains of Hydrogen atoms indicate that the upper bound yields Trotter numbers that may be up to $10 \sqrt{n}$ larger than what empirical data suggests. Rescaling the $T$-count estimates for our catalyst systems by this factor results in an optimistic estimate that can also be found in Table~\ref{tab:resest}.
	
	\begin{figure}[ht]
		\centering
		\resizebox{0.7\linewidth}{!}{\input{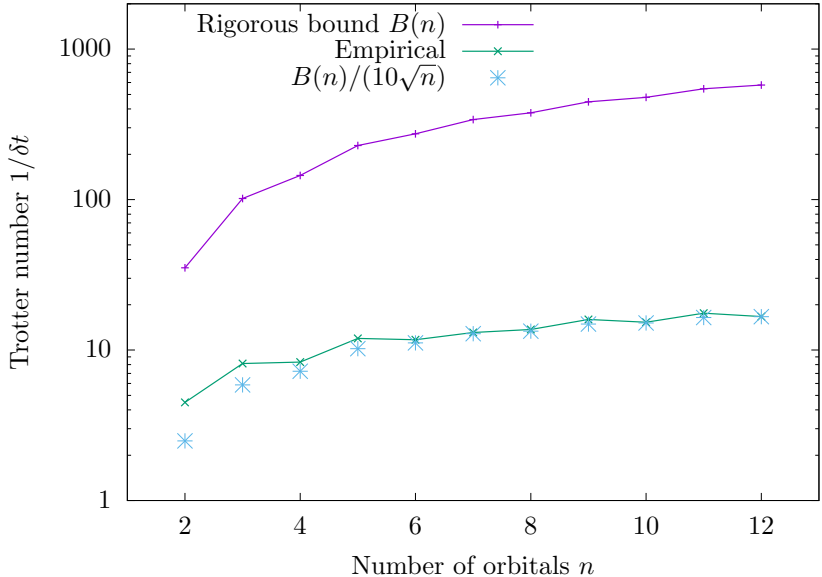}}
		\caption{Comparing Trotter numbers $1/\delta t$, where $\delta t$ is the Trotter step size, derived from simulations (labeled ``empirical'') to rigorous upper bounds. The gap between the empirical data and the bound is approximately $10\sqrt{n}$, as indicated by the rescaled points.}
		\label{fig:hchainsbound}
	\end{figure}

	\subsection{Results}
	In this section, we compare the $T$-gate and qubit counts of Trotter- and qubitization- based each method in Table~\ref{tab:resest} -- 
	note that any Toffoli gate can be synthesized from 4 $T$ gates. Qubitization using a doubly-factorized representation of the Hamiltonian offers the lowest $T$-counts at reasonable qubit numbers.
	The main text therefore focuses on this approach, and thus we tabulate, for completeness, the cost of phase estimation to $1$mHartree and problem parameters of all the carbon dioxide fixation complex structures mentioned in the main text using this approach.
	We also briefly compare with qubitization using the unfactorized or single-factorized representation using the algorithms of Berry et al.~\cite{Berry2019CholeskyQubitization}.
	In all these examples, we reduce the number of non-zero terms by using the incoherent truncation scheme discussed in the main text at an error threshold of $\epsilon_{in}=1$mHartree.
	
	\begin{table}[ht]
		\caption{Comparison of resource estimates for estimating electronic energies on a quantum computer. Results for Trotter (DF) follow from~\cref{sec:productformulas}, results for Qubitization (DF) follow from~\cref{sec:qubitization_electronic_structure}, and results for unfactorized and single-factorized qubitization follow from the algorithms described by Berry et al.~\cite{Berry2019CholeskyQubitization}\label{tab:resest}}
		\centering
		\begin{ruledtabular}\begin{tabular}{ccccccccccc}
				&			&	&	\multicolumn{3}{c}{\bfseries Trotter (DF)} & \multicolumn{3}{c}{\bfseries Qubitization (DF)} \\
				&&&&&&$\alpha$&&&&\\
				Step &  R & M & Qubits & $T$-count & optimistic&/ Hartree & Qubits & Toffoli-count\\
				\hline
				I-cas5-fb-48e52o	&	613&	23566&	104&	$8.29\times 10^{14}$&	$1.15\times 10^{13}$ & 177.3 & 3400 & $1.3\times 10^{10}$\\
				II-cas6-fb-70e62o	&	734&	33629&	124&	$3.50\times 10^{15}$	&$4.45\times 10^{13}$ & 374.4 & 4200 & $3.6\times 10^{10}$ \\
				II-III-cas6-fb-74e65o&		783&	38122&	130	&$4.17\times 10^{15}$&	$5.18\times 10^{13}$ & 416.0 & 4400 & $4.5\times 10^{10}$\\
				V-cas11-fb-68e60o	 &	670&	29319&	120&	$2.76\times 10^{15}$&	$3.57\times 10^{13}$ & 372.1 & 4100 & $3.3\times 10^{10}$\\
				VIII-cas2-fb-76e65o	&	794&39088&	130	&$4.52\times 10^{15}$	&$5.61\times 10^{13}$ & 425.7 & 4400 & $4.6\times 10^{10}$\\
				VIII-IX-cas4-fb-72e59o	&	666	&29286&	118&	$2.58\times 10^{15}$&	$3.36\times 10^{13}$ & 384.4 & 4000 & $3.4\times 10^{10}$\\
				IX-cas16-fb-68e62o	 &	638&	28945&	124	&$2.97\times 10^{15}$&	$3.77\times 10^{13}$ & 396.6 & 4200 & $3.5\times 10^{10}$\\
				XVIII-cas4-fb-64e56o	&	705	&29594&	112	&$2.09\times 10^{15}$	&$2.80\times 10^{13}$ & 293.5 & 3700 & $2.5\times 10^{10}$\\
			\end{tabular}
		\end{ruledtabular}
		\newline
		\begin{ruledtabular}
			\begin{tabular}{ccccccccc}
				&	\multicolumn{4}{c}{\bfseries Qubitization (Unfactorized)} & \multicolumn{4}{c}{\bfseries Qubitization (Single-factorized)} \\
				&   & $\alpha$ &  & Toffoli & & $\alpha$ & & Toffoli\\
				Step &n.n.z& / Hartree&Qubits&Count&n.n.z& / Hartree&Qubits&Count\\
				\hline
				I-cas5-fb-48e52o & 911134 & 5775 & 9832 & $ 2.8 \times 10^{11}$ & 628549 & 19725 & 5480 & $ 7.0 \times 10^{11}$ \\
				II-cas6-fb-70e62o & 1757237 & 8683 & 9980 & $ 7.9 \times 10^{11}$ & 1052858 & 31790 & 5628 & $ 1.9 \times 10^{12}$ \\
				II-III-cas6-fb-74e65o & 2181633 & 10714 & 11010 & $ 9.3 \times 10^{11}$ & 1251855 & 42037 & 5762 & $ 2.1 \times 10^{12}$ \\
				V-cas11-fb-68e60o & 1526106 & 7332 & 9976 & $ 3.6 \times 10^{11}$ & 887611 & 25725 & 5624 & $ 8.4 \times 10^{11}$ \\
				VIII-cas2-fb-76e65o & 2169817 & 10617 & 11010 & $ 9.3 \times 10^{11}$ & 1292050 & 42160 & 5762 & $ 2.1 \times 10^{12}$ \\
				VIII-IX-cas4-fb-72e59o & 1485406 & 8950 & 9974 & $ 7.2 \times 10^{11}$ & 880955 & 33100 & 5622 & $ 1.7 \times 10^{12}$ \\
				IX-cas16-fb-68e62o & 1751213 & 9651 & 9980 & $ 7.9 \times 10^{11}$ & 897570 & 34973 & 5628 & $ 1.7 \times 10^{12}$ \\
				XVIII-cas4-fb-64e56o & 1213076 & 6708 & 9968 & $ 3.2 \times 10^{11}$ & 847812 & 25611 & 5616 & $ 8.2 \times 10^{11}$
			\end{tabular}
		\end{ruledtabular}
		
	\end{table}

	\begin{table}[H]
		\caption{Double-factorization resource estimates for all steps of carbon dioxide fixation for active space sizes from $52$--$65$ orbitals at a truncation threshold between $\epsilon_{\text{in}}=0.1$mHartree and $0.1$Hartree.}
		\begin{ruledtabular}\begin{tabular}{lccccccc}
				\multirow{2}{*}{Step} & $\epsilon_{\text{in}}$ & Orbitals &\multirow{2}{*}{$R$} &\multirow{2}{*}{$M$} & $\alpha_{DF}$ & \multirow{2}{*}{Qubits} & $\#$Toffoli\\
				& / Hartree  & $N$ & &&/ Hartree & & gates\\
				\hline
				I-cas5-fb-48e52o & $ 1.00 \times 10^{-4}$ & 52 & 855 & 37242 & 177.4 & 3448 & $ 1.72 \times 10^{10} $ \\ 
				I-cas5-fb-48e52o & $ 1.58 \times 10^{-4}$ & 52 & 808 & 34633 & 177.4 & 3448 & $ 1.63 \times 10^{10} $ \\ 
				I-cas5-fb-48e52o & $ 2.51 \times 10^{-4}$ & 52 & 761 & 31937 & 177.4 & 3447 & $ 1.55 \times 10^{10} $ \\ 
				I-cas5-fb-48e52o & $ 3.98 \times 10^{-4}$ & 52 & 713 & 29187 & 177.4 & 3447 & $ 1.46 \times 10^{10} $ \\ 
				I-cas5-fb-48e52o & $ 6.31 \times 10^{-4}$ & 52 & 665 & 26399 & 177.4 & 3447 & $ 1.38 \times 10^{10} $ \\ 
				I-cas5-fb-48e52o & $ 1.00 \times 10^{-3}$ & 52 & 613 & 23566 & 177.3 & 3447 & $ 1.29 \times 10^{10} $ \\ 
				I-cas5-fb-48e52o & $ 1.58 \times 10^{-3}$ & 52 & 561 & 20739 & 177.3 & 3447 & $ 1.20 \times 10^{10} $ \\ 
				I-cas5-fb-48e52o & $ 2.51 \times 10^{-3}$ & 52 & 509 & 17950 & 177.2 & 3447 & $ 1.11 \times 10^{10} $ \\ 
				I-cas5-fb-48e52o & $ 3.98 \times 10^{-3}$ & 52 & 457 & 15198 & 177.2 & 3446 & $ 1.02 \times 10^{10} $ \\ 
				I-cas5-fb-48e52o & $ 6.31 \times 10^{-3}$ & 52 & 403 & 12516 & 177.0 & 3446 & $ 9.38 \times 10^{9} $ \\ 
				I-cas5-fb-48e52o & $ 1.00 \times 10^{-2}$ & 52 & 343 & 9993 & 176.8 & 3446 & $ 8.57 \times 10^{9} $ \\ 
				I-cas5-fb-48e52o & $ 1.58 \times 10^{-2}$ & 52 & 280 & 7695 & 176.3 & 3445 & $ 7.83 \times 10^{9} $ \\ 
				I-cas5-fb-48e52o & $ 2.51 \times 10^{-2}$ & 52 & 224 & 5739 & 175.7 & 3445 & $ 7.18 \times 10^{9} $ \\ 
				I-cas5-fb-48e52o & $ 3.98 \times 10^{-2}$ & 52 & 167 & 4161 & 174.7 & 3445 & $ 6.65 \times 10^{9} $ \\ 
				I-cas5-fb-48e52o & $ 6.31 \times 10^{-2}$ & 52 & 125 & 2959 & 173.3 & 3444 & $ 6.22 \times 10^{9} $ \\ 
				I-cas5-fb-48e52o & $ 1.00 \times 10^{-1}$ & 52 & 93 & 2117 & 171.5 & 3444 & $ 5.89 \times 10^{9} $ \\ \hline
				II-cas6-fb-70e62o & $ 1.00 \times 10^{-4}$ & 62 & 998 & 51731 & 374.5 & 4232 & $ 4.84 \times 10^{10} $ \\ 
				II-cas6-fb-70e62o & $ 1.58 \times 10^{-4}$ & 62 & 947 & 48215 & 374.5 & 4232 & $ 4.60 \times 10^{10} $ \\ 
				II-cas6-fb-70e62o & $ 2.51 \times 10^{-4}$ & 62 & 893 & 44655 & 374.4 & 4232 & $ 4.37 \times 10^{10} $ \\ 
				II-cas6-fb-70e62o & $ 3.98 \times 10^{-4}$ & 62 & 840 & 41036 & 374.4 & 4232 & $ 4.13 \times 10^{10} $ \\ 
				II-cas6-fb-70e62o & $ 6.31 \times 10^{-4}$ & 62 & 790 & 37348 & 374.4 & 4232 & $ 3.89 \times 10^{10} $ \\ 
				II-cas6-fb-70e62o & $ 1.00 \times 10^{-3}$ & 62 & 734 & 33629 & 374.4 & 4232 & $ 3.64 \times 10^{10} $ \\ 
				II-cas6-fb-70e62o & $ 1.58 \times 10^{-3}$ & 62 & 678 & 29874 & 374.4 & 4231 & $ 3.39 \times 10^{10} $ \\ 
				II-cas6-fb-70e62o & $ 2.51 \times 10^{-3}$ & 62 & 622 & 26143 & 374.3 & 4231 & $ 3.14 \times 10^{10} $ \\ 
				II-cas6-fb-70e62o & $ 3.98 \times 10^{-3}$ & 62 & 555 & 22497 & 374.2 & 4231 & $ 2.90 \times 10^{10} $ \\ 
				II-cas6-fb-70e62o & $ 6.31 \times 10^{-3}$ & 62 & 491 & 19012 & 374.0 & 4231 & $ 2.67 \times 10^{10} $ \\ 
				II-cas6-fb-70e62o & $ 1.00 \times 10^{-2}$ & 62 & 432 & 15702 & 373.7 & 4230 & $ 2.45 \times 10^{10} $ \\ 
				II-cas6-fb-70e62o & $ 1.58 \times 10^{-2}$ & 62 & 370 & 12580 & 373.3 & 4230 & $ 2.24 \times 10^{10} $ \\ 
				II-cas6-fb-70e62o & $ 2.51 \times 10^{-2}$ & 62 & 309 & 9700 & 372.6 & 4230 & $ 2.04 \times 10^{10} $ \\ 
				II-cas6-fb-70e62o & $ 3.98 \times 10^{-2}$ & 62 & 248 & 7150 & 371.4 & 4229 & $ 1.87 \times 10^{10} $ \\ 
				II-cas6-fb-70e62o & $ 6.31 \times 10^{-2}$ & 62 & 183 & 5096 & 369.3 & 4229 & $ 1.72 \times 10^{10} $ \\ 
				II-cas6-fb-70e62o & $ 1.00 \times 10^{-1}$ & 62 & 136 & 3609 & 366.8 & 4228 & $ 1.61 \times 10^{10} $ \\   
		\end{tabular}\end{ruledtabular}
	\end{table}
	
	\begin{table}[htb]
		\caption{(Continued) Double-factorization resource estimates for all steps of carbon dioxide fixation for active space sizes from $52$--$65$ orbitals at a truncation threshold between $\epsilon_{\text{in}}=0.1$mHartree and $0.1$Hartree.}
		\begin{ruledtabular}
			\begin{tabular}{lccccccc}
				\multirow{2}{*}{Step} & $\epsilon_{\text{in}}$ & Orbitals &\multirow{2}{*}{$R$} &\multirow{2}{*}{$M$} & $\alpha_{DF}$ & \multirow{2}{*}{Qubits} & $\#$Toffoli\\
				& / Hartree  & $N$ & &&/ Hartree & & gates\\
				\hline
				II-III-cas6-fb-74e65o & $ 1.00 \times 10^{-4}$ & 65 & 1064 & 58020 & 416.1 & 4436 & $ 5.91 \times 10^{10} $ \\ 
				II-III-cas6-fb-74e65o & $ 1.58 \times 10^{-4}$ & 65 & 1011 & 54104 & 416.1 & 4436 & $ 5.62 \times 10^{10} $ \\ 
				II-III-cas6-fb-74e65o & $ 2.51 \times 10^{-4}$ & 65 & 954 & 50154 & 416.1 & 4436 & $ 5.33 \times 10^{10} $ \\ 
				II-III-cas6-fb-74e65o & $ 3.98 \times 10^{-4}$ & 65 & 895 & 46172 & 416.0 & 4436 & $ 5.04 \times 10^{10} $ \\ 
				II-III-cas6-fb-74e65o & $ 6.31 \times 10^{-4}$ & 65 & 839 & 42153 & 416.0 & 4436 & $ 4.75 \times 10^{10} $ \\ 
				II-III-cas6-fb-74e65o & $ 1.00 \times 10^{-3}$ & 65 & 783 & 38122 & 416.0 & 4436 & $ 4.45 \times 10^{10} $ \\ 
				II-III-cas6-fb-74e65o & $ 1.58 \times 10^{-3}$ & 65 & 725 & 34075 & 416.0 & 4436 & $ 4.15 \times 10^{10} $ \\ 
				II-III-cas6-fb-74e65o & $ 2.51 \times 10^{-3}$ & 65 & 665 & 30026 & 415.9 & 4435 & $ 3.85 \times 10^{10} $ \\ 
				II-III-cas6-fb-74e65o & $ 3.98 \times 10^{-3}$ & 65 & 604 & 26032 & 415.8 & 4435 & $ 3.56 \times 10^{10} $ \\ 
				II-III-cas6-fb-74e65o & $ 6.31 \times 10^{-3}$ & 65 & 536 & 22119 & 415.6 & 4435 & $ 3.27 \times 10^{10} $ \\ 
				II-III-cas6-fb-74e65o & $ 1.00 \times 10^{-2}$ & 65 & 462 & 18383 & 415.2 & 4435 & $ 2.99 \times 10^{10} $ \\ 
				II-III-cas6-fb-74e65o & $ 1.58 \times 10^{-2}$ & 65 & 400 & 14890 & 414.8 & 4434 & $ 2.73 \times 10^{10} $ \\ 
				II-III-cas6-fb-74e65o & $ 2.51 \times 10^{-2}$ & 65 & 339 & 11652 & 414.1 & 4434 & $ 2.49 \times 10^{10} $ \\ 
				II-III-cas6-fb-74e65o & $ 3.98 \times 10^{-2}$ & 65 & 273 & 8757 & 412.7 & 4434 & $ 2.27 \times 10^{10} $ \\ 
				II-III-cas6-fb-74e65o & $ 6.31 \times 10^{-2}$ & 65 & 213 & 6312 & 410.7 & 4433 & $ 2.08 \times 10^{10} $ \\ 
				II-III-cas6-fb-74e65o & $ 1.00 \times 10^{-1}$ & 65 & 160 & 4463 & 407.9 & 4433 & $ 1.93 \times 10^{10} $ \\ \hline
				V-cas11-fb-68e60o & $ 1.00 \times 10^{-4}$ & 60 & 938 & 46425 & 372.2 & 4096 & $ 4.41 \times 10^{10} $ \\ 
				V-cas11-fb-68e60o & $ 1.58 \times 10^{-4}$ & 60 & 883 & 43035 & 372.2 & 4096 & $ 4.19 \times 10^{10} $ \\ 
				V-cas11-fb-68e60o & $ 2.51 \times 10^{-4}$ & 60 & 828 & 39623 & 372.2 & 4096 & $ 3.97 \times 10^{10} $ \\ 
				V-cas11-fb-68e60o & $ 3.98 \times 10^{-4}$ & 60 & 778 & 36188 & 372.2 & 4096 & $ 3.74 \times 10^{10} $ \\ 
				V-cas11-fb-68e60o & $ 6.31 \times 10^{-4}$ & 60 & 724 & 32739 & 372.1 & 4095 & $ 3.52 \times 10^{10} $ \\ 
				V-cas11-fb-68e60o & $ 1.00 \times 10^{-3}$ & 60 & 670 & 29319 & 372.1 & 4095 & $ 3.29 \times 10^{10} $ \\ 
				V-cas11-fb-68e60o & $ 1.58 \times 10^{-3}$ & 60 & 612 & 25946 & 372.1 & 4095 & $ 3.07 \times 10^{10} $ \\ 
				V-cas11-fb-68e60o & $ 2.51 \times 10^{-3}$ & 60 & 558 & 22638 & 372.0 & 4095 & $ 2.85 \times 10^{10} $ \\ 
				V-cas11-fb-68e60o & $ 3.98 \times 10^{-3}$ & 60 & 501 & 19412 & 371.9 & 4095 & $ 2.64 \times 10^{10} $ \\ 
				V-cas11-fb-68e60o & $ 6.31 \times 10^{-3}$ & 60 & 444 & 16335 & 371.8 & 4094 & $ 2.43 \times 10^{10} $ \\ 
				V-cas11-fb-68e60o & $ 1.00 \times 10^{-2}$ & 60 & 383 & 13448 & 371.5 & 4094 & $ 2.24 \times 10^{10} $ \\ 
				V-cas11-fb-68e60o & $ 1.58 \times 10^{-2}$ & 60 & 329 & 10819 & 371.1 & 4094 & $ 2.07 \times 10^{10} $ \\ 
				V-cas11-fb-68e60o & $ 2.51 \times 10^{-2}$ & 60 & 273 & 8462 & 370.5 & 4094 & $ 1.91 \times 10^{10} $ \\ 
				V-cas11-fb-68e60o & $ 3.98 \times 10^{-2}$ & 60 & 222 & 6436 & 369.5 & 4093 & $ 1.77 \times 10^{10} $ \\ 
				V-cas11-fb-68e60o & $ 6.31 \times 10^{-2}$ & 60 & 174 & 4781 & 368.0 & 4093 & $ 1.65 \times 10^{10} $ \\ 
				V-cas11-fb-68e60o & $ 1.00 \times 10^{-1}$ & 60 & 138 & 3519 & 366.1 & 4092 & $ 1.56 \times 10^{10} $ \\ 
			\end{tabular}
		\end{ruledtabular}
	\end{table}

	\begin{table}
		\caption{(Continued) Double-factorization resource estimates for all steps of carbon dioxide fixation for active space sizes from $52$--$65$ orbitals at a truncation threshold between $\epsilon_{\text{in}}=0.1$mHartree and $0.1$Hartree.}
		\begin{ruledtabular}\begin{tabular}{lccccccc}
				\multirow{2}{*}{Step} & $\epsilon_{\text{in}}$ & Orbitals &\multirow{2}{*}{$R$} &\multirow{2}{*}{$M$} & $\alpha_{DF}$ & \multirow{2}{*}{Qubits} & $\#$Toffoli\\
				& / Hartree  & $N$ & &&/ Hartree & & gates\\
				\hline
				VIII-cas2-fb-76e65o & $ 1.00 \times 10^{-4}$ & 65 & 1081 & 59349 & 425.8 & 4436 & $ 6.15 \times 10^{10} $ \\ 
				VIII-cas2-fb-76e65o & $ 1.58 \times 10^{-4}$ & 65 & 1028 & 55420 & 425.8 & 4436 & $ 5.85 \times 10^{10} $ \\ 
				VIII-cas2-fb-76e65o & $ 2.51 \times 10^{-4}$ & 65 & 969 & 51442 & 425.8 & 4436 & $ 5.55 \times 10^{10} $ \\ 
				VIII-cas2-fb-76e65o & $ 3.98 \times 10^{-4}$ & 65 & 914 & 47379 & 425.8 & 4436 & $ 5.25 \times 10^{10} $ \\ 
				VIII-cas2-fb-76e65o & $ 6.31 \times 10^{-4}$ & 65 & 856 & 43246 & 425.7 & 4436 & $ 4.94 \times 10^{10} $ \\ 
				VIII-cas2-fb-76e65o & $ 1.00 \times 10^{-3}$ & 65 & 794 & 39088 & 425.7 & 4436 & $ 4.63 \times 10^{10} $ \\ 
				VIII-cas2-fb-76e65o & $ 1.58 \times 10^{-3}$ & 65 & 735 & 34927 & 425.7 & 4436 & $ 4.31 \times 10^{10} $ \\ 
				VIII-cas2-fb-76e65o & $ 2.51 \times 10^{-3}$ & 65 & 675 & 30767 & 425.6 & 4435 & $ 4.00 \times 10^{10} $ \\ 
				VIII-cas2-fb-76e65o & $ 3.98 \times 10^{-3}$ & 65 & 615 & 26636 & 425.5 & 4435 & $ 3.69 \times 10^{10} $ \\ 
				VIII-cas2-fb-76e65o & $ 6.31 \times 10^{-3}$ & 65 & 548 & 22559 & 425.3 & 4435 & $ 3.38 \times 10^{10} $ \\ 
				VIII-cas2-fb-76e65o & $ 1.00 \times 10^{-2}$ & 65 & 483 & 18623 & 425.0 & 4435 & $ 3.08 \times 10^{10} $ \\ 
				VIII-cas2-fb-76e65o & $ 1.58 \times 10^{-2}$ & 65 & 414 & 14912 & 424.5 & 4434 & $ 2.80 \times 10^{10} $ \\ 
				VIII-cas2-fb-76e65o & $ 2.51 \times 10^{-2}$ & 65 & 344 & 11494 & 423.7 & 4434 & $ 2.54 \times 10^{10} $ \\ 
				VIII-cas2-fb-76e65o & $ 3.98 \times 10^{-2}$ & 65 & 277 & 8424 & 422.4 & 4434 & $ 2.30 \times 10^{10} $ \\ 
				VIII-cas2-fb-76e65o & $ 6.31 \times 10^{-2}$ & 65 & 206 & 5912 & 420.1 & 4433 & $ 2.10 \times 10^{10} $ \\ 
				VIII-cas2-fb-76e65o & $ 1.00 \times 10^{-1}$ & 65 & 146 & 4112 & 416.8 & 4433 & $ 1.95 \times 10^{10} $ \\ \hline
				VIII-IX-cas4-fb-72e59o & $ 1.00 \times 10^{-4}$ & 59 & 915 & 45289 & 384.4 & 4028 & $ 4.46 \times 10^{10} $ \\ 
				VIII-IX-cas4-fb-72e59o & $ 1.58 \times 10^{-4}$ & 59 & 866 & 42119 & 384.4 & 4028 & $ 4.24 \times 10^{10} $ \\ 
				VIII-IX-cas4-fb-72e59o & $ 2.51 \times 10^{-4}$ & 59 & 819 & 38956 & 384.4 & 4028 & $ 4.03 \times 10^{10} $ \\ 
				VIII-IX-cas4-fb-72e59o & $ 3.98 \times 10^{-4}$ & 59 & 768 & 35764 & 384.4 & 4028 & $ 3.81 \times 10^{10} $ \\ 
				VIII-IX-cas4-fb-72e59o & $ 6.31 \times 10^{-4}$ & 59 & 712 & 32538 & 384.4 & 4027 & $ 3.59 \times 10^{10} $ \\ 
				VIII-IX-cas4-fb-72e59o & $ 1.00 \times 10^{-3}$ & 59 & 666 & 29286 & 384.4 & 4027 & $ 3.37 \times 10^{10} $ \\ 
				VIII-IX-cas4-fb-72e59o & $ 1.58 \times 10^{-3}$ & 59 & 609 & 26030 & 384.3 & 4027 & $ 3.15 \times 10^{10} $ \\ 
				VIII-IX-cas4-fb-72e59o & $ 2.51 \times 10^{-3}$ & 59 & 553 & 22780 & 384.3 & 4027 & $ 2.93 \times 10^{10} $ \\ 
				VIII-IX-cas4-fb-72e59o & $ 3.98 \times 10^{-3}$ & 59 & 496 & 19598 & 384.1 & 4027 & $ 2.71 \times 10^{10} $ \\ 
				VIII-IX-cas4-fb-72e59o & $ 6.31 \times 10^{-3}$ & 59 & 443 & 16506 & 384.0 & 4027 & $ 2.50 \times 10^{10} $ \\ 
				VIII-IX-cas4-fb-72e59o & $ 1.00 \times 10^{-2}$ & 59 & 386 & 13588 & 383.7 & 4026 & $ 2.30 \times 10^{10} $ \\ 
				VIII-IX-cas4-fb-72e59o & $ 1.58 \times 10^{-2}$ & 59 & 330 & 10863 & 383.3 & 4026 & $ 2.11 \times 10^{10} $ \\ 
				VIII-IX-cas4-fb-72e59o & $ 2.51 \times 10^{-2}$ & 59 & 275 & 8405 & 382.7 & 4026 & $ 1.94 \times 10^{10} $ \\ 
				VIII-IX-cas4-fb-72e59o & $ 3.98 \times 10^{-2}$ & 59 & 222 & 6264 & 381.6 & 4025 & $ 1.79 \times 10^{10} $ \\ 
				VIII-IX-cas4-fb-72e59o & $ 6.31 \times 10^{-2}$ & 59 & 178 & 4558 & 380.2 & 4025 & $ 1.67 \times 10^{10} $ \\ 
				VIII-IX-cas4-fb-72e59o & $ 1.00 \times 10^{-1}$ & 59 & 128 & 3327 & 377.6 & 4024 & $ 1.57 \times 10^{10} $ \\ 
		\end{tabular}\end{ruledtabular}
	\end{table}
	
	\begin{table}
		\caption{(Continued) Double-factorization resource estimates for all steps of carbon dioxide fixation for active space sizes from $52$--$65$ orbitals at a truncation threshold between $\epsilon_{\text{in}}=0.1$mHartree and $0.1$Hartree.}
		\begin{ruledtabular}\begin{tabular}{lccccccc}
				\multirow{2}{*}{Step} & $\epsilon_{\text{in}}$ & Orbitals &\multirow{2}{*}{$R$} &\multirow{2}{*}{$M$} & $\alpha_{DF}$ & \multirow{2}{*}{Qubits} & $\#$Toffoli\\
				& / Hartree  & $N$ & &&/ Hartree & & gates\\
				\hline
				IX-cas16-fb-68e62o & $ 1.00 \times 10^{-4}$ & 62 & 886 & 45264 & 396.7 & 4232 & $ 4.67 \times 10^{10} $ \\ 
				IX-cas16-fb-68e62o & $ 1.58 \times 10^{-4}$ & 62 & 834 & 42011 & 396.7 & 4232 & $ 4.44 \times 10^{10} $ \\ 
				IX-cas16-fb-68e62o & $ 2.51 \times 10^{-4}$ & 62 & 787 & 38764 & 396.7 & 4232 & $ 4.22 \times 10^{10} $ \\ 
				IX-cas16-fb-68e62o & $ 3.98 \times 10^{-4}$ & 62 & 738 & 35493 & 396.7 & 4232 & $ 3.99 \times 10^{10} $ \\ 
				IX-cas16-fb-68e62o & $ 6.31 \times 10^{-4}$ & 62 & 686 & 32226 & 396.6 & 4231 & $ 3.76 \times 10^{10} $ \\ 
				IX-cas16-fb-68e62o & $ 1.00 \times 10^{-3}$ & 62 & 638 & 28945 & 396.6 & 4231 & $ 3.53 \times 10^{10} $ \\ 
				IX-cas16-fb-68e62o & $ 1.58 \times 10^{-3}$ & 62 & 583 & 25694 & 396.6 & 4231 & $ 3.30 \times 10^{10} $ \\ 
				IX-cas16-fb-68e62o & $ 2.51 \times 10^{-3}$ & 62 & 532 & 22504 & 396.5 & 4231 & $ 3.08 \times 10^{10} $ \\ 
				IX-cas16-fb-68e62o & $ 3.98 \times 10^{-3}$ & 62 & 479 & 19401 & 396.4 & 4231 & $ 2.86 \times 10^{10} $ \\ 
				IX-cas16-fb-68e62o & $ 6.31 \times 10^{-3}$ & 62 & 428 & 16430 & 396.3 & 4231 & $ 2.65 \times 10^{10} $ \\ 
				IX-cas16-fb-68e62o & $ 1.00 \times 10^{-2}$ & 62 & 379 & 13616 & 396.1 & 4230 & $ 2.45 \times 10^{10} $ \\ 
				IX-cas16-fb-68e62o & $ 1.58 \times 10^{-2}$ & 62 & 318 & 11056 & 395.6 & 4230 & $ 2.27 \times 10^{10} $ \\ 
				IX-cas16-fb-68e62o & $ 2.51 \times 10^{-2}$ & 62 & 270 & 8809 & 395.0 & 4230 & $ 2.10 \times 10^{10} $ \\ 
				IX-cas16-fb-68e62o & $ 3.98 \times 10^{-2}$ & 62 & 227 & 6887 & 394.2 & 4229 & $ 1.97 \times 10^{10} $ \\ 
				IX-cas16-fb-68e62o & $ 6.31 \times 10^{-2}$ & 62 & 178 & 5282 & 392.8 & 4229 & $ 1.84 \times 10^{10} $ \\ 
				IX-cas16-fb-68e62o & $ 1.00 \times 10^{-1}$ & 62 & 140 & 4020 & 390.8 & 4228 & $ 1.75 \times 10^{10} $ \\ \hline
				XVIII-cas4-fb-64e56o & $ 1.00 \times 10^{-4}$ & 56 & 969 & 45760 & 293.6 & 3712 & $ 3.35 \times 10^{10} $ \\ 
				XVIII-cas4-fb-64e56o & $ 1.58 \times 10^{-4}$ & 56 & 921 & 42665 & 293.5 & 3712 & $ 3.19 \times 10^{10} $ \\ 
				XVIII-cas4-fb-64e56o & $ 2.51 \times 10^{-4}$ & 56 & 871 & 39511 & 293.5 & 3712 & $ 3.02 \times 10^{10} $ \\ 
				XVIII-cas4-fb-64e56o & $ 3.98 \times 10^{-4}$ & 56 & 819 & 36271 & 293.5 & 3712 & $ 2.86 \times 10^{10} $ \\ 
				XVIII-cas4-fb-64e56o & $ 6.31 \times 10^{-4}$ & 56 & 762 & 32962 & 293.5 & 3712 & $ 2.69 \times 10^{10} $ \\ 
				XVIII-cas4-fb-64e56o & $ 1.00 \times 10^{-3}$ & 56 & 705 & 29594 & 293.5 & 3711 & $ 2.51 \times 10^{10} $ \\ 
				XVIII-cas4-fb-64e56o & $ 1.58 \times 10^{-3}$ & 56 & 653 & 26201 & 293.5 & 3711 & $ 2.33 \times 10^{10} $ \\ 
				XVIII-cas4-fb-64e56o & $ 2.51 \times 10^{-3}$ & 56 & 592 & 22796 & 293.4 & 3711 & $ 2.16 \times 10^{10} $ \\ 
				XVIII-cas4-fb-64e56o & $ 3.98 \times 10^{-3}$ & 56 & 530 & 19439 & 293.3 & 3711 & $ 1.98 \times 10^{10} $ \\ 
				XVIII-cas4-fb-64e56o & $ 6.31 \times 10^{-3}$ & 56 & 470 & 16199 & 293.1 & 3710 & $ 1.81 \times 10^{10} $ \\ 
				XVIII-cas4-fb-64e56o & $ 1.00 \times 10^{-2}$ & 56 & 400 & 13174 & 292.8 & 3710 & $ 1.65 \times 10^{10} $ \\ 
				XVIII-cas4-fb-64e56o & $ 1.58 \times 10^{-2}$ & 56 & 336 & 10436 & 292.4 & 3710 & $ 1.51 \times 10^{10} $ \\ 
				XVIII-cas4-fb-64e56o & $ 2.51 \times 10^{-2}$ & 56 & 275 & 8001 & 291.7 & 3709 & $ 1.38 \times 10^{10} $ \\ 
				XVIII-cas4-fb-64e56o & $ 3.98 \times 10^{-2}$ & 56 & 216 & 5884 & 290.6 & 3709 & $ 1.26 \times 10^{10} $ \\ 
				XVIII-cas4-fb-64e56o & $ 6.31 \times 10^{-2}$ & 56 & 161 & 4168 & 288.8 & 3709 & $ 1.17 \times 10^{10} $ \\ 
				XVIII-cas4-fb-64e56o & $ 1.00 \times 10^{-1}$ & 56 & 123 & 2979 & 286.8 & 3708 & $ 1.10 \times 10^{10} $ \\ 
		\end{tabular}\end{ruledtabular}
	\end{table}
	
	\begin{table}
		\caption{\label{tab:RE-1} Double-factorization resource estimates for all steps of carbon dioxide fixation for active space sizes from $2$--$250$ orbitals at a truncation threshold of $\epsilon_{\text{in}}=1$mHartree.
			Examples with $\leq20$ orbitals marked by (*) are considered classically tractable by FCI methods.}
		\begin{ruledtabular}\begin{tabular}{lccccccc}
				\multirow{2}{*}{Step}  & Orbitals & \multirow{2}{*}{$R$} &\multirow{2}{*}{$M$} & $\alpha_{DF}$ & \multirow{2}{*}{$\lambda$} & \multirow{2}{*}{Qubits} & $\#$Toffoli\\
				& $N$ & & &  $/ \text{Hartree}$&&& gates\\
				\hline
				I-cas5-fb-4e5o & 5* & 12 & 54 & 6.5 & 0 & 126 & $ 1.3 \times 10^{7}$ \\ 
				I-cas5-fb-14e16o & 16* & 104 & 1293 & 25.6 & 1 & 907 & $ 2.7 \times 10^{8}$ \\ 
				I-cas5-fb-48e52o & 52 & 613 & 23566 & 177.3 & 3 & 6879 & $ 1.1 \times 10^{10}$ \\ I-highCD-cas5-fb-4e5o & 5* & 12 & 54 & 6.5 & 0 & 126 & $ 1.3 \times 10^{7}$ \\ 
				I-highCD-cas5-fb-14e16o & 16* & 104 & 1296 & 25.6 & 1 & 907 & $ 2.7 \times 10^{8}$ \\ 
				I-highCD-cas5-fb-48e52o & 52 & 616 & 23785 & 177.4 & 3 & 6879 & $ 1.1 \times 10^{10}$ \\  \hline
				II-cas6-fb-8e6o & 6* & 15 & 87 & 13.7 & 0 & 163 & $ 3.6 \times 10^{7}$ \\ 
				II-cas6-fb-34e26o & 26 & 203 & 3822 & 117.4 & 1 & 1624 & $ 2.5 \times 10^{9}$ \\ 
				II-cas6-fb-70e62o & 62 & 734 & 33629 & 374.4 & 3 & 8448 & $ 3.1 \times 10^{10}$ \\ 
				II-hf-fb-38e33o & 33 & 361 & 8795 & 138.6 & 2 & 3182 & $ 4.6 \times 10^{9}$ \\ 
				II-hf-fb-74e71o & 71 & 935 & 50506 & 422.8 & 4 & 12086 & $ 4.4 \times 10^{10}$ \\ 
				II-highCD-cas6-fb-8e6o & 6* & 15 & 87 & 13.7 & 0 & 163 & $ 3.6 \times 10^{7}$ \\ 
				II-highCD-cas6-fb-34e26o & 26 & 203 & 3822 & 117.3 & 1 & 1624 & $ 2.5 \times 10^{9}$ \\ 
				II-highCD-cas6-fb-70e64o & 64 & 740 & 33861 & 374.5 & 3 & 8720 & $ 3.2 \times 10^{10}$ \\ \hline
				II-III-cas6-fb-8e6o & 6* & 19 & 99 & 12.4 & 0 & 163 & $ 3.3 \times 10^{7}$ \\ 
				II-III-cas6-fb-38e29o & 29 & 253 & 5497 & 144.4 & 2 & 2710 & $ 3.7 \times 10^{9}$ \\ 
				II-III-cas6-fb-74e65o & 65 & 783 & 38122 & 416.0 & 3 & 8856 & $ 3.7 \times 10^{10}$ \\ 
				II-III-highCD-cas6-fb-8e6o & 6* & 19 & 99 & 12.4 & 0 & 163 & $ 3.3 \times 10^{7}$ \\ 
				II-III-highCD-cas6-fb-38e29o & 29 & 253 & 5500 & 144.4 & 2 & 2710 & $ 3.8 \times 10^{9}$ \\ 
				II-III-highCD-cas6-fb-74e65o & 65 & 789 & 38364 & 416.0 & 3 & 8856 & $ 3.7 \times 10^{10}$ \\ \hline
				V-cas11-fb-12e11o & 11* & 49 & 453 & 31.8 & 0 & 317 & $ 1.9 \times 10^{8}$ \\ 
				V-cas11-fb-32e25o & 24 & 146 & 2738 & 95.5 & 1 & 1500 & $ 1.7 \times 10^{9}$ \\ 
				V-cas11-fb-68e60o & 60 & 670 & 29319 & 372.1 & 3 & 8175 & $ 2.9 \times 10^{10}$ \\ 
				V-highCD-cas11-fb-12e11o & 11* & 49 & 453 & 31.8 & 0 & 317 & $ 2.0 \times 10^{8}$ \\ 
				V-highCD-cas11-fb-32e24o & 24 & 146 & 2739 & 122.1 & 1 & 1500 & $ 2.2 \times 10^{9}$ \\ 
				V-highCD-cas11-fb-68e60o & 60 & 671 & 29503 & 372.7 & 3 & 8175 & $ 2.9 \times 10^{10}$ \\  \hline
				VIII-cas2-fb-2e2o & 2* & 3 & 6 & 1.3 & 0 & 45 & $ 8.2 \times 10^{5}$ \\ 
				VIII-cas2-fb-40e29o & 29 & 270 & 5870 & 146.2 & 2 & 2710 & $ 3.9 \times 10^{9}$ \\
				VIII-cas2-fb-76e65o & 65 & 794 & 39088 & 425.7 & 3 & 8856 & $ 3.8 \times 10^{10}$ \\  VIII-highCD-cas2-fb-2e2o & 2* & 3 & 6 & 1.3 & 0 & 45 & $ 8.2 \times 10^{5}$ \\ 
				VIII-highCD-cas2-fb-40e29o & 29 & 270 & 5878 & 146.2 & 2 & 2710 & $ 3.9 \times 10^{9}$ \\ 
				VIII-highCD-cas2-fb-76e65o & 65 & 805 & 39393 & 425.9 & 3 & 8856 & $ 3.8 \times 10^{10}$ \\ \hline
				VIII-IX-cas4-fb-4e4o & 4* & 9 & 32 & 3.9 & 0 & 97 & $ 5.8 \times 10^{6}$ \\ 
				VIII-IX-cas4-fb-36e23o & 23 & 177 & 3063 & 121.7 & 1 & 1438 & $ 2.2 \times 10^{9}$ \\ 
				VIII-IX-cas4-fb-72e59o & 59 & 666 & 29286 & 384.4 & 3 & 8039 & $ 2.9 \times 10^{10}$ \\  VIII-IX-highCD-cas4-fb-4e4o & 4* & 9 & 32 & 3.9 & 0 & 97 & $ 5.8 \times 10^{6}$ \\ 
				VIII-IX-highCD-cas4-fb-36e23o & 23 & 177 & 3063 & 121.7 & 1 & 1438 & $ 2.2 \times 10^{9}$ \\ 
				VIII-IX-highCD-cas4-fb-72e59o & 59 & 668 & 29417 & 384.5 & 3 & 8039 & $ 3.0 \times 10^{10}$ \\
		\end{tabular}\end{ruledtabular}
	\end{table}
	
	\begin{table}
		\caption{\label{tab:RE-2} Double-factorization resource estimates for all steps of carbon dioxide fixation for active space sizes from $2$--$250$ orbitals at a truncation threshold of $\epsilon_{\text{in}}=1$mHartree.
			Examples with $\leq20$ orbitals marked by (*) are considered classically tractable by FCI methods.}
		\begin{ruledtabular}\begin{tabular}{lccccccc}
				\multirow{2}{*}{Step}  & Orbitals & \multirow{2}{*}{$R$} &\multirow{2}{*}{$M$} & $\alpha_{DF}$ & \multirow{2}{*}{$\lambda$} & \multirow{2}{*}{Qubits} & $\#$Toffoli\\
				& $N$ & & &  $/ \text{Hartree}$&&& gates\\
				\hline
				IX-cas16-fb-16e16o & 16* & 80 & 1050 & 58.3 & 1 & 939 & $ 6.0 \times 10^{8}$ \\ 
				IX-cas16-fb-32e26o & 26 & 155 & 3215 & 136.4 & 1 & 1624 & $ 2.8 \times 10^{9}$ \\ 
				IX-cas16-fb-68e62o & 62 & 638 & 28945 & 396.6 & 3 & 8447 & $ 3.1 \times 10^{10}$ \\ IX-highCD-cas16-fb-16e16o & 16* & 80 & 1050 & 58.3 & 1 & 939 & $ 6.0 \times 10^{8}$ \\ 
				IX-highCD-cas16-fb-32e26o & 26 & 155 & 3215 & 136.4 & 2 & 2430 & $ 2.9 \times 10^{9}$ \\ 
				IX-highCD-cas16-fb-68e62o & 62 & 641 & 29096 & 396.6 & 3 & 8447 & $ 3.1 \times 10^{10}$ \\ \hline
				XVIII-cas4-fb-4e4o & 4* & 9 & 36 & 5.2 & 0 & 102 & $ 8.0 \times 10^{6}$ \\ 
				XVIII-cas4-fb-28e20o & 20* & 155 & 2336 & 71.0 & 1 & 1212 & $ 1.1 \times 10^{9}$ \\ 
				XVIII-cas4-fb-64e56o & 56 & 705 & 29594 & 293.5 & 3 & 7407 & $ 2.1 \times 10^{10}$ \\ 
				XVIII-cas4-fb-100e100o & 100 & 982 & 75449 & 781.0 & 4 & 18017 & $ 1.2 \times 10^{11}$ \\ 
				XVIII-cas4-fb-150e150o & 150 & 1462 & 174699 & 1919.7 & 5 & 34218 & $ 5.5 \times 10^{11}$ \\ 
				XVIII-cas4-fb-250o250e & 250 & 2276 & 443046 & 6793.2 & 6 & 70019 & $ 3.9 \times 10^{12}$ \\ 
				XVIII-highCD-cas4-fb-4e4o & 4* & 9 & 36 & 5.2 & 0 & 102 & $ 8.0 \times 10^{6}$ \\ 
				XVIII-highCD-cas4-fb-28e20o & 20* & 155 & 2337 & 71.0 & 1 & 1212 & $ 1.1 \times 10^{9}$ \\ 
				XVIII-highCD-cas4-fb-64e56o & 56 & 712 & 29763 & 293.5 & 3 & 7407 & $ 2.1 \times 10^{10}$ \\ 
				XVIII-highCD-cas4-fb-100e100o & 100 & 995 & 76323 & 782.1 & 4 & 18017 & $ 1.2 \times 10^{11}$ \\ 
				XVIII-highCD-cas4-fb-150e150o & 150 & 1490 & 177979 & 1923.8 & 5 & 34218 & $ 5.5 \times 10^{11}$ \\ 
				XVIII-highCD-cas4-fb-250e250o & 250 & 2362 & 456067 & 6806.8 & 6 & 70019
				& $4.0 \times 10^{12}$
				\\
		\end{tabular}\end{ruledtabular}
	\end{table}

	\clearpage
	\section{Qubitization of double-factorized Hamiltonian}
	In this section, we present our approach for qubitizing the double-factorized Hamiltonian.
	We make heavy use of quantum circuit diagrams, and many our proofs follow from combining these diagrams.
	Thus in~\cref{sec:preliminaries}, we define our notation for many standard quantum circuit primitives and their costs.
	Of particular importance are quantum circuits for data-lookup in~\cref{sec:datalookup}, state preparation in~\cref{sec:state_preparation},  block-encoding in~\cref{sec:block_encoding}, and qubitization  in~\cref{sec:qubitization}.
	In~\cref{sec:new_circuit_primitives}, we also present new optimized constructions of two quantum circuit primitives.
	These are the programmable rotation gate array in~\cref{sec:programmableRotation}, which applies a single-qubit rotation with a rotation angle controlled by an integer index, and sparse multiplexed data-lookup in~\cref{sec:MultiplexSparseDataLookup}.
	These primitives are then combined in our qubitization of the double-factorized Hamiltonin in~\cref{sec:qubitization_electronic_structure}.
	
	\subsection{Standard quantum circuit primitives}
	\label{sec:preliminaries}
	
	In this section, we define quantum circuit primitives and their costs.
	Throughout, we use the shorthand $\qubits{x}\doteq\lceil \log_2{x}\rceil$ which is the number of bits needed to store an integer of size $x$.
	All the quantum circuits we use are constructed from the following primitive elements.
	
	\begin{align}
	\text{Unitary and its adjoint:}\quad&\image{Notationunitary},\quad\image{NotationunitaryAdjoint}.
	\end{align}
	\begin{align}
	\text{Controlled unitary:}\quad&{\image{NotationControlled}=\ket{0}\bra{0}\otimes \ii + \ket{1}\bra{1}\otimes U.}
	\end{align}
	\begin{align}
	\label{notation:Multiplexed}
	\text{Multiplexed unitaries:}\quad&{\image{NotationMultiplex}=\sum_j\ket{j}\bra{j}\otimes U_j,\quad\image{NotationMultiplex2}=\sum_{j,k}\ket{j}\bra{j}\otimes\ket{k}\bra{k}\otimes U_{j,k}.}
	\end{align}
	\begin{align}
	\text{Controlled multiplexed unitary:}\quad&{\image{NotationControlledMultiplex}=\ket{0}\bra{0}\otimes\ii+\ket{1}\bra{1}\otimes\sum_{j}\ket{j}\bra{j}\otimes U_j.}
	\end{align}
	\begin{align}
	\label{notation:DataLookup}
	\text{Data-lookup that XORs $\vec{x}_k$ with $\vec{z}$:}\quad&{\image{Notationdatalookup}.}
	\end{align}
	\begin{align}
	\label{notation:StatePreparation}
	\text{Unitary that prepares $\ket{\vec{a}}$:}\quad&{\image{statepreparationBasic}=\sum_{j}\sqrt{\frac{|a_j|}{\|\vec{a}\|_1}}\ket{j}.}
	\end{align}
	\begin{align}
	\label{notation:StatePreparationBar}
	\text{Unitary that prepares $\overline{\ket{\vec{a}}}$:}\quad&{\image{statepreparationBar}=\sum_{j}\mathrm{sign}[a_j]\sqrt{\frac{|a_j|}{\|\vec{a}\|_1}}\ket{j}.}
	\end{align}
	\begin{align}
	\label{notation:StatePreparationFunction}
	\text{Unitary that prepares $\ket{\vec{a},f}$:}\quad&{\image{statepreparationFunction}=\sum_{j}\sqrt{\frac{|a_j|}{\|\vec{a}\|_1}}\ket{f(j)}_1\ket{j}_2.}
	\end{align}
	\begin{align}
	\label{notation:BlockEncoding}
	\text{Block-encoding unitary:}\quad&{\image{Notationblockencoding}.}
	\end{align}
	\begin{align}
	\label{notation:BlockEncodingTrivial}
	\text{The block-encoding of a unitary is trivial:}\quad&{\image{NotationblockencodingTrivial}.}
	\end{align}
	\begin{align}
	\label{notation:BlockEncodingMultiplex}
	\text{Multiplexed block-encoding unitary:}\quad&{\image{NotationblockencodingMultiplex}.}
	\end{align}
	\begin{align}
	\label{notation:BlockEncodingAdd}
	\text{Adding block-encodings:}\quad&{\image{NotationblockencodingAdd}.}
	\end{align}
	\begin{align}
	\label{notation:BlockEncodingMultiply}
	\text{Multiplying block-encodings:}\\
	& \hspace*{-2cm}{\image{NotationblockencodingMultiply}.}
	\end{align}
	
	The costs of these quantum circuits are expressed in terms of primitive quantum gates and qubits.
	We define primitive quantum gates as those acting on at most two qubits, and we distinguish between Clifford gates $\{\textsc{Had},S,\textsc{CNOT}\}$, and the non-Clifford $\t$ gate. Note that we denote the Hadamard gate as $\textsc{Had}$ instead of the more traditional symbol $H$ to avoid confusion with the Hamiltonian.
	Some of these circuits may be implemented by ancilla qubits in addition to the `register' qubits. Typically, we do not show the ancilla qubits in the diagrams and implicitly assume that these invisible qubits are all borrowed, meaning that they are returned to the same initial state at the end of the circuit.
	We call these ancilla qubits `clean' if their initial state is the computational basis state $\ket{0}$.
	In contrast, We call these ancilla qubits `dirty' if their initial is arbitrary and unknown.
	
	In some cases, a quantum circuit $U$ is approximated by $U'$ to some error $\|U-U'\|\le\epsilon$ in spectral norm.
	The errors of multiple approximate quantum circuits add linearly, following the triangle inequality
	\begin{align}
	\|U_0U_1\cdots U_{N-1}-U'_0U'_1\cdots U'_{N-1}\|\le \sum_{j\in[N]}\|U_j-U_j'\|.
	\end{align}

	\subsubsection{Data-lookup oracle}
	\label{sec:datalookup}
	Given a list of $d$ bit-strings $\vec{a}\in\{0,1\}^{d\times b}$, each of length $b$, the data-lookup oracle in~\cref{notation:DataLookup} returns the bit-string $a_x$, that is
	$D\ket{x}\ket{z}=\ket{x}\ket{z\oplus a_x}$.
	In addition to the $b+\lceil\log_{2}(d)\rceil$ qubits needed to store its inputs and outputs, this oracle has the following cost.
	\begin{lemma}[Data-lookup oracle~\cite{Babbush2018encoding}]
		\label{lem:DataLookup}
		The data-lookup oracle in~\cref{notation:DataLookup} and its controlled version can be implemented using
		\begin{itemize}
			\item Toffoli gates: $(d-1)$.
			\item Clifford gates: $\Theta(db)$.
			\item Clean ancillary qubits: $\lceil\log_{2}(d)\rceil$ .
		\end{itemize}
	\end{lemma}

	According to Low et al.~\cite{Low2018Trading}, the Toffoli gate cost of data-lookup can be further reduced by adding $\lambda b$ additional qubits, where $\lambda\ge 0$ is an integer.
	These qubits can be clean, meaning they start in and are returned to the computational basis state $\ket{0}$.
	Circuit optimizations by Berry et al.~\cite{Berry2019CholeskyQubitization} further reduce the cost by constant factors.
	
	\begin{lemma}[Data-lookup oracle with clean qubit assistance~\cite{Low2018Trading,Berry2019CholeskyQubitization}]
		\label{lem:DataLookupAncilla}
		For any integer $\lambda\ge 0$, the data-lookup oracle in~\cref{notation:DataLookup} and its controlled version can be implemented using
		\begin{itemize}
			\item Toffoli gates: $d/( 1+\lambda)+\lambda b+\mathcal{O}(\log{d})$.
			\item Clifford gates: $\Theta(db)$.
			\item Clean ancillary qubits: $\lambda b+\lceil\log_{2}(d/\lambda)\rceil+\mathcal{O}(1)$.
		\end{itemize}
	\end{lemma}
	Note that the additional $\mathcal{O}(1)$ clean qubits and $\mathcal{O}(\log{(d)})$ Toffoli gates only needed when $1+\lambda$ is not a power of two.
	In this case, they are used in an intermediate step to reversibly compute the remainder and quotient of $j/(1+\lambda)$, where the numerator $j\in[N]$ is stored in a $\qubits{d}$ qubit register.
	In the following, we will omit this additive cost as it is subdominant in all our applications.
	We find it useful to define the Toffoli gate count function
	\begin{align}
	\data{d}{b}{\lambda}=\min_{\lambda'\in[0,\lambda]}\left(d/( 1+\lambda')+\lambda' b\right)\gtrsim \min{\left[d,2\sqrt{bd}\right]},
	\end{align}
	which returns the smallest possible Toffoli gate count for any number of $\lambda$ available clean qubits.
	The bit-strings output by these lookup oracles may be uncomputed by applying their adjoint.
	This doubles their gate complexity at most.
	However, the gate complexity of uncomputation can be further improved~\cite{Berry2019CholeskyQubitization} used measurement-based uncomputation.
	By using $\lambda + \mathcal{O}(\log(d/\lambda))$ qubits, measurement-based uncomputation of data-lookup reduces the additive $\lambda b$ Toffoli gate term to $\lambda$, which becomes significant when  $\lambda b\sim\sqrt{bd}$.
	More importantly, any garbage qubits produced by~\cref{lem:DataLookupAncilla} can measured and the results stored in classical memory for later uncomputation. 
	This frees up clean ancilla qubits and means that the $\lambda$ parameter for uncomputation can be optimized separately from that of computation in reducing Toffoli costs.
	Thus we also find it useful to define the cost of uncomputation as
	\begin{align}
	\label{eq:Data_uncomputation}
	\dataU{d}{b}{\lambda}=\min_{\lambda'\in[0,\lambda]}\left(d/( 1+\lambda')+\lambda'\right)\gtrsim \min{\left[d,2\sqrt{d}\right]}.
	\end{align}
	In most of the cases we consider, uncomputation of data-lookup is an order of magnitude cheaper than computation. 

	These assisting qubits can also be dirty, meaning that start in and are returned to the same initial state.
	This is useful whenever the quantum algorithm has any idling qubits.
	To simply notation in the following, we will assume that $\lambda$ is a power of two.
	\begin{lemma}[Data-lookup oracle with dirty qubit assistance~\cite{Low2018Trading,Berry2019CholeskyQubitization}]
		\label{lem:DataLookupDirtyAncilla}
		For any integer $\lambda\ge b$ such that $1+\lambda$ is a power of two, the data-lookup oracle in~\cref{notation:DataLookup} and its controlled version can be implemented using
		\begin{itemize}
			\item Toffoli gates: $2d/(1+\lambda)+4\lambda b$.
			\item Clifford gates: $\Theta(db)$.
			\item Clean ancillary qubits: $\lceil\log_{2}(d/\lambda)\rceil$.
			\item Dirty ancillary qubits: $\lambda b$.
		\end{itemize}
	\end{lemma}
	Note that when $\lambda \le 1$, there is no advantage over the original construction in~\cref{lem:DataLookup}.
	We find it useful to define the Toffoli gate count function
	\begin{align}
	\datadirty{d}{b}{b\lambda}&=\min_{\lambda'\in[1,\lambda]}\left(d,2d/(1+\lambda')+4\lambda'b\right)\gtrsim \min{\left[d,4\sqrt{2bd}\right]},\\\nonumber
	\dataUdirty{d}{b}{b\lambda}&=\min_{\lambda'\in[1,\lambda]}\left(d,2d/(1+\lambda')+4\lambda'\right)\gtrsim \min{\left[d,4\sqrt{2d}\right]},
	\end{align}
	which returns the smallest possible Toffoli gate count for any number of $\lambda$ available dirty qubits.
	Note a slight difference in notation compared to~\cref{eq:Data_uncomputation} -- there, the last subscript $\lambda$ is the number of times the output register is duplicated, whereas the last subscript $\lambda b$ parameterizes the total number of dirty qubits available.
	Similar to the case using clean qubits, the Toffoli cost of uncomputation is less than that of computation~\cite{Berry2019CholeskyQubitization}.

	\subsubsection{State preparation unitary}
	\label{sec:state_preparation}
	Quantum state preparation is a unitary circuit that prepares a desired quantum state.
	
	A number of different quantum circuit implementations of~\cref{notation:StatePreparation,notation:StatePreparationBar} are known, each with different trade-offs in qubit count, and the various quantum gates.
	For instance, the approach by Shende et al.~\cite{shende2006synthesis} uses $\lceil\log_2{d}\rceil$ qubits and $d$ arbitrary single-qubit rotations, and $\mathcal{O}(d)$ other two-qubit Clifford gates.

	For our purposes, it suffices to prepare quantum states where each $\ket{j}$ is, in general, entangled with some arbitrary quantum state $\ket{\garb_j}$.
	\begin{align}
	\label{eq:StatePreparationGarbage}
	\ket{\vec{a}}=\sum_{j=0}^{d-1}\sqrt{\frac{|a_j|}{\|\vec{a}\|_1}}\ket{j}\ket{\garb_j}
	\end{align}
	As this includes~\cref{notation:StatePreparation} as a special case, quantum circuits for state preparation with garbage can use fewer $\t$ gates, though at the expense of more qubits.
	This garbage state may be safely ignored in the remainder of this manuscript, so we do not differentiate between the state preparation unitaries of~\cref{notation:StatePreparation} and~\cref{eq:StatePreparationGarbage}.
	Moreover, the circuits for $\ket{\vec{a}}$ and $\overline{\ket{\vec{a}}}$ are very similar and have the same $\t$ gate cost.
	We will repeatedly invoke the following implementation which approximates each coefficient to a targeted precision.
	\begin{lemma}[Approximate quantum state preparation with garbage~\cite{Babbush2018encoding}]
		\label{lem:StatePreparationGarbage}
		Given a list of $d$ positive coefficients $\vec{a}\in\mathbb{R}^d$ and the desired bits of precision $\mu$, the quantum state
		\begin{align}
		\ket{\psi}=\sum_{j=0}^{d-1}\sqrt{p_j}\ket{j}\ket{\garb_{j}},
		\end{align}
		where $\left|p_j-\frac{a_j}{\|\vec{a}\|_1}\right|\le \frac{2^{-\mu}}{d}$ (which implies that $\left\|\vec{p}-\frac{\vec{a}}{\|\vec{a}\|_1}\right\|_1\le 2^{-\mu}$)
		can be prepared by a unitary U that is implemented using one application of any data-lookup oracle from~\cref{sec:datalookup} for $d$ bit-strings of length $\lceil\log_{2}(d)\rceil+\mu$. The total cost for implementing $U$ is given by:
		\begin{itemize}
			\item Toffoli gates: $\mu+\data{d}{\lceil\log_{2}(d)\rceil+\mu}{\lambda}+\Theta(\log{(d)})$.
			\item Arbitrary single-qubit rotations: $1$.
			\item Clifford gates: $\Theta(d\mu)$.
			\item Garbage qubits: $2\mu+\qubits{d}$.
			\item Clean qubits: $\qubits{d}+\mathcal{O}(1)$.
			\item Dirty qubits: $\lambda$.
		\end{itemize}
	\end{lemma}
	Our implementation of~\cref{lem:StatePreparationGarbage} introduces a small modification that is useful for synthesizing the multiplexed version of state preparation. 
	In the original procedure~\cite{Babbush2018encoding}, let $L=\lceil\log_{2}(d)\rceil$ be the number of bits needed to store any integer between $0$ and $d-1$, and let $\mu$ be the bits of precision to which $p_j/d$ is specified.
	Then one prepares a uniform superposition over $d\times 2^\mu$ elements $\frac{1}{\sqrt{d}}\sum^{d-1}_{j=0}\ket{j}\ket{+}^{\otimes\mu}$ and applies data-lookup that writes out two bitstrings $k=f(j)\in[d]$ and $\tilde{p}_j\in[2^\mu]$. 
	Note that the map from $j$ to $f(j)$ can be surjective so $f^{-1}({k})=\{j : f(j)=k\}$.  
	This produces the state 
	\begin{align}
	\label{eq:stateprepstep1}
	\frac{1}{\sqrt{d}}\sum^{d-1}_{j=0}\ket{j}\ket{f(j)}\ket{+}^{\otimes\mu}\ket{\tilde{p}_j}.
	\end{align}
	Let $\textsc{Comp}$ be a unitary that compares two integers $x,y$ and writes the result in a qubit.
	\begin{align}
	\textsc{Comp}\ket{x}\ket{y}\ket{0}=\ket{x}\ket{y}\ket{x\ge y}.
	\end{align}
	When applied to compare $\ket{+}^{\otimes\mu}$ and $\ket{\tilde{p}_j}$, this results in
	\begin{align}
	\label{eq:stateprepstep2}
	%
	%
	\frac{1}{\sqrt{d}}\sum^{d-1}_{j=0}\ket{j}\ket{f(j)}\left(
	\sqrt{\frac{\tilde{p}_j}{2^{\mu}}}\ket{\phi_j}\ket{0}+
	\sqrt{\frac{2^{\mu}-\tilde{p}_j}{2^{\mu}}}\ket{\chi_j}\ket{1}
	\right),
	\end{align}
	where the $\ket{\phi_j}, \ket{\chi_j}$ are some normalized quantum states. 
	Now, swap the registers $\ket{j}$ and $\ket{f(j)}$ controlled on the comparator output qubit. By collecting the $\ket{j}$ index, 
	\begin{align}
	\frac{1}{\sqrt{d}}\sum^{d-1}_{j=0}\ket{j}
	\left[
	\sqrt{\frac{\tilde{p}_j}{2^{\mu}}}\ket{f(j)}\ket{\phi_j}\ket{0}
	+\sum_{k\in f^{-1}(j)}\sqrt{\frac{2^{\mu}-\tilde{p}_k}{2^{\mu}}}\ket{k}\ket{\chi_k}\ket{1}\right]
	=\sum^{d-1}_{j=0}\sqrt{p_j}\ket{j}\ket{\garb}_j.
	\end{align}
	Thus any desired distribution of probabilities $p_j=\frac{1}{d}\left(\frac{\tilde{p}_j}{2^{\mu}}+\sum_{k\in f^{-1}(j)}\frac{2^{\mu}-\tilde{p}_k}{2^{\mu}}\right)$ may be specified up to an error of $\frac{1}{d2^{\mu}}$ by an appropriate choice of $f(j)$ and $\tilde{p}_j$.
	
	We modify~\cref{lem:StatePreparationGarbage} by always choosing $\tilde{p}_j=0$ for any $j>d$, and then increasing $d$ to $2^L$.
	One disadvantage is that the number of Toffoli gates required rises from $d$ to $2^L$, but this is at most only a factor of two larger.
	Moreover, the advantage of this is in simplifying preparation of a uniform superposition $\frac{1}{\sqrt{d}}\sum^{d-1}_{j=0}\ket{j}$. 
	When $d$ is not a power of two, uniform state preparation is complicated by the need for arbitrary single-qubit rotations and amplitude amplification. 
	In contrast, preparing a uniform superposition over $2^L$ elements is accomplished trivially by $L$ Hadamard gates.
	This modification is especially useful when applying multiplexed state preparation, which is a unitary that prepares the state $\ket{\vec{a}_j}$ over $d_j$ elements controlled on an index $\ket{j}$.
	After choosing the $L$ that stores the largest integer $d_j$ and a $\mu$ that ensures the error of all $\ket{\vec{a}_j}$ are suitably bounded, the only circuit element that changes between different target states is data-lookup.
	This one simply replaces data-lookup with its multiplexed version, such as described in~\cref{sec:MultiplexSparseDataLookup}.
	
	It can also be useful to modify~\cref{notation:StatePreparation} to outputs some additional bits specified by an arbitrary Boolean function $g:[d]\rightarrow\{0,1\}^{b}$ is also useful.
	On one hand, the unitary of~\cref{notation:StatePreparationFunction} can be implemented by combining state preparation in~\cref{notation:StatePreparation} with data-lookup~\cref{notation:Multiplexed} as follows.
	\begin{align}
	\label{eq:statepreparationFunctionImplementation}
	\image{statepreparationFunctionImplementation} =\sum_{j\in[d]}\sqrt{\frac{|a_j|}{\|\vec{a}\|_1}}\ket{g(j)}_1\ket{j}_2.
	\end{align}
	Thus the $\t$ gate cost of~\cref{eq:statepreparationFunctionImplementation} is at most that of state preparation plus data-lookup on $d$ elements.
	A more efficient approach~\cite{Babbush2018encoding} following~\cref{eq:stateprepstep1}, is to have data-lookup write out these additional bits $\ket{g(j)}\ket{g(f(j))}$ in addition to $f(j)$ and $\tilde{p}_j$. Then in~\cref{eq:stateprepstep2}, we also apply a controlled swap to this new pair of registers.
	Thus additional bits may be output using only $\mathcal{O}(db)$ additional Clifford gates and $\mathcal{O}(b)$ additional Toffoli gates.
	When many coefficients of $\vec{a}$ are zero, it is useful for the Toffoli count to scale with the number of non-zero elements $\nnz{\vec{a}}$ rather than with $d$. 
	This is accomplished by having $\ket{j}$ index the $j^{\text{th}}$ non-zero element of $\vec{a}$, which is $a_{g(j)}$. Data-lookup then writes out $\ket{g(j)}\ket{g(f(j))}$. In~\cref{eq:stateprepstep2}, we swap these two registers instead of the $\ket{j}$ register.
	
	\begin{lemma}[Approximate sparse quantum state preparation with garbage~\cite{Berry2019CholeskyQubitization}]
		\label{lem:MultiplexSparseStatePreparationGarbage}
		Given a list of $d$ positive coefficients $\vec{a}\in\mathbb{R}^d$, the desired bits of precision $\mu$, and a Boolean function $f:[d]\rightarrow\{0,1\}^b$, the quantum state
		\begin{align}
		\ket{\psi}=\sum_{j=0}^{d-1}\sqrt{p_j}\ket{f(j)}\ket{\garb_{j}},
		\end{align}
		where $\quad\left|p_j-\frac{a_j}{\|\vec{a}\|_1}\right|\le \frac{2^{-\mu}}{d}$ (which implies that $\left\|\vec{p}-\frac{\vec{a}}{\|\vec{a}\|_1}\right\|_1\le 2^{-\mu}$)
		can be prepared by a unitary $U$, and $U$ is approximated to error $\epsilon$ using one application of any data-lookup oracle from~\cref{sec:datalookup} for $d$ bit-strings of length $2b+\mu$. The total cost for implementing $U$ is given by:
		\begin{itemize}
			\item Toffoli gates: $\mu+\data{d}{\lceil\log_{2}(d)\rceil+\mu}{\lambda}+\Theta(\log{(d/\epsilon)})$.
			\item Clifford gates: $\Theta(d(b+\mu)+\log{(1/\epsilon)})$.
			\item Garbage qubits: $2\mu+2\qubits{d}+b$.
			\item Clean qubits: $\qubits{d}+\mathcal{O}(1)$.
			\item Dirty qubits: $\lambda$.
		\end{itemize}
	\end{lemma}

	\subsubsection{Block-encoding framework}
	\label{sec:block_encoding}
	These two components of state preparation and select allow us to implement a block-encoding.
	\begin{definition}[Block-encoding implementation without sign qubit]
		\label{def:blockEncoding}
		Given the unitaries $\Ustate_{\vec{a}}$, and $\Uselect_{\vec{U}}$, let $H=\sum_{j=0}^{d-1}{a_j}U_j$.
		Then the block-encoding $\block{H/\|\vec{a}\|_1}$, where $\bra{0}_{a}\block{H/\|\vec{a}\|_1}\ket{0}_{a}=H/\|\vec{a}\|_1$, is implemented by
		\begin{align}
		\image{blockencodingNoSign},
		\end{align}
	\end{definition}
	Note that the same Hamiltonian may be block-encoded by many quantum circuits.
	For instance, the quantum circuit may explicitly implement the coefficient sign as follows,
	\begin{align}
	\label{eq:BlockEncodingWithSign}
	\image{blockencoding},
	%
	%
	\end{align}
	where $Z$ is the Pauli $Z\ket{x}=(-1)^x\ket{x}$, and $\text{sign}[\vec{a}](j)=\frac{1-\text{sign}(a_j)}{2}\in\{0,1\}$.
	This implementation has the advantage that its controlled version only needs to apply a control to the select unitary.
	In particular, the $\t$ gate cost of controlled-$\block{H/\|\vec{a}\|_1}$ is identical to $\block{H/\|\vec{a}\|_1}$ when the select unitary is implemented following the approach by Babbush et al.~\cite{Babbush2018encoding}.
	Moreover, note that the state preparation unitary is always followed up by by its adjoint.
	Thus the same Hamiltonian is block-encoded even if state preparation in~\cref{eq:StatePreparationGarbage} entangled with an garbage state~\cite{Babbush2018encoding}.

	Errors in state preparation or unitary synthesis introduce errors into the block-encoded Hamiltonian.
	We find it useful to define the approximate block-encoding.
	\begin{definition}[Approximate block-encoding]
		\label{def:approximateBlockEncoding}
		We say that $\blockApprox{H/\alpha}{\epsilon}\doteq \block{H'/\alpha}$ is an $\epsilon$-approximate block-encoding of $\block{H/\alpha}$ if
		\begin{align}
		\|H'/\alpha-H/\alpha\|\le \epsilon.
		\end{align}
	\end{definition}
	For instance, we have the following approximation due to error in the coefficient of the quantum state.
	\begin{lemma}[Approximate block-encoding using approximate state preparation]
		\label{lem:approximateBlockEncoding}
		Using the approximate state preparation circuit of~\cref{lem:StatePreparationGarbage} with precision parameter $\mu$ in the block-encoding circuit of~\cref{def:blockEncoding} produces an $\epsilon$-approximate block-encoding $\blockApprox{H/\|\vec{a}\|_1}{\epsilon}$, where $\epsilon= 2^{-\mu}$.
	\end{lemma}
	\begin{proof}
		Let $\vec{p}$ be such that  $\|\vec{p}-\vec{a}/\|\vec{a}\|_1\|_1\le 2^{-\mu}$.
		Let $H'=\sum_{j}p_j U_j$, and $H=\sum_{j}{a_j}U_j$.
		Then
		\begin{align}
		\left\|\frac{H}{\|\vec{a}\|_1}-H'\right\|&
		=\left\|\sum_{j}\left(\frac{a_j}{\|\vec{a}\|_1}-p_j\right)U_j\right\|
		\le 2^{-\mu}.
		\end{align}
	\end{proof}

	\subsubsection{Qubitization}
	\label{sec:qubitization}
	We also use the following result on qubitization, which is a generalization of quantum walks.
	\begin{theorem}[Qubitization]
		\label{def:qubitization}
		Let $\block{H}$ be a block-encoding of a Hamiltonian $H$ with spectral norm $\|H\|\le 1$, and an ancillary register with $M$ qubits.
		Then there is a quantum circuit  $\block{T_j[H]}$ that block-encodes $T_j[H]$, where $T_j[x]=\cos{(j\cos^{-1}(x))}$ is a Chebyshev polynomial of the first kind.
		\begin{align}
		\block{H}=\left(\begin{array}{cc}
		H & \cdots\\
		\vdots & \ddots
		\end{array}\right)\Rightarrow
		\block{T_j[H]}=\left(\begin{array}{cc}
		T_j[H] & \cdots\\
		\vdots & \ddots
		\end{array}\right).
		\end{align}
		In particular, following the work of Low and Chuang~\cite{Low2016Qubitization}, this circuit
		\begin{align}
		\label{eq:qubitization}
		\image{qubitization},
		\end{align}
		costs $j$ total queries to $\block{H}$ and its inverse, and $j-1$ reflections $\textsc{REF}=2\ket{0}\bra{0}_a-\ii_a$ on the ancilla register.
	\end{theorem}
	Each reflection can be understood as a multi-controlled $Z$ gate, and so has a Toffoli cost equal to the number of qubits it acts on.
	
	\subsection{New quantum circuit primitives}
	\label{sec:new_circuit_primitives}
	
	\subsubsection{Programmable rotation gate array}
	\label{sec:programmableRotation}
	In this section, we present an implementation of the multiplexed single-qubit $Z$-rotation gate.
	Given a list of $N$ angles $\vec{\theta}$ where each $\theta_k = \sum_{b=0}^{\beta-1}\theta_{k,b}/2^{1+b}\in[0,1-2^{-\beta}]$ is specified to $\beta$ bits of precision, we synthesize the unitary
	\begin{align}
	\label{eq:multiplexQubitRotation}
	&\image{multiplexQubitRotation}=\sum_{k\in[N]}\ket{k}\bra{k}\otimes e^{i2\pi\theta_k Z}=\image{multiplexQubitRotationSplit}.
	\end{align}
	For brevity, let $R_{b}=e^{i2\pi Z/2^{1+b}}$, and $R^{\theta}_{-1}=e^{i2\pi \theta Z}$.
	
	In our approach, we define a data register with $\kappa$ qubits that will store $\kappa$ bits of $\theta_{k}$.
	Let $\vec{\theta}_{k,[\mu:\mu+\kappa-1]}=(\theta_{k,\mu},\theta_{k,\mu+1},\cdots,\theta_{k,\mu+\kappa-1})$.
	Now define the data-lookup oracle that outputs $\kappa$ contiguous bits of $\theta_k$, conditioned on index $k$.
	\begin{align}
	\label{eq:dataLookupAngles}
	&\image{DataLookupAngles}.
	\end{align}
	Then~\cref{eq:multiplexQubitRotation} is implemented by the following circuit.
	\begin{align}
	\label{eq:multiplexQubitRotationSplitImplementation}
	&\image{multiplexQubitRotationSplitImplementation}.
	\end{align}
	With only $\kappa$ qubits, clearly $\lceil b/\kappa\rceil$ slices of the circuit within the dotted regions are required.
	Note that the middle pair of data-lookup oracles in the $j^{\text{th}}$ slice can be merged also into one that writes the bits $\vec{\theta}_{k,[j\kappa:(j+1)\kappa-1]}\oplus \vec{\theta}_{(k),[(j+1)\kappa:(j+2)\kappa-1]}$.
	Accounting for this merging, this circuit applies $\lceil b/\kappa\rceil+1$ data-lookup oracles each storing at most $\kappa$ entries.
	
	Another useful situation is where arbitrary unitaries are applied on the system register are interspersed between $M$ multiplexed rotations.
	\begin{align}
	\label{eq:multiplexQubitRotationMany}
	&\image{multiplexQubitRotationMany}.
	\end{align}
	We may use the same construction as in~\cref{eq:multiplexQubitRotationSplitImplementation} to implement this.
	The number of data-lookup oracles required is then $M\lceil b/\kappa\rceil+1$.
	We may reduce this to just $\lceil Mb/\kappa\rceil+1$.
	When $b$ is not an integer multiplier of $\kappa$, the data-lookup in the last slice might store fewer than $\kappa$ entries.
	Thus we fill these empty entries with bit-strings from the nest data-lookup.
	This filling procedure is illustrated by the following example, where the bits of precision $b=2\le \kappa=3$.
	\begin{align}
	\label{eq:multiplexQubitRotationMerge}
	&\image{multiplexQubitRotationMerge}.
	\end{align}
	In the case where many data qubits are available $\kappa\gg b$, we may similarly merge multiple bit-strings into the same lookup oracle.
	Thus, the Toffoli cost of~\cref{eq:multiplexQubitRotationMany} is equal to $\lceil Mb/\kappa+1\rceil$ data lookup oracles with $K$ bit-strings of length $\kappa$ in addition to that of all the $U_j$.
	Moreover, the number of qubits required for the data and index $k$ is equals to $\kappa+\lceil\log_{2}(M)\rceil$.
	It is valuable to express these costs with respect to a tunable number of qubits $\kappa$.
	According to~\cref{sec:datalookup}, the Toffoli gate cost of data-lookup with $K$ elements that outputs $\kappa$ bits an be reduced by using $\lambda$ ancillary qubits.
	When these qubits are clean, the Toffoli cost is $K/\lfloor1+ \frac{\lambda}{\kappa}\rfloor+\lambda$.
	Thus the Toffoli count of all the data-lookup oracles is
	\begin{align}
	\label{eq:ProgrammableArrayCost}
	\left\lceil\frac{Mb}{\kappa}+1\right\rceil\left(K/\left\lfloor 1+ \frac{\lambda}{\kappa}\right\rfloor+\lambda\right),
	\end{align}
	which is minimized by choosing $\lambda\sim\sqrt{K\kappa}$.

	We now bound the bits of precision required to approximate~\cref{eq:multiplexQubitRotationMany} to an overall error of $\epsilon$ in the spectral norm.
	Suppose we are given angles $\vec{\theta}_{\mathrm{exact}}$ that are real numbers.
	Then the error of approximating each angle $\theta_{\mathrm{exact},k}$ with a $\beta$-bit number ${\theta}_{k}$ is at most $2^{-\beta-1}$
	Thus the error of each rotation compared to its binary approximation is $\|R^{\theta_{\mathrm{exact},k}}_{-1}-R_{-1}^{\theta_k}\|\le\|R^{2^{-\beta-1}}-\ii\|=\|\cos{(2^{-\beta}\pi)}-i\sin{(2^{-\beta}\pi)}Z-\ii\| = \sqrt{(\cos{(2^{-\beta}\pi)}-1)^2+\sin^2{(2^{-\beta}\pi)}}
	=
	\sqrt{2}|\sin{(2^{-\beta}\pi)}|
	\le \pi 2^{-\beta+1/2}$.
	The cumulative error  $\epsilon\le M\pi2^{-\beta+1/2}$ of all rotations then follows by the triangle inequality on unitary operators $\|(\prod _j U_{j})-(\prod_j\tilde{U}_j)\|\le\sum_{j}\|U_j-\tilde{U}_j\|$.
	Thus the bits of precision required is
	\begin{align}
	\beta=\left\lceil\frac{1}{2}+\log_{2}\left(\frac{M\pi}{\epsilon}\right)\right\rceil.
	\end{align}
	
	There can be an additional error introduced from approximating each single-qubit rotation gate with Clifford $+$ $\t$ gates.
	However, using the phase gradient technique~\cite{Gidey2018Addition} eliminates this error with a worst-case cost of one Toffoli gate per $R_{b}$ rotation.

	\subsubsection{Multiplexed sparse data-lookup}
	\label{sec:MultiplexSparseDataLookup}
	In this section, we describe an implementation of a multiplexed data-lookup oracle.
	This can be non-trivial as standard data-lookup constructions~\cite{Babbush2018encoding,Low2018Trading} are controlled by a single index register.
	Whereas in this case, there can be two or more index registers such as below.
	\begin{align}
	\label{eq:NotationdatalookupMultiplex}
	\image{NotationdatalookupMultiplex}
	\end{align}
	In the above, $j\in[J]$ and $k\in[K]$.
	Thus there are at most $KJ$ bit-strings $\vec{x}_{j,k}$.
	One solution is to map the indices $(j,k)$ to a unique integer $q=jK+k$.
	Thus~\cref{eq:NotationdatalookupMultiplex} can be implemented by a data-lookup oracle controlled by a single index $q\in[JK]$, combined with an arithmetic circuit that computes $q$ from $j$ and $k$ as follows.
	\begin{align}
	\label{eq:NotationdatalookupMultiplexImplementation}
	\image{NotationdatalookupMultiplexImplementation}
	\end{align}
	
	We consider the situation where for each $j$, only $K_j\le K$ bit-strings are defined.
	Thus the multiplexed data-lookup oracle only encodes $Q=\sum_{j\in[J]}K_j$ elements.
	Using the construction of~\cref{eq:NotationdatalookupMultiplexImplementation} is wasteful as it enumerates over $KJ$ elements, which is more than necessary.
	Our solution uses a data-lookup oracle that enumerates over exactly $Q$ elements.
	The basic idea is to map $(j,k)$ to a unique integer $q=k+\sum_{\alpha\in[j-1]}K_\alpha$.
	Note that the shift $Q_j=\sum_{\alpha\in[j-1]}K_\alpha$ can be classically pre-computed.
	Thus this map is implemented by a data-lookup oracle that outputs $Q_j$, followed by an arithmetic circuit that adds $k$ to $Q_j$ as follows.
	\begin{align}
	\label{eq:NotationdatalookupMultiplexImplementation2}
	\image{NotationdatalookupMultiplexImplementation2}
	\end{align}
	As cost is dominated by the single data-lookup oracle in the middle, this construction lends itself readily to Toffoli gate count reduction using additional ancilla qubits, following~\cref{sec:datalookup}.

	\begin{lemma}[Multiplexed sparse data-lookup oracle]
		\label{lem:SparseDataLookup}
		Given a set of bit-strings $\{\vec{x}_{j,k}\in\{0,1\}^{b}\;:\;j\in[N]\;\text{and}\;k\in[K_j]\}$, the data-lookup oracle in~\cref{eq:NotationdatalookupMultiplex} can be implemented using one application of any data-lookup oracle for $Q=\sum_{j\in[J]}K_j$ bit-strings of length $b$, two applications (one computation and one uncomputation) of any data-lookup oracle for $J$ bit-strings of length $\lceil\log_{2}(Q)\rceil$, and two $\lceil\log_{2}(Q)\rceil$-bit arithmetic adders.
	\end{lemma}
	
	There is much flexibility in choosing the implementation of data-lookup oracles in~\cref{lem:SparseDataLookup}.
	In the following corollary, we implement data-lookup on the $Q$ elements using only clean ancilla qubits in~\cref{lem:DataLookupAncilla}, and implement data-lookup on the shift $Q_j$ using dirty qubits~\cref{lem:DataLookupDirtyAncilla}.
	\begin{corollary}[Multiplexed sparse data-lookup oracle with dirty qubits]
		For any integer $\lambda\ge 0$, let $n=1+b(1+\lambda)$. Then the computation of the data-lookup oracle in~\cref{lem:SparseDataLookup} can be implemented using
		\begin{itemize}
			\item Clean qubits: $\max{[\lceil\log_2 Q\rceil, \lceil \log_2 J\rceil]} + \lambda b$.
			\item Toffoli gates:
			$\data{Q,J}{b}{\lambda}=
			\datadirty{J}{\lceil\log_2 Q\rceil}{n}+
			\dataUdirty{J}{\lceil\log_2 Q\rceil}{n}
			\nonumber
			+\data{Q}{b}{\lambda}+2\lceil\log_2 Q\rceil+\mathcal{O}(1).$
		\end{itemize}
		Uncomputation can be implemented using the same number of qubits and
		\begin{itemize}
			\item Toffoli gates:
			$\dataU{Q,J}{b}{\lambda}=
			\datadirty{J}{\lceil\log_2 Q\rceil}{n}+
			\dataUdirty{J}{\lceil\log_2 Q\rceil}{n}
			\nonumber
			+\dataU{Q}{b}{\lambda}+2\lceil\log_2 Q\rceil+\mathcal{O}(1).$
		\end{itemize}
	\end{corollary}
	\begin{proof}
		At the beginning of the circuit we allocate the stated number of qubits.
		We then tabulate the resources required for each operation to ensure that there are sufficient clean qubits available, and sum the Toffoli gate counts.
		\begin{table}[H]
			\centering
			\begin{ruledtabular}\begin{tabular}{c|c|c|c}
					Operation & Clean qubits & Dirty qubits &Toffoli gates
					\\
					&required & available $(n)$ &
					\\\hline\hline
					Lookup on $\vec{Q}_j$~\cref{lem:DataLookupDirtyAncilla} & $\lceil \log_2 J\rceil$ & $\ge\lceil \log_2 Q\rceil+b(1+\lambda)$ & $\datadirty{J}{\lceil\log_2 Q\rceil}{n}$\\
					&&&$+
					\dataUdirty{J}{\lceil\log_2 Q\rceil}{n}$
					\\\hline
					Two adders~\cite{Cuccaro2004Adder} & 0  & n/a &$2\lceil\log_2 Q\rceil+\mathcal{O}(1)$
					\\\hline
					$\vec{x}_q$ Lookup computation~\cref{lem:DataLookupAncilla} & $\lceil\log_2 Q\rceil$ & n/a & $\data{Q}{b}{\lambda}$
					\\\hline
					$\vec{x}_q$ Lookup uncomputation~\cref{eq:Data_uncomputation} & $\lceil\log_2 Q\rceil$ & n/a & $\dataU{Q}{b}{\lambda}$
			\end{tabular}\end{ruledtabular}
		\end{table}
	\end{proof}
	In most applications, particularly later when we use this to block-encode the molecular Hamiltonian, the total number of bit-strings $Q$ is significantly larger than $J$.
	Thus the cost of computation in~\cref{lem:SparseDataLookup} is dominated by $\data{Q}{b}{\lambda}+\mathcal{O}(\qubits{Q}+\data{J}{\qubits{Q}}{n})$.
	
	\subsection{Double-factorized Hamiltonian}
	\label{sec:qubitization_electronic_structure}
	
	The electronic Hamiltonian in first-quantization is
	\begin{align}
	H_{\text{first}}= \left(-\sum_{n\in\text{electrons}}\frac{\nabla_{n}^{2}}{2}-\sum_{m\in\text{nuclei}}\frac{Z_{m}}{|x_{n}-r_{m}|}\right)+\left(\sum_{n_{1},n_{2}\in\text{electrons}}\frac{1}{|x_{n_{1}}-x_{n_{2}}|}\right),
	\end{align}
	where $\nabla^2_n$ is the Laplace operator on the $n^{\text{th}}$ electron, $Z_m$ is the nuclear charge, and $r_m$ is the nucleus coordinate.
	By choosing a basis of orbitals $\psi_i(x)$, this implies the second-quantized representation
	\begin{align}
	H & =\sum_{ij,\sigma}h_{ij}a_{(i,\sigma)}^{\dagger}a_{(j,\sigma)}+\frac{1}{2}\sum_{ijkl,\sigma\rho}h_{ijkl}a_{(i,\sigma)}^{\dagger}a_{(k,\rho)}^{\dagger}a_{(l,\rho)}a_{(j,\sigma)},\\\nonumber
	h_{ij} & =\int\psi^*_{i}(x_{1})\left(-\frac{\nabla^{2}}{2}-\sum_{m}\frac{Z_{m}}{|x_{1}-r_{m}|}\right)\psi_{j}(x_{1})\mathrm{d^{3}}x_{1},\\\nonumber
	h_{ijlk} & =\int\psi^*_{i}(x_{1})\psi_{j}(x_{1})\left(\frac{1}{|x_{1}-x_{2}|}\right)\psi^*_{k}(x_{2})\psi_{l}(x_{2})\mathrm{d^{3}}x_{1}\mathrm{d^{3}}x_{2}.
	\end{align}
	This is equal to the single-factorized $H_{\text{CD}}$ and double-factorized $H_{\text{DF}}$ Hamiltonians
	\begin{align}
	H_{\text{CD}} & \doteq{\sum_{ij,\sigma}\tilde{h}_{ij}a_{(i,\sigma)}^{\dagger}a_{(j,\sigma)}}
	+{\frac{1}{2}\sum_{r\in[R]}\left(\sum_{ij,\sigma}L_{ij}^{(r)}a_{(i,\sigma)}^{\dagger}a_{(j,\sigma)}\right)^{2}},\quad \tilde{h}_{ij}\doteq h_{ij}-\frac{1}{2}\sum_{l}h_{illj},
	\\\nonumber
	H_{\text{DF}} & ={\sum_{ij,\sigma}\tilde{h}_{ij}a_{(i,\sigma)}^{\dagger}a_{(j,\sigma)}}
	+{\frac{1}{2}\sum_{r\in[R]}\left(\sum_{ij,\sigma}\sum_{m\in[M^{(r)}]}\lambda^{(r)}_m \vec{R}_{m,i}^{(r)}\cdot \vec{R}_{m,j}^{(r)\top}a_{(i,\sigma)}^{\dagger}a_{(j,\sigma)}\right)^{2}},
	\end{align}
	where the $\lambda_m^{(r)}$ are eigenvalues of  $L^{(r)}$.

	The dominant cost in qubitizing a Hamiltonian is the synthesis of a unitary quantum circuit $\block{H/\alpha}$ with the property
	\begin{align}
	\block{H/\alpha} & =\left(\begin{array}{cc}
	H/\alpha & \cdots\\
	\vdots & \ddots
	\end{array}\right).
	\end{align}
	This implies that $\block{H/\alpha}\ket 0_{a}\ket{\psi}_{s}=\ket 0_{a}\frac{H}{\alpha}\ket{\psi}_{s}+\ket{0\psi^{\perp}}_{as}$, where the unnormalized residual state $\ket{0\psi^{\perp}}_{as}$ has no support on the ancilla state $\ket 0_{a}$.
	As $H$ is embedded in a contiguous block of the unitary, we say that $B[H/\alpha]$ `block-encodes' the Hamiltonian $H$.
	
	The block-encoding framework supports the addition and multiplication of encoded matrices.
	For instance, one may add block-encoded Hamiltonians $H=\sum_{j}H_{j}$ with a new normalizing
	constant such as $\alpha=\sum_{j}\alpha_{1}$, using the following quantum circuit
	\begin{align}
	\block{H/\alpha} & =\left(\sum_{j}\bra{j}_a\otimes \ii_s\sqrt{\frac{\alpha_{j}}{\alpha}}\right)\left(\sum_{j}\ketbra{j}{j}_a\otimes \block{H_{j}/\alpha_{j}}\right)\left(\sum_{j}\sqrt{\frac{\alpha_{j}}{\alpha}}\ket{j}_a \otimes{\ii}_s\right).
	\end{align}
	As mentioned in eq.~(\ref{notation:BlockEncodingMultiply}), block-encoded Hamiltonians may also be multiplied to obtain $\block{H_{2}H_{1}/\alpha_{2}\alpha_{1}}$.
	In general, these addition and multiplication operations increase
	the ancilla register size in a straightforward manner.
	Using these
	ingredients, one may block-encode matrices specified
	by a variety of common input models, such as sparse matrix oracle,
	linear-combination-unitaries, or other quantum data structures.
	
	In qubitizing $H_{\text{DF}}$, we find it more natural to work in the Majorana representation of the fermion operators
	\begin{align}
	\label{eq:Majorana}
	\gamma_{p,0} & =a_{p}+a_{p}^{\dagger},\quad
	\gamma_{p,1}  =-i\left(a_{p}-a_{p}^{\dagger}\right),\quad
	\left\{ \gamma_{p,x},\gamma_{q,y}\right\}  =2\delta_{pq}\delta_{xy}\ii.
	\end{align}
	Thus
	$a_p=(\gamma_{p,0}+i\gamma_{p,1})/2$ and
	$a_q=(\gamma_{p,0}-i\gamma_{p,1})/2$,
	where $p\doteq (i,\sigma)$ is a combined orbital and spin index.
	Some useful identities are
	\begin{align}
	a_{(i,\sigma)}^{\dagger}a_{(j,\sigma)} + a_{(j,\sigma)}^{\dagger}a_{(i,\sigma)}
	=
	\begin{cases}
	\ii+i\left(\gamma_{i,\sigma,0}\gamma_{i,\sigma,1}\right),&i=j,\\
	\frac{i}{2}\left(\gamma_{i,\sigma,0}\gamma_{j,\sigma,1}+\gamma_{j,\sigma,0}\gamma_{i,\sigma,1}\right),&i\neq j,
	\end{cases}
	\end{align}
	which implies the Majorana representation for one-electron Hamiltonian
	\begin{align}
	\sum_{ij,\sigma}L_{ij}a_{(i,\sigma)}^{\dagger}a_{(j,\sigma)}
	&=\sum_{i}L_{ii}\ii + \onebody_{L},
	\quad
	\onebody_{L}\doteq\frac{i}{2}\sum_{ij}\sum_{\sigma}L_{ij}\gamma_{i,\sigma,0}\gamma_{j,\sigma,1}
	\end{align}
	that separates into a trivial identity component, and a non-trivial one-electron component $\onebody_{L}$.
	As $L$ is a symmetric matrix, it has the eigendecomposition $L=\sum_{k}\lambda_k \vec{R}_{k}\cdot \vec{R}_{k}^\top$.
	Thus we may diagonalize the one-electron Hamiltonian to obtain
	\begin{align}
	\label{eq:one_body_Majorana}
	\onebody_{L}&
	=\frac{i}{2}\sum_{k}\lambda_{k}\sum_{\sigma}\majorana_{\vec{R}_{k},\sigma,0}\majorana_{\vec{R}_{k},\sigma,1},\quad \majorana_{\vec{u},\sigma,x}\doteq\sum_{j}u_j \gamma_{j,\sigma,x}.
	\end{align}
	One should verify that $\majorana_{\vec{u},\sigma,x}^2=\ii$, which follows from~\cref{eq:Majorana} and the unit-length normalization~\cref{eq:one_body_Majorana} of $\vec{u}$.

	Thus the spectral norm $\left\|\onebody_{L}\right\|=\normSH{L}$ is seen to be the one-norm of the eigenvalues of $L$.
	Substituting the Majorana representation into the doubly-factorized Hamiltonian results in
	\begin{align}
	\label{eq:HamiltonianDoubleFactorized2}
	H_\text{DF}&=
	\left(\sum_{i}h_{ii}-\frac{1}{2}\sum_{il}h_{illi}+\frac{1}{2}\sum_{il}h_{llii}\right)\mathcal{I}
	+\onebody_{L^{(-1)}}
	+\frac{1}{2}\sum_{r}\onebody^2_{L^{(r)}},
	\\\nonumber
	L^{(-1)}_{ij}&\doteq h_{ij}-\frac{1}{2}\sum_{l}h_{illj}+\sum_{l}h_{llij}.
	\end{align}
	
	Our strategy for block-encoding~\cref{eq:HamiltonianDoubleFactorized2} is to build it up from block-encodings of its component pieces.
	Let us further collect terms as follows.
	\begin{align}
	\label{eq:HamiltonianCollected}\nonumber
	H_\text{DF}&=\left(\sum_{i}h_{ii}-\frac{1}{2}\sum_{il}h_{illi}+\frac{1}{2}\sum_{il}h_{llii}+\frac{1}{4}\sum_r\normSH{L^{(r)}}^2\right)\mathcal{I}+{\onebody_{L^{(-1)}}}+{\twobody_{H}},
	\\
	\twobody_{H}&\doteq\frac{1}{4}\sum_{r}\normSH{L^{(r)}}^2T_{2}\left[\frac{\onebody_{L^{(r)}}}{\normSH{L^{(r)}}}\right],
	\end{align}
	where $T_2(x)=2x^2-1$ is a Chebyshev polynomial of the first kind.
	The identity term only contributes a constant shift in energy and may be ignored.
	Observe that we have reduced the double-factorized Hamiltonian to sums of products of basis-rotated Majorana operators $\majorana_{\vec{u},\sigma,x}$.

	We thus block-encode the two-body term as follows.
	\begin{enumerate}
		\item Block-encode $\block{\gamma_{\vec{u},\sigma,x}}$ of the basis transformed Majorana operator~\cref{sec:BE_basis_transformed_Majorana}.
		\item Block-encode $\block{\gamma_{\vec{u},\sigma,0}\gamma_{\vec{u},\sigma,1}}$ by multiplying $\block{\gamma_{\vec{u},\sigma,x}}$~\cref{sec:BE_product_basis_transformed_Majorana}.
		\item Block-encode $\block{\frac{\onebody_{L}}{\normSH{L}}}$ by taking a linear combination of
		$\lambda_k\block{\gamma_{\vec{R}_{k},\sigma,0}\gamma_{\vec{R}_{k},\sigma,1}}$ over the eigenvalues and eigenvectors of $L$ and the spins~\cref{sec:BE_one_electron}.
		\item Block-encode $\block{T_2\left[\frac{\onebody_{L}}{\normSH{L}}\right]}$ by applying $\block{\frac{\onebody_{L}}{\normSH{L}}}$ twice using qubitization~\cite{Low2016Qubitization}.
		\item Block-encode $\block{\frac{\twobody_{H}}{\frac{1}{4}\sum_{r}\normSH{L^{(r)}}^2}}$ by taking a linear combination of $\normSH{L^{(r)}}^2\block{T_2\left[\frac{\onebody_{L^{(r)}}}{\normSH{L^{(r)}}}\right]}$ over the rank components of the two-electron tensor~\cref{sec:BE_twoe_electron}.
		\item Block-encode $\block{\frac{H_\text{DF}}{\normSH{L^{(-1)}}+\frac{1}{4}\sum_{r}\normSH{L^{(r)}}^2}}$ by adding one- and two-body terms~\cref{sec:BE_costs}.
	\end{enumerate}
	Generally, the cost of block-encoding the large number of two-electron terms dominate that of the smaller number of one-electron terms.
	A common theme throughout will be the use of symmetries.
	Many coefficients turn out to be identical.
	For instance, $L^{(r)}_{ij}=L^{(r)}_{ji}$, and is independent of spin.
	Moreover, the same coefficients $\vec{R}^{(r)}_{k}$ occur in both of $\gamma_{\vec{u},\sigma,0}\gamma_{\vec{u},\sigma,1}$.
	Wherever possible, we use this redundancy to optimize the number of bits of classical data we need to encode into our quantum circuits.
	These optimizations are combined with recent advances using ancillary qubits~\cite{Low2018Trading} to substantially reduce Toffoli gate count
	
	We now provide optimized quantum circuits that implement the above steps.
	We make heavy use of quantum circuit notation, outlined in~\cref{sec:preliminaries}.
	
	\subsubsection{Block-encoded, basis-transformed Majorana operator}
	\label{sec:BE_basis_transformed_Majorana}
	
	We synthesize the basis-transformed Majorana operator $\gamma_{\vec{u},\sigma,x}$ by conjugating $\gamma_{0,\sigma,x}$ with a sequence of unitary rotations.
	The required sequence of unitary rotations follows from the following observation.
	\begin{lemma}[Sum of Majorana operators by Majorana rotations]
		\label{lem:blockencodingMajoranaRot}
		Let the unitary $U_{\vec{u}}$ be the sequence
		\begin{align}
		U_{\vec{u},\sigma,x}\doteq V^{(0)}_{\vec{u},\sigma,x} V^{(1)}_{\vec{u},\sigma,x}\cdots V^{(N-2)}_{\vec{u},\sigma,x},\quad V^{(p)}_{\vec{u},\sigma,x}\doteq e^{\theta_p\gamma_{p,\sigma,x}\gamma_{p+1,\sigma,x}}.
		\end{align}
		Then for all $\sigma\in\{0,1\}$ and $x\in\{0,1\}$, there exists rotation angles $\theta_p\doteq\theta_{\vec{u},p}$ that are function of $\vec{u}$ such that $U^\dagger_{\vec{u},\sigma,x}\cdot \gamma_{0,\sigma,x}\cdot U_{\vec{u},\sigma,x} = \gamma_{\vec{u},\sigma,x}$.
	\end{lemma}
	\begin{proof}
		In the interests of clarity, we drop the $\sigma,x$ subscript in the following.
		By taking a Taylor series expansion, observe that $V^{(p)}_{\vec{u}}=\cos{(\theta_{p})}\ii+\sin{(\theta_p)}\gamma_{p}\gamma_{p+1}$.
		Thus
		\begin{align}
		\label{eq:MajoranaRotation}
		V^{(p)\dagger}_{\vec{u}}\gamma_{q}V^{(p)}_{\vec{u}}=
		\begin{cases}
		\gamma_{q},&\quad q\neq p,p+1,\\
		\cos{(2\theta_p)}\gamma_{p}+\sin{(2\theta_p)}\gamma_{p+1},&\quad q=p,\\
		\cos{(2\theta_p)}\gamma_{p+1}-\sin{(2\theta_p)}\gamma_{p},&\quad q=p+1.
		\end{cases}
		\end{align}
		Thus $U^\dagger_{\vec{u}}\cdot \gamma_{0}\cdot U_{\vec{u}}=\sum_{p\in[N]}u_p\gamma_{p}$, by choosing
		\begin{align}
		u_0=\cos{(2\theta_0)},\quad u_1=\sin{(2\theta_0)}\cos(2\theta_1),\quad\cdots\quad,u_p=\cos(2\theta_p)\prod_{j<p}\sin(2\theta_j).
		\end{align}
		We obtain the angles $\theta_p$ by recursively solving this linear chain of equations.
	\end{proof}
	
	As $\gamma_{\vec{u},\sigma,x}$ is unitary, it is also trivially its own block-encoding $\block{\gamma_{\vec{u},\sigma,x}}=\gamma_{\vec{u},\sigma,x}$.

	\subsubsection{Block-encoded product of Majorana operators}
	\label{sec:BE_product_basis_transformed_Majorana}
	
	We synthesize the product $\gamma_{\vec{u},\sigma,0}\gamma_{\vec{u},\sigma,1}$ using the implementation described in~\cref{lem:blockencodingMajoranaRot}.
	Observe that the commutator $[\gamma_{p,\sigma,x}\gamma_{p+1,\sigma,x},\gamma_{q,\rho,y}\gamma_{q+1,\rho,y}]$ is zero for all $\sigma\neq \rho$ or $x\neq y$.
	Thus all rotations in $U_{\vec{u},\sigma,0}$ and $U_{\vec{u},\sigma,1}$ commute with each other.
	This allows us to collect rotations in the product
	\begin{align}
	\label{eq:MajoranaRotationProductGates}
	\gamma_{\vec{u},\sigma,0}\gamma_{\vec{u},\sigma,1}=
	\left((V^{(0)}_{\vec{u},\sigma,0} V^{(0)}_{\vec{u},\sigma,1})\cdots (V^{(N-2)}_{\vec{u},\sigma,0} V^{(N-2)}_{\vec{u},\sigma,1})\right)^\dagger\gamma_{0,\sigma,0}\gamma_{0,\sigma,1}
	\underbrace{(V^{(0)}_{\vec{u},\sigma,0} V^{(0)}_{\vec{u},\sigma,1})\cdots (V^{(N-2)}_{\vec{u},\sigma,0} V^{(N-2)}_{\vec{u},\sigma,1})}
	_{U_{\vec{u},\sigma,0}U_{\vec{u},\sigma,1}}.
	\end{align}
	
	In a Pauli representation of the Majorana operators, each rotation
	\begin{align}
	\label{eq:CliffordConjugation}
	V_{\vec{u},p,x}=C_{p,x}^\dagger\cdot e^{i\theta_{\vec{u},p}Z}\cdot C_{p,x}
	\end{align}
	is implemented by a single-qubit $Z$ rotation conjugated by some Clifford gate $C_{p,x}$ such that $C_{p,x}^\dagger \cdot iZ\cdot C_{p,x} =\gamma_{p,\sigma,x}\gamma_{p+1,\sigma,x}$.
	We use the Jordan-Wigner representation (momentarily ignoring the spin index) that maps
	\begin{align}
	\label{eq:JW_nospin}
	\gamma_{p,0}\gamma_{p+1,0}\rightarrow -iY_{p}X_{p+1}, \quad \gamma_{p,1}\gamma_{p+1,1}\rightarrow iX_{p}Y_{p+1},
	\quad \gamma_{0,0}\gamma_{0,1} \rightarrow iZ_0.
	\end{align}
	In the example where $N=4$, the rotated Majorana operator is hence implemented by the circuit
	\begin{align}
	\label{eq:MajoranaRotationProductB}
	\image{MajoranaRotationProductB},
	\end{align}
	where the basis transformation is
	\begin{align}
	\label{eq:MajoranaRotationProduct}
	\image{MajoranaRotationProduct},
	\end{align}
	which, upon substituting the Jordan-Wigner representation and defining the $Z$ Pauli rotation $R_\theta\doteq e^{i\theta Z}$, is equal to
	\begin{align}
	\label{eq:MajoranaRotationProductC}
	\image{MajoranaRotationProductC}.
	\end{align}
	
	With the spin indices restored, the Jordan-Wigner representation maps
	\begin{align}
	\gamma_{p,1,0}\gamma_{p+1,1,0}\rightarrow -iY_{p+N}X_{p+1+N}, \quad \gamma_{p,1,1}\gamma_{p+1,1,1}\rightarrow iX_{p+N}Y_{p+1+N},
	\quad \gamma_{0,1,0}\gamma_{0,1,1} \rightarrow iZ_N,
	\end{align}
	for $\sigma=1$.
	The case $\sigma=0$ is identical to~\cref{eq:JW_nospin}.
	Thus the quantum circuit for
	\begin{align}
	\gamma_{\vec{u},1,0}\gamma_{\vec{u},1,1} = \prod_{j=0}^{N-1}\left(\textsc{SWAP}_{j\leftrightarrow j+N}\right)\cdot \gamma_{\vec{u},0}\gamma_{\vec{u},1}\prod_{j=0}^{N-1}\cdot\left(\textsc{SWAP}_{j\leftrightarrow j+N}\right),
	\end{align}
	is implemented by swapping all pairs of qubits $j\leftrightarrow j+N$.
	Thus the spin-multiplexed unitary $\ketbra{0}{0}\otimes\gamma_{\vec{u},0,0}\gamma_{\vec{u},0,1}+\ketbra{1}{1}\otimes\gamma_{\vec{u},1,0}\gamma_{\vec{u},1,1}$ is implemented by adding a control to the $\textsc{SWAP}$ gates, as illustrated below (for $N=4$).
	\begin{align}
	\label{eq:MajoranaRotationMultiplexE}
	\image{MajoranaRotationMultiplexE}.
	\end{align}
	
	As $\gamma_{\vec{u},\sigma,0}\gamma_{\vec{u},\sigma,1}$ is unitary, it is also trivially its own block-encoding $\block{\gamma_{\vec{u},\sigma,0}\gamma_{\vec{u},\sigma,1}}=\gamma_{\vec{u},\sigma,0}\gamma_{\vec{u},\sigma,1}$.

	\subsubsection{Block-encoded one-electron operator and its square}
	\label{sec:BE_one_electron}
	
	The one-electron operator is
	\begin{align}
	\onebody_{L}&\doteq\frac{1}{2}\sum_{k\in[K]}\lambda_{k}\sum_{\sigma}\majorana_{\vec{R}_{k},\sigma,0}\majorana_{\vec{R}_{k},\sigma,1},\quad L=\sum_{k\in[K]}\lambda_k\vec{R}_k\vec{R}_k^{\dagger}.
	\end{align}
	The block-encoding of $\onebody_{L}$ is accomplished by adding block-encodings of $\block{\gamma_{\vec{u},\sigma,0}\gamma_{\vec{u},\sigma,1}}$, each weighted by the eigenvalue $\lambda_k$, following the construction in~\cref{notation:BlockEncodingAdd}.
	This requires synthesizing the quantum state
	\begin{align}
	\ket{\vec{\lambda}}=\sum_{k}\sqrt{|\lambda_k|/\normSH{L}}\ket{k}\ket{\text{sign}[\lambda_k]}
	\end{align}
	and the multiplexed unitary
	\begin{align}
	\label{eq:multiplexedMajoranaProduct}
	Z\otimes \sum_{k}\sum_{\sigma}\ketbra{k}{k}\otimes \ketbra{\sigma}{\sigma}\otimes \block{\gamma_{\vec{u},\sigma,0}\gamma_{\vec{u},\sigma,1}}.
	\end{align}
	Note that we encode the sign of $\lambda_k$ in a single qubit using the Pauli $Z$ operation, as described in~\cref{eq:BlockEncodingWithSign}, which uses the fact $\bra{\text{sign}[\lambda_k]}Z\ket{\text{sign}[\lambda_k]}=\text{sign}[\lambda_k]$.
	In the following circuit diagrams, this sign qubit is assumed to be implicitly present and will not be shown.
	Given these two components, the block-encoding is implemented by the quantum circuit
	\begin{align}
	\label{eq:blockEncodingOneDoubleA}
	\image{blockencodingOneDoubleA}.
	\end{align}
	
	We now discuss the implementation of~\cref{eq:multiplexedMajoranaProduct} which is controlled by the spin index $\sigma$ and the eigenvalue index $k$.
	Control by the spin index remains unchanged from~\cref{eq:MajoranaRotationMultiplexE}.
	Control by the eigenvalue index is implemented by multiplexing the basis transformation
	\begin{align}
	\label{eq:MultiplexBasis}
	\sum_{k}\ketbra{k}{k}\otimes U_{\vec{R}_k,0}U_{\vec{R}_k,1}.
	\end{align}
	From~\cref{eq:MajoranaRotationProductC}, $U_{\vec{R}_k,0}U_{\vec{R}_k,1}$ is a sequence of $Z$-rotations by a $k$-dependent angle $\theta_{\vec{R}_k,p}$, each conjugated by a Clifford gate that is independent of $k$.
	Thus~\cref{eq:MultiplexBasis} may be implemented by multiplexing the phase-rotations
	\begin{align}
	\sum_{k}\ketbra{k}{k}\otimes e^{i\theta_{\vec{R}_k,p}Z},
	\end{align}
	as follows
	\begin{align}
	\label{eq:MajoranaRotationProductD}
	\image{MajoranaRotationProductD}.
	\end{align}
	
	We implement multiplexed phase-rotations using a data-lookup oracle $D_{p}\ket{k}\ket{z}=\ket{k}\ket{z\oplus\tilde{\theta}_{\vec{R}_k,p}}$ on $K$ $\beta$-bit entries.
	This oracle computes a $\beta$-bit binary representation of the rotation angle $\frac{\theta_{\vec{R}_k,p}}{2\pi} \approx \tilde{\theta}_{\vec{R}_k,p}= \sum_{b=0}^{\beta-1}\tilde{\theta}_{\vec{R}_k,p,b}/2^{1+b}\in[0,1-2^{-\beta}]$, and may be implemented according to~\cref{sec:datalookup}.
	With this oracle, we perform the following sequence
	\begin{align}
	\label{eq:MultiplexPhaseRotations}
	D_{p}\ket{k}\ket{0}\ket{\psi}
	&=\ket{k}\ket{\tilde{\theta}_{\vec{R}_k,p}}\ket{\psi}=\ket{k}\left(\bigotimes_{b=0}^{\beta-1}\ket{\tilde{\theta}_{\vec{R}_k,p,b}}\right)\ket{\psi}\\\nonumber
	\text{Controlled rotations}&\rightarrow\ket{k}\left(\bigotimes_{b=0}^{\beta-1}\ket{\tilde{\theta}_{\vec{R}_k,p,b}}\right)\prod_{b=0}^{\beta-1}e^{i2\pi\tilde{\theta}_{\vec{R}_k,p,b}Z}\ket{\psi}=\ket{k}\ket{\tilde{\theta}_{\vec{R}_k,p,b}}e^{i2\pi\tilde{\theta}_{\vec{R}_k,p}Z}\ket{\psi}\\\nonumber
	\text{Uncompute}&\rightarrow\ket{k}\ket{0}e^{i2\pi\tilde{\theta}_{\vec{R}_k,p}Z}\ket{\psi}.
	\end{align}
	This costs two queries to data-lookup $D_{p}$ and $\beta$ arbitrary single-qubit rotations.
	As the circuit for~\cref{eq:multiplexedMajoranaProduct} contains $4N$ multiplexed rotations, these rotations costs at most $8N$ queries to data-lookup on $K$ $\beta$-bit entries, and $4N\beta$ arbitrary single-qubit rotations.

	However, some optimizations are possible.
	First, the rotation angle bits computed once by $D_{p}$ may be used to implement both rotations $R_{\tilde{\theta}_{\vec{R}_k,p}}$ in~\cref{eq:MajoranaRotationProductD}.
	This reduces queries by a factor of two to $4N$.
	Second, the uncomputation of $\ket{\tilde{\theta}_{\vec{R}_k,p}}$ may be merged with the computation of $\tilde{\theta}_{\vec{R}_k,p+1}$.
	Instead of applying $D_{p}$ and $D_{p+1}$, we merge them into a single data-lookup oracle $D_{p,p+1}\ket{k}\ket{z}=\ket{p}\ket{z\oplus\tilde{\theta}_{\vec{R}_k,p}\oplus\tilde{\theta}_{\vec{R}_k,p+1}}$.
	This reduces queries by another factor of two to $2N$.
	Third, the data-lookup oracles may output $\kappa>\beta$ bits.
	Thus multiple angles may be written out at once, e.g. $\ket{z_0\oplus\tilde{\theta}_{\vec{R}_k,p}}\ket{z_1\oplus\tilde{\theta}_{\vec{R}_k,p+1}}\cdots\ket{z_{\kappa/\beta-1}\oplus\tilde{\theta}_{\vec{R}_k,p+\kappa/\beta-1}}$.
	These angles allow the application of more controlled-rotations for each application of data-lookup.
	This is detailed in~\cref{sec:programmableRotation}, and reduces the number of queries to $2\lceil\frac{N\beta}{\kappa}\rceil$, but now to data-lookup on $K$ $\kappa$-bit entries.
	Fourth, many of arbitrary single-qubit rotations are by identical angles (powers of $2$).
	These may be implemented exactly using $4N\beta$ Toffoli gates by the phase gradient technique~\cite{Gidey2018Addition}.
	This also requires a one-time cost of preparing and storing $\beta$ single-qubit resource states whose cost is negligible and thus ignored.
	Fifth, using the choice $\kappa=N\beta$, only one computation and uncomputation step is required.
	This can be advantageous as uncomputation as described in~\cref{sec:datalookup} has a Toffoli cost that can be smaller than computation when many ancillary qubits are used.
	
	We now bound the bits of precision $\beta$ required to approximate~\cref{eq:multiplexedMajoranaProduct} to an overall error of $\epsilon$ in spectral norm.
	As each rotation angle $\theta_{\vec{R}_k,p}$ is approximated to an error of at most $\pi2^{-\beta}$, the error of each unitary rotation compared to its binary approximation is \begin{align}\|R_{\theta_{\vec{R}_k,p}}-R_{2\pi\tilde{\theta}_{\vec{R}_k,p}}\|\le\|R_{\pi2^{-\beta}}-\ii\|=\|\cos{(2^{-\beta}\pi)}-i\sin{(2^{-\beta}\pi)}Z-\ii\|
	\end{align}
	The cumulative error  $\epsilon\le (4N)\pi2^{-\beta+1/2}$ of all rotations then follows by the triangle inequality on unitary operators $\|(\prod _j U_{j})-(\prod_j\tilde{U}_j)\|\le\sum_{j}\|U_j-\tilde{U}_j\|$.
	Thus the bits of precision required is
	\begin{align}
	\beta=\left\lceil\frac{1}{2}+\log_{2}\left(\frac{4N\pi}{\epsilon}\right)\right\rceil.
	\end{align}
	Including the error $\mu$ of state-preparation results in an $\epsilon'=\epsilon+\mu$-approximate block-encoding $\blockApprox{\frac{\onebody_{L}}{\normSH{L}}}{\epsilon+\mu}$ instead of the exact  $\block{\frac{\onebody_{L}}{\normSH{L}}}$.
	A simple choice of error budgeting is $\epsilon=\epsilon'/2$ and $\mu=\epsilon'/2$, though this may be optimized with more weight on $\epsilon$ as its data-lookup oracles are more costly.
	
	\begin{table}
		\caption{\label{tab:OneElectron} Resources used to block-encode a one-electron operator $\blockApprox{\frac{\onebody_{L}}{\normSH{L}}}{\epsilon}$ on $N$ orbitals where $L$ has $K$ eigenvectors. The parameters $\beta=\left\lceil5.152+\log_{2}\left(\frac{N}{\epsilon}\right)\right\rceil$ and $\mu=2+\lceil\log_{2}(1/\epsilon)\rceil$.}
		\centering
		\begin{ruledtabular}\begin{tabular}{cccc}
				Operation & \# & Toffoli cost each & Ancillary qubits\\
				\hline
				Data-lookup~\cref{lem:DataLookupAncilla}&$1$&$\data{K}{N\beta}{\lambda}+\dataU{K}{N\beta}{\lambda}$&$2\lceil\log_2{K}\rceil+N\beta(1+\lambda)$\\
				Arbitrary rotations~\cite{Gidey2018Addition} & $4N\beta$ & $1$ & 0 \\
				Controlled-\textsc{SWAP}s & $2N$ & $1$ & 0 \\
				State-preparation~\cref{lem:StatePreparationGarbage} & $2$ & $\mu+\datadirty{K}{\qubits{K}+\mu+1}{N\beta(1+\lambda)+2N+1}$ & $\lceil\log_{2}{K}\rceil+2\mu+1$
		\end{tabular}\end{ruledtabular}
	\end{table}
	
	The overall cost of block-encoding the one-electron operator is summarized in~\cref{tab:OneElectron}, and is dominated by data-lookup in multiplexing the basis transformation rotations.
	As state-preparation is cheap in comparison, we choose an implementation that is assisted by $\kappa+\lambda+2N$ dirty qubits.
	Given $\blockApprox{\frac{\onebody_{L}}{\normSH{L}}}{\epsilon}$, qubitization in~\cref{def:qubitization} describes how it may be applied twice, hence doubling the cost, to block-encode $\blockApprox{T_2[\frac{\onebody_{L}}{\normSH{L}}]}{2\epsilon}$.
	This also applies a reflection about the all-zero state $\ket{0\cdots0}$ on $\mathcal{O}(\log_{K}+\mu)$ qubits, whose cost is negligible and thus ignored.

	\subsubsection{Block-encoded two-electron operator}
	\label{sec:BE_twoe_electron}
	In this section, we describe a block-encoding of the double-factorized Hamiltonian in~\cref{eq:HamiltonianCollected}.
	We focus on the two-body term, which is the most costly component.
	\begin{align}
	\twobody_{H}&\doteq\frac{1}{4}\sum_{r\in[R]}\normSH{L^{(r)}}^2 T_{2}\left[\frac{\onebody_{L^{(r)}}}{\normSH{L^{(r)}}}\right],
	\quad \overrightarrow{\Lambda_{\text{SH}}}\doteq (\normSH{L^{(1)}}^2,\normSH{L^{(2)}}^2,\cdots, \normSH{L^{(R)}}^2)^\top,\\
	\nonumber
	L^{(r)}&=\sum_{k\in[M^{(r)}]}\lambda_k^{(r)}\vec{R}_k^{(r)}\vec{R}_k^{(r)\top}.
	\end{align}
	At a high level, we achieve this block-encoding using the following circuit identities, which are defined in~\cref{sec:preliminaries}.
	\begin{align}
	\label{eq:blockencodingTwoDoubleFullA}\text{Add block-encodings~\cref{notation:BlockEncodingAdd}:}\;&\image{blockencodingTwoDoubleFullA},\\
	\label{eq:blockencodingTwoDoubleFullB}\text{Qubitization~\cref{eq:qubitization}:}\;&\image{blockencodingTwoDoubleFullB},\\
	\label{eq:blockencodingTwoDoubleFullC}\text{Add block-encodings~\cref{notation:BlockEncodingAdd}:}\;&\image{blockencodingTwoDoubleFullC}.
	\end{align}
	Compared to the one-electron operator in~\cref{eq:blockEncodingOneDoubleA}, we see that the only new element is multiplexing by another control variable $r$.
	Thus instead of the singly-multiplexed basis transformation in~\cref{eq:MajoranaRotationProductD}, we implement the doubly-multiplexed basis transform
	\begin{align}
	\label{eq:MajoranaRotationProductE}
	\image{MajoranaRotationProductE}.
	\end{align}
	Following the implementation of multiplexed-phase rotations in~\cref{eq:MultiplexPhaseRotations},~\cref{eq:MajoranaRotationProductE} can be implemented given a multiplexed data-lookup oracle $DD_p\ket{r}\ket{k}\ket{z}=\ket{r}\ket{k}\ket{z\oplus \theta_{\vec{R}_k^{(r)},p}}$.
	
	Two challenges make the implementation of $DD_p$ non-obvious.
	First, the number of entries $k\in[M^{(r)}]$ can depend on the index $r$.
	Rather than storing the worst-case number of $R\max_{r}M^{(r)}$ entries, storing only $\sum_{r\in[R]}M^{(r)}$ entries would minimize the cost of data-lookup.
	Second, the multiplexed data-lookup oracle should allow for the best possible trade-off between Toffoli gates and assisting ancillary qubits.
	Our implementation of $DD_p$ is detailed in~\cref{sec:MultiplexSparseDataLookup}.
	
	The next challenge is multiplexed state preparation of $\ket{\vec{\lambda}^{(r)}}$.
	By drawing~\cref{eq:blockencodingTwoDoubleFullA} in full and inserting identity in the middle,
	\begin{align}
	\label{eq:blockencodingTwoDoubleFullE}
	\image{blockencodingTwoDoubleFullE},
	\end{align}
	this suggests preparing the quantum state represented by
	\begin{align}
	\image{blockencodingTwoDoubleFullF}
	&=
	\sum_{r\in[R]}\sqrt{\frac{\normSH{L^{(r)}}^2}{\|\overrightarrow{\Lambda_{\text{SH}}}\|_1}}\ket{r}\ket{\vec{\lambda}^{(r)}}
	=
	\sum_{r\in[R]}\sqrt{\frac{\normSH{L^{(r)}}^2}{\|\overrightarrow{\Lambda_{\text{SH}}}\|_1}}\ket{r}\sum_{k\in[M^{(r)}]}\sqrt{\frac{\lambda^{(r)}_k}{\normSH{L^{(r)}}}}\ket{k}\nonumber\\
	&=
	\sum_{r\in[R]}\sum_{k\in[M^{(r)}]}\sqrt{\frac{\normSH{L^{(r)}}\lambda^{(r)}_k}{\|\overrightarrow{\Lambda_{\text{SH}}}\|_1}}\ket{r}\ket{k}.
	\end{align}
	However, it is difficult to apply the reflection in this case.
	
	Thus we use the multiplexed sparse data-lookup oracle of~\cref{lem:SparseDataLookup} in the state preparation routine of~\cref{lem:MultiplexSparseStatePreparationGarbage}.
	In this case, the desired state with $M=\sum_{r\in[R]}M^{(r)}$ coefficients can be created with a cost that scales with that outlined in~\cref{lem:MultiplexSparseStatePreparationGarbage}.
	
	\begin{table}
		\centering
		\caption{\label{tab:TwoElectron}Resources used to block-encode a two-electron operator $\blockApprox{\frac{\twobody_{H}}{\frac{1}{4}\sum_{r}\normSH{L^{(r)}}^2}}{\epsilon}$ on $N$ orbitals where $M$ is the number of eigenvectors in the doubly-factorized representation. Above, $\beta=\left\lceil5.652+\log_{2}\left(\frac{N}{\epsilon}\right)\right\rceil$, $\mu=\lceil2.5+\log_{2}(1/\epsilon)\rceil$, and $C=\datadirty{M}{2\qubits{\max_{r}M^{(r)}}+2\qubits{R}+\mu+1}{\lambda+\kappa+2N}$}
		\begin{ruledtabular}\begin{tabular}{cccc}
				Operation & Applications & Toffoli cost each & Ancillary qubits\\
				\hline
				Data-lookup~\cref{lem:SparseDataLookup}&$2$&$
				\data{M,R}{N\beta}{\lambda}+\dataU{M,R}{N\beta}{\lambda}$ & $\max{[\lceil\log_2 M\rceil, \lceil \log_2 R\rceil]} + \lambda N\beta$\\
				Arbitrary rotations~\cite{Gidey2018Addition} & $8N\beta$ & $1$ & 0 \\
				Controlled-\textsc{SWAP}s & $4N$ & $1$ & 0 \\
				State-preparation~\cref{lem:MultiplexSparseStatePreparationGarbage} & $4$ & $\mu+C$ & $\qubits{K}+2\mu+1$
		\end{tabular}\end{ruledtabular}
	\end{table}

	The overall cost of block-encoding the two-electron operator is summarized in~\cref{tab:TwoElectron}, and is dominated by data-lookup in multiplexing the basis transformation rotations.

	\subsubsection{Asymptotic costs}
	\label{sec:BE_costs}
	Although Toffoli gate costs may be obtained by summing over elements in~\cref{tab:TwoElectron}, it is informative to evaluate asymptotic costs for large $N$.
	Note that empirical fits by Peng et al.~\cite{Peng2017Cholesky} observe that $R=\mathcal{O}(N)$ and $M^{(r)}=\mathcal{O}(\log{N})$ when $N\gg 1000$, in the case where $N$ scales with the number of atoms in the systems considered.
	In our main text, we vary the active space $N\le 250$ for systems with a fixed number of atoms and observe a more conservative scaling of $R=\mathcal{O}(N^{1.5})$, $M^{(r)}=\mathcal{O}(N)$, and hence $M=\mathcal{O}(N^{2.5})$.
	We also choose the number of qubits that store the fermionic basis transform angles to be $\kappa=N\beta$, which is just a constant factor more than the $2N$ qubits representing spin-orbitals, as $\beta=\left\lceil5.652+\log_{2}\left(\frac{N}{\epsilon}\right)\right\rceil$.

	Thus the computation of data-lookup has a cost
	\begin{align}
	\data{M,R}{N\beta}{\lambda}&\le2\datadirty{R}{\qubits{M}}{\qubits{M}+N\beta+\lambda+2N}+\data{M}{N\beta}{\lambda}+2\qubits{M}+\mathcal{O}(1)\\\nonumber
	&=\data{M}{N\beta}{\lambda}+\mathcal{O}(\sqrt{R\qubits{M}})\\\nonumber
	&=\min_{\lambda'\in[0,\lambda]}\left(\frac{M}{1+\lambda'}+N\beta\lambda'\right)+\mathcal{O}(\sqrt{R\qubits{M}}),
	\end{align}
	and its uncomputation has a cost
	\begin{align}
	\dataU{M,R}{N\beta}{\lambda''}
	&=\min_{\lambda'\in[0,\lambda'']}\left(\frac{M}{1+\lambda'}+\lambda'\right)+\mathcal{O}(\sqrt{R\qubits{M}}).
	\end{align}
	As the number of clean ancilla qubits used by the computation step is $\sim\lambda'N\beta$, the uncomputation step may choose $\lambda''=\lambda'N\beta$.
	With this choice, the Toffoli cost of $\dataU{M,R}{N\beta}{\lambda''}$ is always a factor of at least $\sqrt{N\beta}$ smaller than $\data{M,R}{N\beta}{\lambda}$, so long as $M=\Omega(N\beta)$, and is sub-dominant.
	For instance, in the case of $\lambda=0$, we obtain almost-linear scaling with respect to $M$ as follows.
	\begin{align}
	\data{M,R}{N\beta}{0}&=M+\mathcal{O}(\sqrt{R\qubits{M}})=\mathcal{O}(N\log{N}).
	\end{align}
	Choosing $\lambda=\Theta(\sqrt{M/(N\beta)})$ minimizes this as $\data{M,R}{N\beta}{\Theta(\sqrt{M/(N\beta)})}=\mathcal{O}(\sqrt{MN\beta})+\mathcal{O}(\sqrt{R\qubits{M}})$.
	%
	
	
	State-preparation has a cost that is sub-dominant to data-lookup.
	It suffices to use $\mu=\lceil2.5+\log_{2}(1/\epsilon)\rceil$ bits of precision.
	\begin{align}
	\state{M}{2\qubits{R}+\mu+1}{}&=\mu+\datadirty{M}{b}{\lambda+\kappa+2N}\\\nonumber
	&=\min_{\lambda'\in[0,\lambda+N(2+\beta)]}\left(2\frac{M}{1+\lfloor\frac{\lambda'}{b}\rfloor}+4b\left\lfloor\frac{\lambda'}{b}\right\rfloor\right)+\mu,
	\end{align}
	where $b=2\qubits{\max_{r\in[R]}M^{(r)}}+2\qubits{R}+\mu+1=\mathcal{O}(\log{N}+\mu)$.
	As the number of dirty qubits $N(2+\beta) \gg b$,
	\begin{align}
	\state{M}{2\qubits{R}+\mu+1}{}&\lesssim \min\left[M, 4\sqrt{2Mb}\right]+\mu=\mathcal{O}(\sqrt{M\log{(1/\epsilon)}}).
	\end{align}
	
	Thus the total number of Toffoli gates is
	\begin{align}
	&2(\data{M,R}{N\beta}{\lambda}+\dataU{M,R}{N\beta}{\lambda})+4(\state{M}{2\qubits{R}+\mu+1}{}+N(2\beta+1))\\\nonumber
	&\le \min_{\lambda'\in[0,\lambda]}\left(\frac{2M}{1+\lambda'}+{2\lambda'}N\beta+4N(2\beta+1)\right)+\mathcal{O}(\sqrt{R\qubits{M}}+\sqrt{M\log{(1/\epsilon)}}).
	\end{align}
	
	Note that converting this block-encoding to a quantum walk using the qubitization procedure~\cref{def:qubitization} requires an additional reflection $\textsc{REF}$ about the state $\ket{0\cdots0}$ on the ancillary register of the block-encoding.
	This reflection may be implemented using a number of Toffoli gates that is at most equal to the number of qubits comprising the ancillary register. 
	Importantly, this ancillary register excludes any additional clean or dirty qubits that are used in intermediate steps, such as in data-lookup. 
	Thus the cost of $\textsc{REF}$ is subdominant to that of the block-encoding.
	
	\section{Matrix Schatten and entrywise $1$-norm}
	Let $h\in\mathbb{C}^{N\times N}$ be any Hermitian matrix.
	This matrix has eigenvalues $h\cdot U_k=\sigma_k U_k$, where the orthonormal $k^{\text{th}}$ eigenvector $U_k$ can be viewed as the $k^{th}$ column of some unitary matrix $U\in\mathbb{C}^{N\times N}$.
	We now prove that the following inequality is true
	\begin{align}
	\frac{1}{N}\normEW{h}\le\normSH{h}\le\normEW{h},
	\end{align}
	where
	\begin{align}
	\normSH{h}\doteq\sum_{k\in[N]}|\sigma_k|,\quad \normEW{h}\doteq\sum_{j,k\in[N]}|h_{jk}|.
	\end{align}
	Moreover, we show that this inequality is tight. We have not been able to find a proof of this statement in the literature, so our derivation may be of independent interest.
	
	To prove the lower bound, observe that
	\begin{align}
	\sum_{k\in[N]}|\sigma_k|
	\ge \sqrt{\sum_{k\in[N]}|\sigma_k|^2}
	=\sqrt{\sum_{j,k\in[N]}|h_{jk}|^2}
	\ge \frac{1}{N}\sum_{j,k\in[N]}|h_{jk}|.
	\end{align}
	The middle equality follows from the definition of the Frobenius norm, and the first and last inequalities follow from the usual vector norm inequality $\sqrt{\sum_{j\in[M]}|x_j|^2}\le \sum_{j\in[M]}|x_j|\le \sqrt{M\sum_{j\in[M]}|x_j|^2}$ for any vector $x$.
	This is an equality when every entry in $h$ is one, and is hence tight.
	
	To prove the upper bound, observe that
	\begin{align}
	\sum_{k\in[N]}|\sigma_k|
	&=\sum_{k\in[N]}| (U_k)^\dagger\cdot h\cdot U_k|
	=\sum_{k\in[N]}\left|\sum_{i,j\in[N]} U^{\dagger}_{ki}h_{ij} U_{jk}\right|
	\le \sum_{i,j\in[N]}\left(\sum_{k\in[N]}\left| U^{\dagger}_{ki}\right|\left|U_{jk}\right|\right)| h_{ij}|\\\nonumber
	&\le \sum_{i,j\in[N]}|h_{ij}|.
	\end{align}
	The final inequality follows from applying the Cauchy-Schwartz inequality on the term in round brackets.
	Using the fact that any row or column of any unitary has unit $2$-norm,
	\begin{align}
	\sum_{k\in[N]}\left| U^{\dagger}_{ki}\right|\left|U_{jk}\right|
	\le
	\sqrt{\sum_{k\in[N]}\left| U^{\dagger}_{ki}\right|^2}\sqrt{\sum_{k\in[N]}\left| U_{jk}\right|^2}=1.
	\end{align}
	This upper bound is also an equality by choosing $h$ to contain only a single non-zero entry that lies on the diagonal, and is hence tight.

	\section{M06-L/def2-SVP optimized Cartesian coordinates for stable intermediate I}
	\label{sec:coords-I}
	The M06-L/def2-SVP optimized Cartesian coordinates of stable intermediate I given in the Supporting Information of Ref.~\cite{wesselbaum15} were erroneous and we obtained the correct ones in a private communication from Dr. Markus H\"olscher. We reproduce them here for further reference.
	\begin{verbatim}
	100
	
	H         5.400100000000     -3.818200000000     -1.311700000000
	H         4.155000000000     -2.978300000000     -3.299400000000
	C         4.393600000000     -3.410200000000     -1.197400000000
	C         3.697900000000     -2.939300000000     -2.308200000000
	H         4.317900000000     -3.761200000000      0.932900000000
	C         3.788500000000     -3.373200000000      0.059400000000
	C         2.410900000000     -2.418300000000     -2.164100000000
	H         1.899900000000     -2.055300000000     -3.059100000000
	H        -1.522700000000     -6.160300000000     -1.987500000000
	H        -2.707500000000     -6.642600000000      0.150600000000
	H         6.870300000000      1.841700000000      0.070400000000
	H         5.886300000000     -0.074000000000     -1.189900000000
	C         2.503900000000     -2.857300000000      0.201800000000
	C        -1.511200000000     -5.428600000000     -1.176900000000
	C        -2.172700000000     -5.699200000000      0.022400000000
	C         1.799100000000     -2.357700000000     -0.904400000000
	H        -0.296300000000     -4.039200000000     -2.277800000000
	C         5.790200000000      1.776700000000     -0.076800000000
	C        -0.827200000000     -4.228100000000     -1.339600000000
	C         5.240300000000      0.706700000000     -0.781200000000
	H         2.028700000000     -2.855300000000      1.184600000000
	C        -2.142100000000     -4.766000000000      1.055100000000
	H        -2.657200000000     -4.972500000000      1.996200000000
	H         5.373500000000      3.615100000000      0.978600000000
	C         4.951900000000      2.768800000000      0.431800000000
	C        -0.805400000000     -3.275100000000     -0.307900000000
	H         0.001800000000     -1.791600000000     -3.029400000000
	C        -1.462300000000     -3.558400000000      0.891500000000
	C         3.861600000000      0.623600000000     -0.968800000000
	H         3.456300000000     -0.237100000000     -1.503900000000
	H         1.680200000000      0.186000000000     -2.572700000000
	H        -1.451300000000     -2.818200000000      1.696600000000
	H         0.111400000000     -0.038800000000     -4.688600000000
	P         0.105700000000     -1.704000000000     -0.588100000000
	C         3.574200000000      2.692000000000      0.238400000000
	C        -0.526000000000     -1.214100000000     -2.256800000000
	C         3.009600000000      1.616600000000     -0.465900000000
	H         2.932700000000      3.481100000000      0.641600000000
	H        -1.563400000000     -1.577000000000     -2.300200000000
	H        -1.642400000000     -0.121600000000     -4.426900000000
	C         0.891300000000      0.908400000000     -2.309900000000
	C        -0.723200000000      0.400500000000     -4.122100000000
	H         1.056100000000      1.773900000000     -2.969600000000
	C        -0.491700000000      0.285600000000     -2.613500000000
	H        -0.817400000000      1.451300000000     -4.432100000000
	P         1.181300000000      1.455800000000     -0.562300000000
	Ru       -0.047300000000      0.041300000000      0.958600000000
	H         1.827000000000      3.646000000000     -2.436100000000
	H        -5.061500000000     -2.431900000000      2.138800000000
	H        -3.264200000000     -0.743000000000      2.004000000000
	C        -4.719600000000     -2.084100000000      1.161400000000
	C        -3.716000000000     -1.124500000000      1.083500000000
	C        -1.633700000000      1.056800000000     -1.906800000000
	C         1.154200000000      4.050900000000     -1.674200000000
	H        -2.570100000000      0.942400000000     -2.474500000000
	C         0.687500000000      3.231600000000     -0.632200000000
	H         0.039700000000      1.450700000000      2.022500000000
	H        -6.054700000000     -3.375300000000      0.054400000000
	C        -5.277400000000     -2.610600000000     -0.004100000000
	H        -1.408400000000      2.135900000000     -1.932200000000
	P        -1.937200000000      0.628900000000     -0.136500000000
	C        -3.262700000000     -0.655000000000     -0.158400000000
	H         1.167400000000      6.016000000000     -2.553900000000
	C         0.794000000000      5.393400000000     -1.738000000000
	C        -4.842400000000     -2.149000000000     -1.243200000000
	C        -3.853900000000     -1.166300000000     -1.320200000000
	C        -0.122100000000      3.803300000000      0.353300000000
	H        -5.280600000000     -2.545000000000     -2.161700000000
	H        -0.484400000000      3.193700000000      1.183100000000
	H        -3.557200000000     -0.808300000000     -2.309000000000
	C        -0.034300000000      5.944900000000     -0.757900000000
	C        -2.891400000000      2.068700000000      0.492900000000
	C        -0.485300000000      5.149800000000      0.291700000000
	H        -2.541300000000      1.552100000000      2.567000000000
	H        -3.463700000000      2.849500000000     -1.453300000000
	H        -0.318600000000      6.998000000000     -0.811400000000
	C        -3.024300000000      2.250000000000      1.876400000000
	C        -3.530900000000      2.968800000000     -0.369700000000
	H        -1.127000000000      5.571000000000      1.068900000000
	C        -3.765000000000      3.313500000000      2.386400000000
	C        -4.267400000000      4.036400000000      0.141900000000
	H        -3.859700000000      3.439500000000      3.467100000000
	H        -4.755200000000      4.732400000000     -0.543900000000
	C        -4.384000000000      4.213900000000      1.519100000000
	H        -4.962100000000      5.050200000000      1.917600000000
	H        -0.952100000000     -0.951500000000      1.897300000000
	C         1.382600000000     -1.559900000000      3.369100000000
	C         3.440300000000     -0.516100000000      3.819500000000
	C         2.732600000000     -1.853200000000      3.979000000000
	H         0.855300000000     -2.435000000000      2.959000000000
	H         0.698800000000     -1.064900000000      4.082100000000
	H         4.534000000000     -0.585300000000      3.873900000000
	H         3.119200000000      0.180400000000      4.609200000000
	H         3.250700000000     -2.640200000000      3.408000000000
	H         2.664400000000     -2.195000000000      5.019400000000
	O         1.665600000000     -0.662300000000      2.280600000000
	C         2.957100000000     -0.044500000000      2.462600000000
	H         2.840700000000      1.047200000000      2.388100000000
	H         3.612800000000     -0.372800000000      1.637900000000
	H        -0.485500000000      0.871800000000      2.422700000000
	\end{verbatim}
	
	\section{PBE/def2-TZVP optimized Cartesian coordinates of intermediates and transition states}
	\subsection{Stable intermediate I}
	\begin{verbatim}
	100
	
	H     5.3467652   -4.0620811   -1.5408059
	H     4.0926440   -3.1077082   -3.4724386
	C     4.3733646   -3.5936827   -1.3864493
	C     3.6714082   -3.0599560   -2.4668243
	H     4.3466242   -3.9621961    0.7431586
	C     3.8127382   -3.5354622   -0.1077529
	C     2.4188585   -2.4692755   -2.2728935
	H     1.8982865   -2.0719643   -3.1452685
	H    -1.2303386   -6.3183450   -1.9390001
	H    -2.3770141   -6.8302644    0.2162593
	H     6.8894103    1.9218885   -0.0035308
	H     5.9314313    0.0065669   -1.2845387
	C     2.5645983   -2.9459478    0.0853854
	C    -1.2664022   -5.5702254   -1.1455328
	C    -1.9074697   -5.8571319    0.0630226
	C     1.8479758   -2.4005269   -0.9928855
	H    -0.1392416   -4.1388533   -2.2801310
	C     5.8092962    1.8284340   -0.1278109
	C    -0.6597761   -4.3310091   -1.3395956
	C     5.2727557    0.7583394   -0.8455377
	H     2.1307566   -2.9221419    1.0848827
	C    -1.9386697   -4.8957205    1.0724767
	H    -2.4361609   -5.1112241    2.0195865
	H     5.3587395    3.6295988    0.9806459
	C     4.9524295    2.7841309    0.4226824
	C    -0.6965107   -3.3523507   -0.3302362
	H    -0.0341901   -1.8215111   -3.0690435
	C    -1.3390064   -3.6491474    0.8753093
	C     3.8905920    0.6406711   -1.0111176
	H     3.5035862   -0.2172111   -1.5606581
	H     1.6854069    0.1946881   -2.6002976
	H    -1.3758827   -2.8904869    1.6590155
	H     0.1002470   -0.0584394   -4.7135150
	P     0.1555732   -1.7405875   -0.6273468
	C     3.5724202    2.6731840    0.2513664
	C    -0.5320094   -1.2444568   -2.2769909
	C     3.0210291    1.6010240   -0.4716906
	H     2.9209553    3.4398669    0.6738715
	H    -1.5758438   -1.5926867   -2.2691807
	H    -1.6572136   -0.1365786   -4.4433076
	C     0.8887494    0.9035529   -2.3328536
	C    -0.7336832    0.3839739   -4.1478944
	H     1.0248231    1.7778900   -2.9853773
	C    -0.4950701    0.2645326   -2.6290296
	H    -0.8248285    1.4381977   -4.4496939
	P     1.1802478    1.4418906   -0.5767848
	Ru   -0.0362848    0.0070243    0.9119775
	H     1.8212552    3.6119472   -2.4832596
	H    -5.1810334   -2.3446956    2.1435413
	H    -3.2016593   -0.8755944    1.9833979
	C    -4.8598719   -1.9775745    1.1673356
	C    -3.7477001   -1.1434336    1.0769211
	C    -1.6402163    1.0360418   -1.9181179
	C     1.2043586    4.0438832   -1.6928281
	H    -2.5718953    0.9180300   -2.4893150
	C     0.7332662    3.2400964   -0.6383483
	H    -0.0098568    1.4084727    1.9225571
	H    -6.4276167   -3.0022283    0.0834920
	C    -5.5584608   -2.3460940    0.0148140
	H    -1.4088869    2.1127432   -1.9239510
	P    -1.9316460    0.5878885   -0.1460388
	C    -3.3208409   -0.6412765   -0.1628399
	H     1.3004971    6.0156855   -2.5548840
	C     0.9255300    5.4088262   -1.7289309
	C    -5.1381297   -1.8649201   -1.2242237
	C    -4.0347117   -1.0109405   -1.3118186
	C    -0.0024344    3.8482106    0.3847007
	H    -5.6743639   -2.1451819   -2.1324256
	H    -0.3751631    3.2487936    1.2152503
	H    -3.7533400   -0.6366137   -2.2967743
	C     0.1786679    5.9998088   -0.7051680
	C    -2.8846015    2.0320140    0.5076653
	C    -0.2785850    5.2188262    0.3545650
	H    -2.5050619    1.5265310    2.5789523
	H    -3.5109659    2.7932772   -1.4293835
	H    -0.0387882    7.0686817   -0.7335780
	C    -3.0151135    2.2055321    1.8935469
	C    -3.5734918    2.9065438   -0.3467612
	H    -0.8581185    5.6709513    1.1612919
	C    -3.8039041    3.2316844    2.4142423
	C    -4.3560825    3.9389862    0.1746224
	H    -3.8959846    3.3477833    3.4954999
	H    -4.8804404    4.6123560   -0.5054398
	C    -4.4734473    4.1060444    1.5550677
	H    -5.0901960    4.9096892    1.9606077
	H    -0.9280923   -1.0292973    1.7973158
	C     1.3378154   -1.5319181    3.4011766
	C     3.2728388   -0.3202031    4.0156118
	C     2.6497891   -1.7175315    4.1406225
	H     0.8939997   -2.4524575    3.0001828
	H     0.5828711   -1.0207579    4.0231085
	H     4.3644121   -0.3235517    4.1291005
	H     2.8553795    0.3490331    4.7819271
	H     3.2715028   -2.4723203    3.6357164
	H     2.5036511   -2.0320824    5.1819182
	O     1.6798343   -0.6894853    2.2644245
	C     2.8450502    0.1293127    2.6183662
	H     2.5545055    1.1877945    2.5764244
	H     3.6086816   -0.0587448    1.8513541
	H    -0.5264219    0.7815522    2.3441392
	\end{verbatim}
	\subsection{Stable intermediate II}
	\begin{verbatim}
	90
	
	H    -3.5339194   -5.8134021    0.7154303
	H    -2.8828583   -4.5857087    2.7848281
	C    -2.7878389   -5.0187876    0.6707413
	C    -2.4249571   -4.3298811    1.8276514
	H    -2.4367085   -5.2430098   -1.4508818
	C    -2.1745456   -4.6972246   -0.5430550
	C    -1.4648308   -3.3145957    1.7730738
	H    -1.1972253   -2.8074961    2.7010195
	H     4.0447509   -5.0032186    1.5248357
	H     4.6257704   -5.4865551   -0.8498353
	H    -6.9424288   -1.0572568   -0.8489737
	H    -5.6541614   -2.1885011    0.9605540
	C    -1.2131616   -3.6888161   -0.5962258
	C     3.5324311   -4.4778357    0.7172409
	C     3.8577104   -4.7485084   -0.6135181
	C    -0.8501479   -2.9745776    0.5597285
	H     2.2928246   -3.3649200    2.0694393
	C    -5.9424138   -0.7091293   -0.5867670
	C     2.5443401   -3.5405037    1.0224599
	C    -5.2212192   -1.3414937    0.4256799
	H    -0.7092122   -3.4805713   -1.5422166
	C     3.1943722   -4.0721958   -1.6392311
	H     3.4419659   -4.2786302   -2.6819584
	H    -5.9318348    0.8799617   -2.0540790
	C    -5.3756389    0.3750716   -1.2625045
	C     1.8842706   -2.8451860   -0.0030466
	H     0.5482637   -1.8718378    2.8271484
	C     2.2180869   -3.1221608   -1.3355503
	C    -3.9409031   -0.8959159    0.7676030
	H    -3.4025233   -1.4184765    1.5573864
	H    -1.7520289   -0.5852893    2.4090356
	H     1.7216499   -2.5731918   -2.1386833
	H    -0.2276327   -0.3572974    4.5490732
	P     0.4870886   -1.7047662    0.3801736
	C    -4.0998357    0.8234416   -0.9225842
	C     0.8525218   -1.1077287    2.0974671
	C    -3.3662972    0.1971043    0.1037136
	H    -3.6794455    1.6861233   -1.4444495
	H     1.9502640   -1.0454284    2.1616566
	H     1.4427245    0.2226901    4.3439378
	C    -1.2653987    0.3736318    2.1828482
	C     0.3997269    0.3856686    4.0338455
	H    -1.7195915    1.1197769    2.8510531
	C     0.2529006    0.2656288    2.5028157
	H     0.0928914    1.3838068    4.3798787
	P    -1.6693102    0.8441027    0.4312741
	Ru    0.0415500    0.0787294   -1.0430842
	H    -3.4364971    2.3732688    2.2276231
	H     5.5508989   -1.1292919   -1.7375133
	H     3.1714975   -0.4324962   -1.7053632
	C     5.1686562   -0.5511105   -0.8947219
	C     3.8293119   -0.1582569   -0.8800472
	C     1.0515757    1.4409814    1.8698538
	C    -2.9679142    3.0880432    1.5473566
	H     1.9394366    1.6412793    2.4854269
	C    -2.0549238    2.6476564    0.5712613
	H    -0.4369552    1.4561661   -2.0038409
	H     7.0572918   -0.5236336    0.1563813
	C     6.0113482   -0.2132042    0.1644635
	H     0.4333868    2.3509809    1.9064705
	P     1.5743010    1.2220596    0.1041735
	C     3.3108499    0.5814438    0.1914567
	H    -4.0258849    4.7597303    2.3994422
	C    -3.3134529    4.4347265    1.6394014
	C     5.5093639    0.5318195    1.2339090
	C     4.1732011    0.9330985    1.2450332
	C    -1.5123513    3.5856436   -0.3140454
	H     6.1619527    0.8113919    2.0626371
	H    -0.8027279    3.2670739   -1.0783201
	H     3.8259940    1.5437210    2.0793896
	C    -2.7593723    5.3639604    0.7524551
	C     1.9325683    2.9291860   -0.4951383
	C    -1.8619483    4.9369995   -0.2252055
	H     2.2575338    2.2337655   -2.5225002
	H     1.7383905    3.9532013    1.4139813
	H    -3.0337675    6.4176549    0.8231227
	C     2.2606460    3.0950337   -1.8510831
	C     1.9739393    4.0452473    0.3535690
	H    -1.4265658    5.6533717   -0.9236661
	C     2.6072277    4.3501706   -2.3492191
	C     2.3212349    5.3021967   -0.1484847
	H     2.8621323    4.4615150   -3.4043972
	H     2.3465729    6.1612497    0.5240211
	C     2.6353199    5.4591196   -1.4987626
	H     2.9081739    6.4412764   -1.8877254
	H     1.1990551   -0.5574839   -2.0047565
	O    -1.5496481   -1.0274756   -2.4549650
	H     0.2386696    1.0535691   -2.4415570
	C    -2.5477139   -1.5735677   -2.7519342
	O    -3.5175158   -2.1258482   -3.0858970
	\end{verbatim}
	\subsection{Transition state II-III}
	\begin{verbatim}
	90
	
	H    -3.3110113    5.9892207    1.0614647
	H    -4.2699731    4.2457791    2.5646655
	C    -3.0206823    4.9443578    0.9411378
	C    -3.5551641    3.9679033    1.7884389
	H    -1.7138037    5.3312324   -0.7361600
	C    -2.1264645    4.5776821   -0.0631642
	C    -3.1811684    2.6340605    1.6402700
	H    -3.6277160    1.8806884    2.2931144
	H    -5.6160706   -2.7121710    1.1039665
	H    -6.9243089   -1.7456590   -0.7848797
	H     2.2922500    6.7135450   -1.2644906
	H     1.7977477    4.9663614   -2.9725257
	C    -1.7490437    3.2402931   -0.2142156
	C    -5.2161958   -1.8742172    0.5304563
	C    -5.9476568   -1.3334376   -0.5263815
	C    -2.2602892    2.2598538    0.6443485
	H    -3.4101655   -1.8109259    1.6817662
	C     2.1117394    5.6823443   -0.9566678
	C    -3.9621095   -1.3533817    0.8608104
	C     1.8365741    4.7034730   -1.9140823
	H    -1.0589769    2.9633107   -1.0083301
	C    -5.4252112   -0.2587683   -1.2511941
	H    -5.9898498    0.1730898   -2.0785599
	H     2.3737919    6.0988230    1.1483993
	C     2.1569557    5.3382049    0.3964698
	C    -3.4314232   -0.2746544    0.1400394
	H    -1.9093419    0.6457296    2.8766625
	C    -4.1811057    0.2707808   -0.9169211
	C     1.5998822    3.3856847   -1.5195481
	H     1.3618145    2.6317415   -2.2730301
	H     0.1324406    2.1468070    2.1220360
	H    -3.7991235    1.1276967   -1.4762929
	H    -0.1622500    0.9200354    4.5156734
	P    -1.7898612    0.4866969    0.4682011
	C     1.9314756    4.0197348    0.7944430
	C    -1.3581619   -0.0181335    2.1948070
	C     1.6445597    3.0334958   -0.1618108
	H     1.9840510    3.7707294    1.8557653
	H    -1.7398248   -1.0358147    2.3616180
	H    -0.3053609   -0.8547770    4.5168659
	C     0.8384432    1.3039855    2.0506034
	C     0.2504868   -0.0066745    4.0898762
	H     1.6781083    1.5528995    2.7144019
	C     0.1467645    0.0036826    2.5516048
	H     1.2993952   -0.0924259    4.4116776
	P     1.4154808    1.2704176    0.2981193
	Ru   -0.0580826   -0.0105605   -1.0282051
	H     3.4484813    0.7669848    2.5245745
	H    -3.2870024   -4.5938061   -1.2882619
	H    -1.8154783   -2.6079569   -1.3493182
	C    -2.5166489   -4.4718823   -0.5251153
	C    -1.6833561   -3.3513062   -0.5612275
	C     0.8531797   -1.2773087    2.0363862
	C     3.9135837    0.5869699    1.5546861
	H     0.5775641   -2.1172447    2.6901033
	C     3.1837715    0.7222675    0.3644458
	H     0.9114846   -0.7396226   -2.2942093
	H    -3.0277398   -6.2905288    0.5234630
	C    -2.3720066   -5.4190782    0.4879301
	H     1.9394546   -1.1417143    2.1330924
	P     0.5082569   -1.7662838    0.2799070
	C    -0.6882511   -3.1688492    0.4063762
	H     5.8094743    0.1187158    2.4675051
	C     5.2626919    0.2217700    1.5284322
	C    -1.3765958   -5.2519381    1.4551523
	C    -0.5332464   -4.1429149    1.4097780
	C     3.8478677    0.5010974   -0.8521713
	H    -1.2477823   -5.9948543    2.2439597
	H     3.3167662    0.6211521   -1.7996011
	H     0.2634056   -4.0610900    2.1519374
	C     5.9063693   -0.0053451    0.3122289
	C     2.0092923   -2.6851927   -0.2709110
	C     5.1945846    0.1439739   -0.8807738
	H     1.1253032   -3.0341700   -2.2209773
	H     3.1671390   -2.5809671    1.5682949
	H     6.9580223   -0.2951589    0.2922381
	C     1.9984905   -3.1903690   -1.5823841
	C     3.1239399   -2.9388468    0.5397158
	H     5.6875698   -0.0238491   -1.8396120
	C     3.0862125   -3.9048771   -2.0797414
	C     4.2143493   -3.6557303    0.0388926
	H     3.0601414   -4.2877262   -3.1012456
	H     5.0784948   -3.8334063    0.6811896
	C     4.2035194   -4.1325937   -1.2711213
	H     5.0591475   -4.6860065   -1.6611601
	H    -1.2075135   -0.8482710   -1.9275095
	O    -0.8679910    1.3150199   -2.6888645
	H     1.4073377   -0.0947416   -1.9902143
	C    -1.4523573    0.2818584   -3.0210830
	O    -2.1601399   -0.2683888   -3.8009511
	\end{verbatim}
	\subsection{Stable intermediate V}
	\begin{verbatim}
	101
	
	H     1.8269673   -6.4995226   -1.8331652
	H    -0.1987775   -5.8444899   -3.1316653
	C     1.3357599   -5.5452970   -1.6361930
	C     0.2004959   -5.1775756   -2.3656564
	H     2.7207290   -4.9627857   -0.0853206
	C     1.8313185   -4.6876922   -0.6553924
	C    -0.4354660   -3.9645232   -2.1074060
	H    -1.3415289   -3.7201738   -2.6648629
	H    -5.4381068   -3.4308995   -1.4499820
	H    -5.7231633   -4.1081567    0.9313397
	H     6.8126364   -2.2405036   -0.5533612
	H     5.0462895   -3.1389195   -2.0682447
	C     1.2007010   -3.4643685   -0.4040235
	C    -4.6524851   -3.2205525   -0.7223230
	C    -4.8130288   -3.5981606    0.6118353
	C     0.0608975   -3.0881173   -1.1250500
	H    -3.3956802   -2.2873749   -2.1843592
	C     5.8318240   -1.7643652   -0.5971816
	C    -3.4881512   -2.5699563   -1.1352722
	C     4.8442785   -2.2658724   -1.4452872
	H     1.6041304   -2.7885258    0.3494945
	C    -3.8030790   -3.3164865    1.5343161
	H    -3.9218460   -3.6006504    2.5814336
	H     6.3275628   -0.2279502    0.8421000
	C     5.5595456   -0.6389062    0.1844456
	C    -2.4662044   -2.2841338   -0.2159046
	H    -1.1415037   -1.6116831   -3.1391415
	C    -2.6401863   -2.6642794    1.1234962
	C     3.5908833   -1.6518536   -1.5113046
	H     2.8440551   -2.0769602   -2.1801291
	H     1.3623215   -0.9359586   -2.8199790
	H    -1.8736126   -2.4241663    1.8588721
	H    -0.1338313   -0.3032581   -4.8952739
	P    -0.8563124   -1.5244629   -0.7170093
	C     4.3112658   -0.0213328    0.1140267
	C    -1.2345476   -0.8090490   -2.3930810
	C     3.3046685   -0.5194259   -0.7333012
	H     4.1294584    0.8753389    0.7083816
	H    -2.2879545   -0.4913255   -2.3861758
	H    -1.5963034    0.6712922   -4.6105906
	C     1.1412825    0.1082683   -2.5575217
	C    -0.5327454    0.5648136   -4.3493954
	H     1.7583057    0.7306354   -3.2204188
	C    -0.3554164    0.3870354   -2.8286795
	H    -0.0043198    1.4621509   -4.7045371
	P     1.6863344    0.3895464   -0.8090269
	Ru    0.0389462   -0.0005865    0.7252379
	H     3.2250833    1.9918618   -2.8187053
	H    -5.5616749   -0.0730762    1.5080617
	H    -3.0995036    0.0692163    1.4036755
	C    -5.0912348    0.5200997    0.7222509
	C    -3.6967614    0.5926086    0.6606666
	C    -0.8058406    1.7022572   -2.1467946
	C     3.1041163    2.5953737   -1.9176497
	H    -1.6802305    2.1081867   -2.6727341
	C     2.3760418    2.1111900   -0.8156072
	H    -6.9628336    1.1363562   -0.1658733
	C    -5.8742870    1.1959233   -0.2124547
	H    -0.0032940    2.4459640   -2.2507600
	P    -1.2230027    1.5568335   -0.3491218
	C    -3.0659053    1.3453680   -0.3394307
	H     4.2796205    4.2051116   -2.7335698
	C     3.7218335    3.8443641   -1.8676434
	C    -5.2560505    1.9631166   -1.2035989
	C    -3.8653575    2.0450311   -1.2620198
	C     2.3018202    2.8984359    0.3427203
	H    -5.8584004    2.5116529   -1.9297298
	H     1.7439962    2.5460789    1.2092780
	H    -3.4174294    2.6879964   -2.0210913
	C     3.6439259    4.6211635   -0.7074815
	C    -1.1572670    3.2886783    0.2895171
	C     2.9387037    4.1425219    0.3957026
	H    -2.3118478    2.7421554    2.0422977
	H    -0.0363799    4.1968011   -1.3387761
	H     4.1381050    5.5932326   -0.6663737
	C    -1.8041316    3.5490818    1.5106727
	C    -0.5524518    4.3537888   -0.3921529
	H     2.8715419    4.7397278    1.3065679
	C    -1.8384943    4.8383125    2.0375311
	C    -0.5987907    5.6484078    0.1320905
	H    -2.3486429    5.0219847    2.9844605
	H    -0.1308477    6.4664138   -0.4182261
	C    -1.2372162    5.8945216    1.3471326
	H    -1.2740729    6.9067416    1.7528829
	C     0.6236532   -2.5200666    2.8855962
	C     2.4293019   -1.6723168    4.1713895
	C     1.6770189   -2.9768655    3.8829715
	H     0.3086067   -3.3025469    2.1820304
	H    -0.2608002   -2.1002196    3.3881616
	H     3.4302248   -1.8336599    4.5922643
	H     1.8592748   -1.0454939    4.8739104
	H     2.3496344   -3.7192823    3.4274900
	H     1.2272891   -3.4258756    4.7779230
	O     1.2614393   -1.4553573    2.1062169
	C     2.4874088   -1.0259966    2.7960327
	H     2.4919838    0.0700512    2.8166982
	H     3.3427331   -1.3902960    2.2082809
	C    -0.4839753    0.8530861    3.0969704
	O    -1.2519694   -0.0067675    2.5614182
	O     0.4999295    1.3293584    2.4521989
	H    -0.6736869    1.1882045    4.1351206
	\end{verbatim}
	\subsection{Stable intermediate VIII}
	\begin{verbatim}
	92
	
	H     5.3410691    4.6103038   -0.7286968
	H     5.6351375    2.6660762   -2.2625989
	C     4.6692268    3.7531301   -0.6607933
	C     4.8326629    2.6645963   -1.5231454
	H     3.5170406    4.5758757    0.9712987
	C     3.6517439    3.7331876    0.2910867
	C     3.9818701    1.5656267   -1.4308562
	H     4.1541688    0.7093404   -2.0863048
	H     5.3441144   -3.5263402   -1.2875126
	H     6.3265024   -3.4109990    0.9988512
	H    -0.2194183    7.1237700    1.5006299
	H     0.3818143    6.5065105   -0.8378647
	C     2.7938990    2.6325550    0.3767315
	C     4.9201146   -2.8513314   -0.5421202
	C     5.4694667   -2.7877908    0.7387700
	C     2.9447130    1.5387409   -0.4810763
	H     3.4202316   -2.1181527   -1.8906689
	C    -0.3175908    6.0908944    1.1634668
	C     3.8276699   -2.0491872   -0.8814573
	C     0.0166225    5.7455689   -0.1461321
	H     1.9949747    2.6417497    1.1182516
	C     4.9206702   -1.9178771    1.6848405
	H     5.3472060   -1.8581592    2.6875602
	H    -1.0603318    5.3701172    3.0625313
	C    -0.7863201    5.1079734    2.0394093
	C     3.2713351   -1.1743183    0.0622773
	H     2.1946169    0.0350266   -2.7286776
	C     3.8295682   -1.1182642    1.3496498
	C    -0.1124367    4.4242007   -0.5823675
	H     0.1555086    4.1841315   -1.6115250
	H     0.7015766    2.1672788   -1.9844840
	H     3.4102245   -0.4384432    2.0936124
	H     0.7653964    0.9882288   -4.3831132
	P     1.9113448    0.0104144   -0.3093269
	C    -0.9160337    3.7893734    1.6059889
	C     1.4246694   -0.3701459   -2.0562641
	C    -0.5756526    3.4303314    0.2915612
	H    -1.3039069    3.0346863    2.2937037
	H     1.4373852   -1.4647132   -2.1728465
	H     0.3132470   -0.7317237   -4.4733566
	C    -0.2371225    1.5929680   -1.9597140
	C     0.0395239    0.2389821   -4.0335098
	H    -0.9410518    2.0976545   -2.6364120
	C     0.0333062    0.1608490   -2.4939166
	H    -0.9527088    0.5195701   -4.4176326
	P    -0.8962277    1.6899576   -0.2303813
	Ru   -0.0238083    0.0640677    1.0843357
	H    -2.5842051    3.1984707   -2.1387174
	H     2.2239112   -5.1614667    0.6413664
	H     1.2762615   -2.8955149    0.9246553
	C     1.3794016   -4.8504076    0.0239379
	C     0.8463580   -3.5672878    0.1844869
	C    -1.1018803   -0.8114226   -2.0894157
	C    -3.2491400    2.6641909   -1.4576714
	H    -1.1186280   -1.6618118   -2.7870707
	C    -2.7300027    1.8238793   -0.4557841
	H    -1.4667938    0.0410227    1.9091058
	H     1.2550573   -6.7239207   -1.0431248
	C     0.8398346   -5.7222469   -0.9197954
	H    -2.0642370   -0.2869640   -2.1904927
	P    -1.0195287   -1.4892020   -0.3653207
	C    -0.2343198   -3.1473190   -0.5990082
	H    -5.0079336    3.5173888   -2.3599551
	C    -4.6232194    2.8624432   -1.5765175
	C    -0.2441571   -5.3130698   -1.7032069
	C    -0.7805561   -4.0379677   -1.5410167
	C    -3.6189847    1.2110202    0.4344335
	H    -0.6801801   -5.9938817   -2.4362601
	H    -3.2318789    0.5675248    1.2245666
	H    -1.6491382   -3.7494463   -2.1369265
	C    -5.5019440    2.2380120   -0.6867936
	C    -2.7446225   -2.0357360   -0.0188175
	C    -4.9962038    1.4189572    0.3213980
	H    -2.2503101   -2.5061510    2.0364853
	H    -3.5747692   -1.7216522   -2.0057686
	H    -6.5773453    2.4000441   -0.7765359
	C    -3.0238833   -2.5230440    1.2670865
	C    -3.7517322   -2.0852668   -0.9931031
	H    -5.6731647    0.9323251    1.0256158
	C    -4.2804323   -3.0393222    1.5770725
	C    -5.0128779   -2.5994562   -0.6799775
	H    -4.4773207   -3.4147045    2.5825244
	H    -5.7884109   -2.6233868   -1.4472268
	C    -5.2816081   -3.0765302    0.6026616
	H    -6.2665695   -3.4801891    0.8428596
	O     0.8694350   -1.3975618    2.5338386
	C     0.3062612   -1.6250898    3.6154095
	H     0.7705045   -2.2836480    4.3663928
	O    -0.8503584   -1.1356522    3.9761388
	H    -1.1598523   -0.5655116    3.1644263
	H     0.9748962    1.1899218    1.9454773
	H     0.1735184    1.2213257    2.3440611
	\end{verbatim}
	
	\subsection{Transition state VIII-IX}
	\begin{verbatim}
	90
	
	H    -4.6740691    5.2558103    0.7566860
	H    -3.6565780    4.3324926    2.8349996
	C    -4.1354118    4.3074042    0.7292914
	C    -3.5661585    3.7903366    1.8922326
	H    -4.4625976    3.9910424   -1.3853212
	C    -4.0157646    3.5992936   -0.4700558
	C    -2.8873823    2.5684966    1.8623977
	H    -2.4691171    2.1869670    2.7946513
	H    -5.9128906   -2.6923533    1.2233453
	H    -6.4076056   -2.8870920   -1.2100055
	H     1.1573864    6.8809743   -1.0129677
	H    -0.8506188    5.4569338   -1.4299762
	C    -3.3307875    2.3861081   -0.5044693
	C    -5.2652676   -2.2056811    0.4921476
	C    -5.5429165   -2.3139708   -0.8721107
	C    -2.7615192    1.8505448    0.6650799
	H    -3.9799328   -1.3760665    2.0008758
	C     1.1655870    5.8257319   -0.7358246
	C    -4.1653394   -1.4654646    0.9292547
	C     0.0437524    5.0287268   -0.9742557
	H    -3.2431514    1.8484092   -1.4501893
	C    -4.7123334   -1.6855437   -1.8016271
	H    -4.9237644   -1.7656228   -2.8692849
	H     3.1803520    5.8816053    0.0436688
	C     2.2991020    5.2658453   -0.1434690
	C    -3.3203178   -0.8399957    0.0003320
	H    -2.1425426    0.2094387    2.9427620
	C    -3.6018057   -0.9591481   -1.3700645
	C     0.0578379    3.6763202   -0.6261028
	H    -0.8283114    3.0659516   -0.8063629
	H    -0.2237945    2.0407436    2.2035585
	H    -2.9507921   -0.4807551   -2.1055095
	H    -0.4689263    0.7373564    4.6029011
	P    -1.9054598    0.2207814    0.5050528
	C     2.3198105    3.9137427    0.2033037
	C    -1.4836231   -0.3009075    2.2257558
	C     1.1963386    3.1045336   -0.0325657
	H     3.2192495    3.4885514    0.6490400
	H    -1.7275057   -1.3727420    2.2903341
	H    -0.3976348   -1.0403009    4.5706550
	C     0.5618104    1.2714340    2.1562441
	C     0.0646738   -0.1240085    4.1740517
	H     1.3560399    1.6141361    2.8354213
	C     0.0029848   -0.1002639    2.6344315
	H     1.1056714   -0.0885412    4.5286102
	P     1.1937713    1.3136515    0.4186758
	Ru   -0.0502612    0.0907209   -0.9359617
	H     3.2400038    1.0439519    2.6791411
	H    -2.2708536   -4.6301779   -2.2456786
	H    -0.9002952   -2.6033156   -1.8754508
	C    -1.7962634   -4.4684538   -1.2766370
	C    -1.0245996   -3.3266690   -1.0650562
	C     0.8745284   -1.2724616    2.1150515
	C     3.7308124    0.9132315    1.7139380
	H     0.6474571   -2.1747953    2.7003186
	C     3.0008941    0.9637346    0.5156672
	H     1.1979307   -0.0817098   -2.0463828
	H    -2.5683667   -6.2918788   -0.4107251
	C    -1.9641008   -5.3980079   -0.2486223
	H     1.9333601   -1.0382542    2.2997734
	P     0.7143326   -1.6350116    0.3026225
	C    -0.4035018   -3.0981142    0.1720287
	H     5.6638259    0.6708030    2.6390232
	C     5.1129118    0.7071425    1.6975290
	C    -1.3517570   -5.1807289    0.9870476
	C    -0.5717110   -4.0422139    1.1958857
	C     3.6904433    0.8118030   -0.6981318
	H    -1.4745639   -5.9043785    1.7944560
	H     3.1425629    0.8554028   -1.6419038
	H    -0.0857470   -3.9164201    2.1644913
	C     5.7854488    0.5561642    0.4846131
	C     2.3108267   -2.4370568   -0.1492510
	C     5.0698177    0.6126831   -0.7140876
	H     1.9134758   -2.2134053   -2.2624267
	H     2.9957102   -2.8609007    1.8708624
	H     6.8647676    0.3970375    0.4725288
	C     2.5912894   -2.6233298   -1.5105191
	C     3.1892928   -2.9729186    0.8033170
	H     5.5868648    0.4979821   -1.6679094
	C     3.7253951   -3.3253232   -1.9144325
	C     4.3328686   -3.6650056    0.3975521
	H     3.9261437   -3.4656270   -2.9778593
	H     5.0130722   -4.0682375    1.1496283
	C     4.6030751   -3.8443649   -0.9594408
	H     5.4939874   -4.3904368   -1.2732091
	O    -0.9132205    0.8654868   -2.6709122
	C     0.3095768    0.7863575   -3.1225920
	H     0.5452669    0.0469906   -3.9019490
	O     1.0473226    1.9109786   -3.3266879
	H     0.7127012    2.6178760   -2.7307738
	\end{verbatim}

	\subsection{Stable intermediate IX}
	\begin{verbatim}
	104
	
	H    -6.7995766   -0.3501329   -1.2158299
	H    -5.1753779   -0.9590749   -3.0073767
	C    -5.7494307   -0.0685308   -1.1249154
	C    -4.8414893   -0.4064133   -2.1277195
	H    -6.0099582    0.9370106    0.7720970
	C    -5.3053977    0.6476996   -0.0097862
	C    -3.4980684   -0.0377207   -2.0157358
	H    -2.8205216   -0.3190427   -2.8220142
	H    -3.2315181    5.5927273   -1.9725616
	H    -2.4924030    6.8084725    0.0777993
	H    -4.3211783   -5.4897958   -0.9008141
	H    -4.5905028   -3.3058516    0.2709790
	C    -3.9658356    1.0187301    0.0969376
	C    -2.6596967    5.0727789   -1.2020891
	C    -2.2471292    5.7533495   -0.0526568
	C    -3.0375664    0.6805307   -0.9029689
	H    -2.7130496    3.2024809   -2.2540851
	C    -3.4898211   -4.7843036   -0.8565248
	C    -2.3547852    3.7227087   -1.3636205
	C    -3.6390470   -3.5637666   -0.1975612
	H    -3.6478379    1.6145266    0.9537025
	C    -1.5240554    5.0741083    0.9266362
	H    -1.1936406    5.5946199    1.8269241
	H    -2.1420149   -6.0571500   -1.9711422
	C    -2.2684227   -5.1030952   -1.4567385
	C    -1.6182261    3.0276471   -0.3860160
	H    -1.4574439    1.1865906   -3.1168537
	C    -1.2101190    3.7221361    0.7568958
	C    -2.5737994   -2.6606659   -0.1356760
	H    -2.7017218   -1.7083980    0.3780825
	H    -0.9941388   -1.3447239   -2.8400148
	H    -0.6215437    3.2283114    1.5305027
	H    -0.2280357    0.0395821   -4.8662646
	P    -1.2756409    1.2251309   -0.6716396
	C    -1.2005443   -4.2086014   -1.3919092
	C    -0.6128586    1.2766046   -2.4182513
	C    -1.3442298   -2.9760656   -0.7291134
	H    -0.2464680   -4.4802099   -1.8464876
	H    -0.2166420    2.2928316   -2.5572826
	H     0.9118141    1.3848621   -4.6301559
	C     0.0392206   -1.1980355   -2.4947567
	C     0.6796760    0.3524804   -4.3289267
	H     0.6452415   -1.9024327   -3.0837550
	C     0.4805290    0.2535758   -2.8016234
	H     1.5082061   -0.2906873   -4.6610048
	P     0.0627704   -1.7810041   -0.7282679
	Ru    0.0930258    0.0068958    0.7993430
	H     2.4742873   -2.4438087   -2.4802423
	H     2.9855487    4.7401895    2.5258626
	H     2.6625124    2.3446909    2.0250055
	C     2.6682402    4.4368614    1.5266525
	C     2.4769153    3.0858398    1.2434949
	C     1.8347820    0.5919089   -2.1394964
	C     2.4861307   -3.0588992   -1.5798602
	H     2.2696554    1.4842751   -2.6149850
	C     1.4863220   -2.9369876   -0.6046473
	H     2.6290528    6.4549888    0.7479404
	C     2.4731109    5.3975275    0.5292998
	H     2.5495402   -0.2221114   -2.3300301
	P     1.8718447    0.8499631   -0.2966950
	C     2.0839300    2.6675862   -0.0408503
	H     4.2892360   -4.0661169   -2.1980894
	C     3.5229412   -3.9828688   -1.4259395
	C     2.0903948    4.9949559   -0.7493951
	C     1.8975109    3.6394096   -1.0325173
	C     1.5496737   -3.7572998    0.5369440
	H     1.9421909    5.7356658   -1.5368044
	H     0.7886846   -3.6667854    1.3108696
	H     1.6076595    3.3583150   -2.0455412
	C     3.5741800   -4.7951918   -0.2936153
	C     3.5876016    0.2876868    0.0868466
	C     2.5849085   -4.6787563    0.6869250
	H     4.4825058    1.9480644   -0.9984736
	H     3.0102813   -1.4404467    1.2339408
	H     4.3807546   -5.5206835   -0.1761200
	C     4.6637632    1.0289478   -0.4376464
	C     3.8482108   -0.8730074    0.8261054
	H     2.6168757   -5.3127650    1.5745904
	C     5.9757214    0.6073355   -0.2287307
	C     5.1677117   -1.2833586    1.0395582
	H     6.8018130    1.1906840   -0.6388095
	H     5.3600896   -2.1849733    1.6237376
	C     6.2309044   -0.5497987    0.5136163
	H     7.2594827   -0.8719464    0.6849303
	O     1.0670624   -1.2835636    2.0778444
	C     1.5871606   -0.8290624    3.2992178
	H     1.9242493   -1.6992758    3.8944679
	H     2.4705860   -0.1665927    3.1772942
	H    -0.1490045    1.4846206    1.6989277
	H     0.4885841    1.0944823    2.1179673
	H     0.8475740   -0.2803425    3.9228606
	O    -1.7466699   -0.6275675    2.2837152
	C    -1.7317061   -1.8827200    3.0229225
	C    -2.2224936    0.3401404    3.2524878
	C    -3.0380220   -1.8690967    3.8225342
	H    -1.6541370   -2.6937791    2.2919727
	H    -0.8446431   -1.9038831    3.6731363
	C    -3.3722768   -0.3612382    3.9797136
	H    -1.3999265    0.5896606    3.9475022
	H    -2.5124323    1.2461145    2.7073037
	H    -2.9148452   -2.3747266    4.7888091
	H    -3.8360056   -2.3876118    3.2749787
	H    -3.4298618   -0.0474345    5.0299026
	H    -4.3342064   -0.1240359    3.5066121
	\end{verbatim}
	
	\subsection{Stable intermediate XVIII}
	\begin{verbatim}
	90
	
	H    -2.5470040    6.3737910   -0.0663331
	H    -2.8123688    5.2837003    2.1590310
	C    -2.3575811    5.3050645    0.0461674
	C    -2.5053953    4.6951169    1.2928448
	H    -1.8557613    5.0087803   -2.0352058
	C    -1.9710474    4.5402664   -1.0564995
	C    -2.2738967    3.3255410    1.4375671
	H    -2.4225212    2.8750703    2.4198606
	H    -6.4658186   -0.4517580    1.3317265
	H    -7.1063609   -0.4963390   -1.0755395
	H     2.6557529    6.2942052   -1.4687730
	H     2.6184172    5.8584372    0.9834332
	C    -1.7294505    3.1746660   -0.9123259
	C    -5.7147847   -0.2699971    0.5610415
	C    -6.0732010   -0.2947159   -0.7877950
	C    -1.8803888    2.5510336    0.3364899
	H    -4.1455477    0.0191627    1.9948402
	C     2.3956306    5.3032142   -1.0935693
	C    -4.3952745   -0.0058546    0.9334925
	C     2.3751741    5.0591297    0.2813670
	H    -1.4301507    2.5835298   -1.7811090
	C    -5.1048465   -0.0547732   -1.7657306
	H    -5.3774960   -0.0715606   -2.8223592
	H     2.1024472    4.4592222   -3.0617405
	C     2.0867078    4.2748854   -1.9863009
	C    -3.4181977    0.2319862   -0.0454710
	H    -2.0467422    0.9808698    2.7867311
	C    -3.7839051    0.2069953   -1.4010590
	C     2.0524934    3.7902287    0.7651273
	H     2.0614547    3.6161788    1.8424160
	H     0.2799551    2.1595239    2.2435168
	H    -3.0318347    0.3759459   -2.1729904
	H    -0.4329712    1.0976984    4.5735037
	P    -1.6937413    0.7161533    0.3812452
	C     1.7531024    3.0089158   -1.5045840
	C    -1.5279993    0.2401619    2.1618835
	C     1.7341262    2.7546438   -0.1254665
	H     1.5072438    2.2099149   -2.2089384
	H    -2.0783239   -0.7063544    2.2815788
	H    -0.8311301   -0.6372895    4.6187758
	C     0.8360055    1.2104261    2.1943488
	C    -0.1214917    0.1031833    4.2204333
	H     1.6924210    1.3170373    2.8744149
	C    -0.0796827    0.0562400    2.6813970
	H     0.8692688   -0.1126510    4.6478115
	P     1.4151070    1.0356322    0.4486581
	Ru   -0.0120573   -0.0641878   -0.9092097
	H     3.3997536    0.3976754    2.7128807
	H    -3.9488523   -3.7342434   -1.3914406
	H    -2.0900837   -2.0952718   -1.3008549
	C    -3.2197666   -3.7905063   -0.5815364
	C    -2.1728248   -2.8665122   -0.5318431
	C     0.4819087   -1.3320907    2.2709993
	C     3.8678869    0.2287837    1.7421868
	H     0.0202657   -2.1082765    2.8983765
	C     3.1569307    0.4221504    0.5485563
	H    -4.1575089   -5.4876614    0.3702908
	C    -3.3363978   -4.7698197    0.4048525
	H     1.5622408   -1.3565187    2.4750974
	P     0.2047409   -1.7607250    0.4927662
	C    -1.2317972   -2.9166010    0.5050001
	H     5.7458365   -0.3063248    2.6569492
	C     5.2073851   -0.1709298    1.7172452
	C    -2.3945911   -4.8347603    1.4358905
	C    -1.3456095   -3.9175421    1.4847999
	C     3.8216208    0.2062141   -0.6715346
	H    -2.4733336   -5.6066754    2.2032278
	H     3.2927936    0.3596796   -1.6134471
	H    -0.6056583   -4.0041204    2.2830725
	C     5.8566025   -0.3811118    0.5004738
	C     1.5535622   -2.9109735   -0.0092705
	C     5.1576602   -0.1900365   -0.6944923
	H     0.3720305   -3.5806934   -1.7031940
	H     2.9739559   -2.5290542    1.5906045
	H     6.9048341   -0.6835812    0.4825197
	C     1.3385830   -3.6429203   -1.1953357
	C     2.7672232   -3.0758877    0.6711517
	H     5.6572457   -0.3448820   -1.6522310
	C     2.3197484   -4.5014654   -1.6923373
	C     3.7421841   -3.9473198    0.1775979
	H     2.1314131   -5.0683171   -2.6056658
	H     4.6776283   -4.0699228    0.7260186
	C     3.5274686   -4.6540340   -1.0055431
	H     4.2921683   -5.3327582   -1.3861909
	O    -0.9205047   -0.1437941   -2.7216529
	C     0.0671504   -0.6760065   -3.4936715
	H    -0.1305676   -1.7179961   -3.8145929
	O     1.2946578   -0.6984889   -2.6269831
	H     1.5840222   -1.6252772   -2.4886499
	H     0.3195803   -0.0480147   -4.3664138
	\end{verbatim}
	
	\newpage
	\bibliographystyle{apsrev4-2}

	\bibliography{begin,Refs_QComp_sup,end}